\documentclass[11pt]{article}

\usepackage{physics}
\usepackage[utf8]{inputenc}
\usepackage{epsfig}
\usepackage{graphicx}
\usepackage{subfig}
\usepackage{amsmath}
\usepackage{amssymb}
\usepackage{amsthm}
\usepackage[colorlinks=true, allcolors=black]{hyperref}
\usepackage{float}
\usepackage[left=1.5cm,right=1.5cm,top=1.3cm,bottom=1.35cm]{geometry}
\usepackage[skip=1pt plus1pt, indent=10pt]{parskip}
\usepackage[scr]{rsfso}
\usepackage{mathtools}
\usepackage{stix2}
\usepackage{tensor}
\usepackage{dsfont}

\theoremstyle{plain}
\newtheorem{tm}{Theorem}[section]%[subsection]
\newtheorem{prop}[tm]{Proposition}
\newtheorem{lem}[tm]{Lemma}
\newtheorem{cor}[tm]{Corollary}
\newtheorem{conj}[tm]{Conjecture}
\theoremstyle{definition}
\newtheorem{defs}[tm]{Definitions}
\newtheorem{defi}[tm]{Definition}
\newtheorem{pro}[tm]{Property}
\newtheorem{cond}[tm]{Condition}

\theoremstyle{remark}
\newtheorem{rem}[tm]{Remark}
\newtheorem{nota}[tm]{Notation}
\newcommand{\R}{\mathbb{R}}
\newcommand{\Z}{\mathbb{Z}}
\newcommand{\N}{\mathbb{Z}_{+}}
\newcommand{\Nz}{\mathbb{Z}_{\geq0}}
\newcommand{\Hil}{\mathscr{H}}
\newcommand{\I}{\operatorname{I}}
\newcommand{\leb}{\operatorname{L}}
\newcommand{\sob}{\operatorname{H}}
\newcommand{\OO}{\mathcal{O}}

\newcommand{\id}{\mathcal{I}}

\newcommand{\D}{\mathcal{D}}
\newcommand{\Po}{\mathcal{P}}
\newcommand{\Qo}{\mathcal{Q}}

\newcommand{\T}{\mathcal{T}}
\newcommand{\Ro}{\mathcal{R}}
\newcommand{\Ri}{\mathscr{R}}

\newcommand{\So}{\mathcal{S}}

\newcommand{\Ko}{\mathcal{K}}
\newcommand{\Ki}{\mathscr{K}}
\newcommand{\Ci}{\mathscr{C}}
\newcommand{\Ha}{\mathcal{H}}

\newcommand{\F}{\operatorname{F}}
\newcommand{\tra}{\operatorname{tr}}
\newcommand{\de}{\operatorname{det}}
\newcommand{\ex}{\operatorname{exp}}
\newcommand{\lo}{\operatorname{log}}

\newcommand{\ipr}[2]{\left\langle#1,#2\right\rangle}

\newcommand{\Ai}{\operatorname{Ai}}

\newcommand{\ind}[1]{\mathds{1}_{#1}}
\newcommand{\compo}[2]{\left(#1\right)_{#2}}
\newcommand{\comps}[3]{\left(#1\right)_{#2,#3}}
\newcommand{\eref}[1]{~\eqref{#1}}
\newcommand{\sref}[1]{$\mathsection$\ref{#1}}

\newcommand{\af}[2]{\chi_{#1,#2}}
\newcommand{\Af}{\mathcal{X}}
\newcommand{\au}[2]{\mu_{#1,#2}}
\newcommand{\an}[2]{\nu_{#1,#2}}
\newcommand{\gao}{\bigg(\frac{\textbf{q}}{\psi}-1\bigg)}
\newcommand{\ga}[1]{\bigg(\frac{\textbf{q}}{\psi}-1\bigg)^{#1}}

\newcommand{\fa}[1]{\,\eta_{#1}}
\newcommand{\dl}[2][\I]{\big(\partial_#2\Pi_{#1}\big)}
\newcommand{\tj}{\tau_j}
\newcommand{\nth}{n^{\operatorname{th}}}
\newcommand{\ipro}[2]{\theta_{#1,#2}}
\newcommand{\qcv}{\textbf{q}}
\newcommand{\pcv}{\textbf{p}}
\newcommand{\pbn}[3][n]{\{#2,#3\}_{#1}}
\newcommand{\B}{\operatorname{B}}
\newcommand{\A}{\operatorname{A}}
\newcommand{\Ns}[2]{\operatorname{N}_{#1,#2}}
\newcommand{\ud}{\dot{u}_0}
\newcommand{\udd}{\ddot{u}_0}
\newcommand{\udf}{\frac{\dot{u}_0}{\gamma}}
\newcommand{\cdj}[1]{\frac{\A_{#1}}{\xi-\tau_{#1}}}
\newcommand{\cdx}{\left(\B_\xi+\sum_{i=1}^{2m}\frac{\A_{i}}{\xi-\tau_{i}}\right)}

\numberwithin{equation}{subsection}
\allowdisplaybreaks

\title{Fredholm Determinants from Schrödinger Type Equations,\\ and Deformation of Tracy-Widom Distribution}

\author{Taro Kimura and Xavier Navand%
\footnote{Institut de Mathématiques de Bourgogne, Université de Bourgogne, CNRS, UMR 5584, France}
}

\date{}

\begin{document}

\maketitle

\thispagestyle{empty}

\begin{abstract}
We undertake an analysis of Fredholm determinants arising from kernels whose defining functions satisfy a Schrödinger type equation. When this defining function is the Airy one, the evaluation of the corresponding Fredholm determinant yields the notorious Tracy-Widom distribution \cite{4}, which has found many applications in numerous domains. In this paper, we unveil a generalization of the Tracy-Widom distribution for a generic class of defining functions. Furthermore, we bring forth a direct application of our upshot and survey the relation between the framework which we employ and isomonodromic systems.
\end{abstract}

\tableofcontents

\setcounter{page}{0}
%\newpage

\section{Introduction and Summary}\label{s1}

In this section we begin with a brief introduction of the quantities which will have a central role in our discussion. We then state a few contexts in which these quantities arise and are relevant to compute. After what we give a summary of our main results and discuss possible directions for future investigations, followed by a presentation of the plan for the remaining sections which are devoted to the derivation of these results.

\subsection{Setup}%{General Considerations}
\label{s11}

Given a generic real-valued potential $v(x;\xi)$, our starting point is the Schrödinger type equation,
\begin{flalign}\label{1.1.1}
\begin{split}
\Big(\partial_x^2-v(x;\xi)\Big)\varphi_\xi(x)=\xi\varphi_\xi(x)\,.
\end{split}
\end{flalign}
Once we are provided with a solution $\varphi_\xi(x)$, which we refer to as the wave function, we can define the main quantities of interest for this text.

\begin{defs}\label{t1.1.1}
Denoting by $\leb^2(\operatorname{J})\coloneqq \leb^2(\operatorname{J};\R)$ the Hilbert space of real-valued square integrable functions on $\operatorname{J}$ and by $\D(\OO)\subseteq\leb^2(\R)$ the domain of any operator $\OO$, the integral operator
\begin{flalign}\label{1.1.2}
\begin{split}
\Ko_0:\D(\Ko_0)&\longrightarrow\leb^2(\R) \\
f&\longmapsto\int_\R\Ki(\cdot,\zeta)f(\zeta)\,d\zeta
\end{split}
\end{flalign}
is defined by the kernel
\begin{flalign}\label{1.1.3}
\begin{split}
\Ki(\xi,\zeta)\coloneqq \int_{\R_+}\varphi_\xi(x)\varphi_\zeta(x)\,dx\,.
\end{split}
\end{flalign}
And, given any interval $\I\subset\R$, we define the corresponding projector as
\begin{flalign}\label{1.1.4}
\begin{split}
\Pi_{\I}:\leb^2(\R)&\longrightarrow\leb^2(\R) \\
f&\longmapsto\ind{\I}(\cdot)f\,,
\end{split}
\end{flalign}
multiplication by the indicator function $\ind{\I}$. Then the Fredholm determinant is given by
\begin{flalign}\label{1.1.5}
\begin{split}
\F(\I)\coloneqq \de(\id-\Ko)_{\leb^2(\R)}\,,
\end{split}
\end{flalign}
where $\id$ is the identity operator and $\Ko\coloneqq \Ko_0\Pi_{\I}$, hence we have explicitly the series expansion,
\begin{flalign}
    \F(\I) = 1 + \sum_{k=1}^\infty \frac{(-1)^k}{k!} \int_{\I} d\xi_1 \cdots \int_{\I} d\xi_k \det_{1 \le i, j \le k} \Ki(\xi_i,\xi_j)\,.
\end{flalign}
\end{defs}

\begin{cond}\label{t1.1.2}
All along we assume that $\varphi_\xi(x)$ decreases sufficiently fast as $\xi\rightarrow\infty$, for any $x$, in order to ensure that $\Ko$ is bounded, so that the Fredholm determinant is well-defined. Furthermore, denoting by $\sob^n(\operatorname{J})$ the Sobolev spaces associated to $\leb^2(\operatorname{J})$, and by $\D(\Qo)$ the domain of multiplication by $\xi$, see Definition \ref{t2.2.7}, there are two precise $\xi$ regularity conditions that $\psi(\xi)\coloneqq \varphi_\xi(0)$ must satisfy, namely $\psi\in\sob^2(\R)$ and $\psi\in\D(\Qo)$.
\end{cond}

The incoming discussions are mostly devoted to the evaluation of these quantities under the following conditions.

\begin{cond}\label{t1.1.3}
We demand that the wave function is a real-valued square integrable function on $\R_+$ for any $\xi\in\R$, by which we mean
\begin{flalign}\label{1.1.7}
\begin{split}
\varphi_\xi\in\leb^2(\R_+)\,,\hspace{0.4cm}\forall \xi\in\R\,.
\end{split}
\end{flalign}
One may observe that this ensures $\xi\varphi_\xi\in\leb^2(\R_+)\,,\,\,\forall \xi\in\R$.
\end{cond}

\begin{cond}\label{t1.1.4}
We assume that there exist some functions $u$ and $\phi$ such that
\begin{flalign}\label{1.1.8}
\begin{split}
\varphi_\xi(x)=\phi\big(u(x)+\gamma \xi\big)\,,
\end{split}
\end{flalign}
with $\gamma\in\R\char`\\\{0\}$. The only restriction on $u$ and $\phi$ is that they must be twice differentiable, i.e. $u,\phi\in\operatorname{C}^2(\R)$, but these functions must still ensure that $\varphi_\xi$ obeys Condition \ref{t1.1.3}.
\end{cond}

\begin{cond}\label{t1.1.5}
We require that the function $u$ has a non-vanishing first derivative at $x=0$,
\begin{flalign}\label{1.1.6}
\begin{split}
\dot{u}_0\coloneqq \frac{d}{dx}\Big|_{x=0}u(x)\neq0\,.
\end{split}
\end{flalign}
\end{cond}

\begin{cond}\label{t1.1.6}
The potential evaluated at $x=0$ should not depend on $\xi$, namely $\partial_\xi v(0;\xi)=0$.
\end{cond}

\begin{rem}\label{t1.1.7}
In appendix \sref{sA}, we provide three generic cases in which Condition \ref{t1.1.4} is satisfied.
\end{rem}

\begin{rem}\label{t1.1.8}
In appendix \sref{sB}, we lay out a precise application of our endeavors, namely to the so-called finite temperature kernels introduced in the context of the Kardar-Parisi-Zhang (KPZ) equation \cite{28,29}. See also~\cite{Calabrese:2010EPL,6}.
\end{rem}

\subsection{Motivations}\label{s12}

We first describe how some Fredholm determinants have been expressed as functionals depending on an auxiliary function which is fixed by a boundary value problem, stating several established results and a few applications of theirs. After this, we present the generalization that we undertake in the present paper and also mention a different approach which enables one to assess the asymptotics of a broad class of Fredholm determinants. Lastly, we emphasize on an example of a specific domain in which generic methods for the computation of Fredholm determinants could find applications.

\paragraph{Functional Formulae of Fredholm Determinants.} Fredholm determinants (defined as in Definitions \ref{t1.1.1}) provide insights in various contexts ranging from integrable systems through their relation with isomonodromic $\tau$ functions \cite{1,3}, to hierarchies of differential equations (such as Painlevé transcendents \cite{4,5} or the KPZ equation \cite{28,29,Calabrese:2010EPL,6}), to combinatorics \cite{30,7} and many more. More generally, Fredholm determinants are relevant whenever determinantal point processes \cite{9} occur and have thereby gathered more and more attention over the last decades. See \cite{10} for a more in-depth list of domains in which Fredholm determinants related to determinantal processes arise.

In particular, Fredholm determinant methods unveiled important results in eigenvalues statistics of Random Matrix Theories (RMTs) \cite{15,11,13,14,24,12}, as a notorious instance of such a result, one may mention the Tracy-Widom distribution \cite{4} which has found various statistical applications.

Hence, having convened all of these interests, many evaluations of Fredholm determinants were carried out under the condition that there exists a function $\phi$ such that $\varphi_\xi(x)=\phi(\xi+x)$ satisfying several regularity conditions \cite{4,16,17,18,21,19,20,2}. Typically, for a generic differential equation satisfied by $\varphi_\xi(x)$, the associated Fredholm determinant can be generically expressed in terms of some solution to a given (integro-)differential equation, whose asymptotics are fixed by $\varphi_\xi(0)$. Let us illustrate this recalling the classical results of \cite{4}. For the Tracy-Widom distribution, we have $\varphi_\xi(x)=\Ai(\xi+x)$ with $\Ai(\cdot)$ denoting the Airy function, a solution to\eref{1.1.1} for $v(x;\xi)=x$, which leads to the Airy kernel,
\begin{align}
    \Ki_\text{Ai}(\xi,\zeta) = \int_{\mathbb{R}_+} \operatorname{Ai}(\xi +x) \operatorname{Ai}(\zeta +x) \, dx \, .
\end{align}
The associated Fredholm determinant is then expressed in terms of a solution to the Painlevé II equation asymptotic to the Airy function, namely the Hastings-McLeod solution \cite{31}. Besides, it has been clarified that a similar formalism is applicable also to the higher-order analog of the Tracy-Widom distribution~\cite{19,2}.

In this paper, considering\eref{1.1.1} as our generic differential equation for $\varphi_\xi(x)$, we relax the condition $\varphi_\xi(x)=\phi(\xi+x)$ and instead impose the less restrictive Condition \ref{t1.1.4}. Under this Condition as well as the regularity Conditions \ref{t1.1.2}, \ref{t1.1.3}, \ref{t1.1.5}, \ref{t1.1.6} and \ref{t3.1.6}, we still manage to obtain the functional formula of the Fredholm determinant in terms of an auxiliary function, whose asymptotics are governed by $\varphi_\xi(0)$, satisfying a second order integro-differential equation which we also explicitly derive. Moreover, the intermediate steps to do so involve not one but an infinite set of auxiliary functions, whose dynamics are expressible in terms of infinite size matrices which are formally related to an isomonodromic system through Schlesinger equations. And we motivate this formal relation by proving that some components of these matrices indeed satisfy a system of Schlesinger equations.

\paragraph{Riemann-Hilbert Problem (RHP) Characterization of Fredholm Determinants.} Besides, the asymptotic behaviour of numerous Fredholm determinants have been characterized with RHPs. This alternative method was first developed in the early 1990s too by \cite{34}, and this approach has been frequently appearing in recent works, including the four following examples in which some of the kernels considered also belong to the scope of the present paper. A specific finite temperature (see below) extension of the higher order Airy kernels is explored in \cite{36}, whilst \cite{37} focuses on a large class of finite temperature deformations of the Airy kernel. Besides, integrable kernels of the form \eref{1.1.3} (and their asymmetric version, i.e. replacing one the $\varphi$ with another function) with additive, $\varphi_\xi(x)=\phi(x+\xi)$, or multiplicative, $\varphi_\xi(x)=\phi(x\xi)$, compositions are considered in \cite{35}. And \cite{38} investigates a finite temperature version of the Bessel kernel.

The RHP process leads to a functional formula for the Fredholm determinant which is characterized by an operator-valued RHP and $\varphi_\xi$ does \emph{not} need to satisfy a differential equation, like\eref{1.1.1} for us. The payoff for this procedure is that the closed scalar boundary value problem completely determining the function for which the functional evaluates to the Fredholm determinant seems to be out of reach. Indeed, in order to establish this integro-differential equation, we will extensively use the closure relation of Proposition \ref{t2.2.36} which is essentially a consequence of\eref{1.1.1}.

Additionally, allow us to point out that a RHP characterization was carried out in \cite{40} for a kernel appearing as a universal limit in RMT which coincides with a kernel arising in the construction developed by \cite{39} for conditional thinned ensembles of point processes. And it turns out that an equation of the type\eref{1.1.1} is directly available for this kernel \cite[eq.(2.11)]{40}.

\paragraph{Eigenvalue Statistics and Two-Dimensional Quantum Gravity.} Our original motivation towards Fredholm determinants arising from a generic differential equation of the form\eref{1.1.1}, which does not necessarily admit solutions that may be written as $\varphi_\xi(x)=\phi(\xi+x)$, comes from the dual description of two-dimensional quantum gravity as a double scaled matrix integral \cite{15,23,22,25}. In this framework, the double scaling limit of random matrices maps the difference equation of the orthogonal polynomials to\eref{1.1.1}, entailing that the wave function $\varphi_\xi$ is provided with the interpretation of the \textit{double scaled orthogonal polynomial}, and along these lines, the continuous variable $x \in \mathbb{R}_+$ is to be regarded as the scaling of the discrete variable labeling the polynomials: The Christoffel-Darboux kernel of degree $N$,
\begin{align}
    \Ki_N(\xi,\zeta) = \sum_{n=0}^{N-1} f_n(\xi) f_n(\zeta) \, ,
\end{align}
is asymptotic to the kernel \eqref{1.1.3} in the double scaling limit, where $\{f_n\}_{n \in \mathbb{Z}_{\ge0}}$ is a family of orthonormal functions, $\left< f_n, f_m \right> = \delta_{n,m}$ (see Lemma~\ref{t2.1.1} for the definition of the inner product), constructed from the orthogonal polynomials arising in RMT.

According to these considerations, numerical computations of the related Fredholm determinant have been proven to yield insightful results on the non-perturbative behaviour of various models \cite{26,27}. Namely, it enables one to notice signs of an underlying discrete spectrum, which is highly relevant in the quantum gravity side of the duality. See, e.g., \cite{22} for a recent review. Whence the present paper can be deemed as a step towards an analytic derivation of these results, also bringing forth generic results in the hope that they can be fruitful in other domains such as the ones mentioned in the previous paragraphs.

\subsection{Results}\label{s13}

\paragraph{Functional formula of the Fredholm Determinant.} Our main result is the following one, all the quantities involved are defined in Definitions \ref{t1.1.1} and Notation \ref{t2.1.6}, apart from $u$ and $\gamma$ which are defined through Condition \ref{t1.1.4}.

\begin{tm}[Theorem \ref{t4.0.1}]\label{t1.3.1}
If Condition \ref{t1.1.4} as well as Conditions \ref{t1.1.2}, \ref{t1.1.3}, 
\ref{t1.1.5}, \ref{t1.1.6} and \ref{t3.1.6} are satisfied, then the Fredholm determinant for the interval $\I=[\tau,\infty)$ is explicitly given only in terms of $\psi(\xi)\coloneqq\varphi_\xi(0)$ and $\qcv(\xi)\coloneqq\big(\id-\Ko\big)^{-1}\psi(\xi)$ as follows,
\begin{flalign}\label{1.3.1}
\begin{split}
\F\Big([\tau,\infty)\Big)=\ex\left(\rule{0cm}{0.6444cm}\int_\tau^\infty\Bigg[\qcv_\sigma\Bigg(\udf\qcv''_\sigma-\qcv^3_\sigma+\frac{\udd}{\ud}\qcv_\sigma\bigg(\frac{\qcv'_\sigma}{\psi_\sigma}-\frac{\psi'_\sigma\qcv_\sigma}{\psi^2_\sigma}\bigg)\Bigg)-\udf\qcv'^2_\sigma\Bigg]d\sigma\right)\,,
\end{split}
\end{flalign}
where $f_\sigma\coloneqq f(\sigma)$, $f'_\sigma\coloneqq\frac{df}{d\sigma}(\sigma)$ and $f''_\sigma\coloneqq\frac{d^2f}{d\sigma^2}(\sigma)$ with $f$ representing either $\qcv$ or $\psi$. Moreover the auxiliary function $\qcv$ satisfies the following second order integro-differential equation,
\begin{flalign}\label{1.3.2}
\begin{split}
\qcv''_\tau=\frac{\gamma^2}{\ud^2}\big(v_0+\tau\big)\qcv_\tau+\frac{2\gamma}{\ud}\Bigg(\qcv^3_\tau-\frac{\gamma\udd}{\ud^2}\qcv_\tau\int_\tau^\infty\qcv^2_\sigma\,d\sigma\Bigg)-\frac{2\gamma^2\udd^2}{\ud^4}\Bigg(\frac{\qcv^3_\tau}{\psi^2_\tau}-\frac{\qcv^2_\tau}{\psi_\tau}\Bigg)+\frac{\gamma\udd}{\ud^2}\Bigg(\qcv'_\tau+2\frac{\qcv^2_\tau}{\psi^2_\tau}\psi'_\tau-4\frac{\qcv_\tau}{\psi_\tau}\qcv'_\tau\Bigg)\,.
\end{split}
\end{flalign}
Whenceforth, provided with $\psi$, $u$ and $\gamma$, the Fredholm determinant $\F\Big([\tau,\infty)\Big)$ is completely determined by the solution of\eref{1.3.2} whose behaviour as $\tau\rightarrow\infty$ coincides with $\psi$, more precisely $\qcv_\tau\xrightarrow{\tau\rightarrow\infty}\psi_\tau$.

Alternatively, the Fredholm determinant also admits the following explicit representation,
\begin{flalign}\label{1.3.3}
\begin{split}
\F\Big([\tau,\infty)\Big)=\ex\left(\rule{0cm}{0.65cm}-\frac{\gamma}{\ud}\int_\tau^\infty\big(\sigma-\tau\big)\left(\rule{0cm}{0.6444cm}\qcv^2_\sigma+\frac{\udd}{\gamma}\bigg(\frac{2\,\qcv_\sigma}{\psi_\sigma}-1\bigg)\Bigg[\qcv'^2_\sigma-\qcv_\sigma\Bigg(\qcv''_\sigma-\frac{\gamma}{\ud}\qcv^3_\sigma+\frac{\gamma\udd}{\ud^2}\qcv_\sigma\bigg(\frac{\qcv'_\sigma}{\psi_\sigma}-\frac{\psi'_\sigma\qcv_\sigma}{\psi^2_\sigma}\bigg)\Bigg)\Bigg]\right)d\sigma\right)\,.
\end{split}
\end{flalign}
\end{tm}

Section \sref{s4} is dedicated to the proof of this Theorem, which will require several results from sections \sref{s2} and \sref{s3}.

Notice that, setting $v(x;\xi)=x$ in\eref{1.1.1} and considering the solution $\varphi_\xi(x)=\Ai(x+\xi)$ we mentioned in \sref{s12}, we have $\ud=1=\gamma$ and $\udd=0$, so that\eref{1.3.3} reduces to
\begin{flalign}\label{1.3.4}
\begin{split}
\F\Big([\tau,\infty)\Big)=\ex\left(\rule{0cm}{0.6444cm}-\int_\tau^\infty\big(\sigma-\tau\big)\,\qcv^2_\sigma\,d\sigma\right)\,,
\end{split}
\end{flalign}
whilst\eref{1.3.2} becomes
\begin{flalign}\label{1.3.5}
\begin{split}
\qcv''_\tau=\tau\qcv_\tau+2\qcv^3_\tau\,,
\end{split}
\end{flalign}
which indeed agrees with \cite{4} since $v_0=0$ for $v(x;\xi)=x$. Hence our main result nicely reduces to the notorious Tracy-Widom distribution as one would expect for $\varphi_\xi=\Ai(\cdot+\xi)$.

Furthermore, observe that, for any $l\in\N$, if $u(x)=u_0+\ud x+\sum_{k=3}^lu_{0,k}x^k$, then
\begin{flalign}\label{1.3.6}
\begin{split}
&\F\Big([\tau,\infty)\Big)=\ex\left(\rule{0cm}{0.6444cm}-\frac{\gamma}{\ud}\int_\tau^\infty\big(\sigma-\tau\big)\,\qcv^2_\sigma\,d\sigma\right)\,, \\
&\qcv''_\tau=\frac{\gamma^2}{\ud^2}\big(v_0+\tau\big)\qcv_\tau+\frac{2\gamma}{\ud}\qcv^3_\tau\,.
\end{split}
\end{flalign}
Actually, this considerable simplification holds for any $u$ such that $\udd=0$, not necessarily polynomial. Therefore, for any $\gamma$, as long as $\udd=0$, it turns out that our expression for $\F\big([\tau,\infty)\big)$ coincides with the Tracy-Widom one up to a constant factor, and this is even if $u(x)\neq x$. Nonetheless, let us emphasize the fact that, if $(\gamma,\ud,\udd)=(1,1,0)$ but $u(x)\neq x$, then\eref{1.3.4} is indeed satisfied, but the actual value of the Fredholm determinant is not the same as the one for $u(x)=x$ because $\qcv$ then satisfies a differential equation involving $v_0$ and its asymptotic behaviour agrees with $\psi(\xi)\neq\Ai(\xi)$.

On the other hand, when we turn $\udd$ on, our generalization is more involved. But even so, our central result remains a continuous deformation of the Tracy-Widom distribution. Continuous in the sense that\eref{1.3.3} and\eref{1.3.2} respectively tend to\eref{1.3.4} and\eref{1.3.5} as $\big(\gamma,u,v\big)\rightarrow\big(1,x\mapsto x,(x;\xi)\mapsto x\big)$. Note that to comment in a more precise sense the "continuity" of this deformation, one would have to take an interest in the eventual bifurcations of\eref{1.1.1}.

As a last comment on our principal upshot, even if we consider $\big(\gamma,u\big)=\big(1,x\mapsto x\big)$, i.e. additive composition $\varphi_\xi=\phi(\cdot+\xi)$, but $v(x;\xi)\neq x$, then the functional formula is the Tracy-Widom one\eref{1.3.4} but the value of the Fredholm determinant is given by the evaluation of this functional for a different function than the Hastings-McLeod solution to Painlevé II because\eref{1.3.2} then involves $v_0$ and, since $\psi(\xi)\neq\Ai(\xi)$, $\qcv$ is not asymptotic to the Airy function. Also, it is worth emphasizing that our kernel\eref{1.1.3} with additive composition $\big(\gamma,u\big)=\big(1,x\mapsto x\big)$ is considered in \cite{35} with a RHP approach, and the resulting functional formula agrees with\eref{1.3.4} too, see \cite[eq.(17)]{35}.

\paragraph{Relation to Isomonodromic systems.} The second interest on which we focus our endeavors is the relation of the auxiliary functions, involved in the proof of Theorem \ref{t4.0.1}, to an isomonodromic system. To be more precise, in order to obtain a proof for Theorem \ref{t4.0.1}, the \textit{auxiliary wave functions} $\chi_{n,a}$ defined in Definition \ref{t2.2.2} shall be introduced, $n\in\Z_{\geq0}$ and $a\in\{0,1,2\}$. Then we discuss their relation to an eventual isomonodromic system employing the vector notation of Definition \ref{t3.2.1}, i.e. $\Af\coloneqq\begin{pmatrix}\af{0}{0}&\af{0}{1}&\af{1}{0}&\af{1}{1}&\cdots\end{pmatrix}^t$.

Our first step in this direction is to show that the dynamics of the auxiliary wave functions admit a Lax formulation. Namely, Propositions \ref{t3.2.3} and \ref{t3.2.7} explicitly determine two infinite size matrices $\A_j$ and $\B_\xi$ satisfying\eref{3.2.2}, i.e.
\begin{flalign}\label{1.3.7}
\begin{split}
\partial_j\Af(\xi)=-\frac{\A_j}{\xi-\tj}\Af(\xi)\,,\,\,\,\,\,\,\,\,\,\,\partial_\xi\Af(\xi)=\left(\B_\xi+\sum_{j=1}^{2m}\frac{\A_j}{\xi-\tj}\right)\Af(\xi)\,,
\end{split}
\end{flalign}
where the index $j$ is related to the interval $\I$ through\eref{2.2.6}, $\I=\bigsqcup_{j=1}^m[\tau_{2j-1};\tau_{2j}]\,,\,\,\tau_{2j - 1} < \tau_{2j}\,,\,\forall j$ and $\partial_j\coloneqq\partial_{\tj}$. This Fuchsian system can also be written in terms of the corresponding covariant derivatives,
\begin{align}\label{zce}
    D_j \Af(\xi) = 0 \, , \quad 
    D_\xi \Af(\xi) = 0
    \qquad \text{where} \qquad
    D_j = \partial_j + \frac{\A_j}{\xi - \tau_j}
    \, , \quad
    D_\xi = \partial_\xi - \left( \B_\xi + \sum_{j=1}^{2m} \frac{\A_j}{\xi - \tau_j} \right)
    \, .
\end{align}
Then, the compatibility condition of the auxiliary wave functions (also known to be the zero curvature equation), $[D_j,D_\xi] = 0$, implies the following differential system.
\begin{prop}[Proposition \ref{t3.2.16}, Schlesinger Equations]\label{t1.3.3}
Let $i,j\in\{1,...,2m\}\subset\N$ such that $i\neq j$. The infinite size matrices $\A_j$ and $\B_j\coloneqq\B_{\tj}$, that solve\eref{3.2.2} thereby describing the dynamics of the auxiliary wave functions, formally satisfy the following Schlesinger equations,
\begin{flalign}\label{1.3.9}
\begin{split}
\partial_j\A_i=\frac{[\A_i,\A_j]}{\tau_i-\tj}\,,\,\,\,\,\,\,\,\,\,\,\partial_j\A_j=\sum_{\substack{i=1\\i\neq j}}^{2m}\frac{[\A_i,\A_j]}{\tj-\tau_i}-[\A_j,\B_j]\,.
\end{split}
\end{flalign}
\end{prop}
We would like to emphasize again that the matrices $\A_j$ and $\B_j$ are of infinite size, hence the above-mentioned Schlesinger equations make sense only at the formal level.
Then, a considerable part of \sref{s3} is dedicated to the proof of the following result formulated with Notation \ref{t2.2.3} left implicit, i.e. $\alpha,\beta\in\{0,1\}$.

\begin{prop}[Proposition \ref{t3.2.14}]\label{t1.3.2}
Let $i,j\in\{1,...,2m\}\subset\N$ such that $i\neq j$, and $n,p\in\Z_{\geq0}$. The infinite size matrices $\A_j$ and $\B_j\coloneqq\B_{\tj}$, that solve\eref{3.2.2} thereby describing the dynamics of the auxiliary wave functions, satisfy the following set of equations,
\begin{flalign}\label{1.3.8}
\begin{split}
\comps{\partial_j\A_i}{2n+\alpha}{2p+\beta}\big|_{p\geq n}=\comps{\frac{[\A_i,\A_j]}{\tau_i-\tj}}{2n+\alpha}{2p+\beta}\,,\,\,\,\,\,\,\,\,\,\,\comps{\partial_j\A_j}{2n+\alpha}{2p+\beta}\big|_{p\geq n}=\comps{\sum_{\substack{i=1\\i\neq j}}^{2m}\frac{[\A_i,\A_j]}{\tj-\tau_i}-[\A_j,\B_j]}{2n+\alpha}{2p+\beta}\,.
\end{split}
\end{flalign}
\end{prop}

\if0
Once established, this Proposition serves as a reasonable motivation for the following Proposition.

\begin{conj}[Proposition \ref{t3.2.16}, Schlesinger Equations]\label{t1.3.3}
Let $i,j\in\{1,...,2m\}\subset\N$ such that $i\neq j$. The infinite size matrices $\A_j$ and $\B_j\coloneqq\B_{\tj}$, that solve\eref{3.2.2} thereby describing the dynamics of the auxiliary wave functions, satisfy the following Schlesinger equations,
%
\begin{flalign}\label{1.3.9}
\begin{split}
\partial_j\A_i=\frac{[\A_i,\A_j]}{\tau_i-\tj}\,,\,\,\,\,\,\,\,\,\,\,\partial_j\A_j=\sum_{\substack{i=1\\i\neq j}}^{2m}\frac{[\A_i,\A_j]}{\tj-\tau_i}-[\A_j,\B_j]\,.
\end{split}
\end{flalign}
\end{conj}
\fi

One may observe that, since the matrices $\A_j$ and $\B_j$ are completely determined, there is nothing conceptual preventing one from computing\eref{1.3.9} for an arbitrary component. But, as shall be seen in \sref{s3}, even the components considered in\eref{1.3.8} are rather lengthy to equate, and it turns out that the components that remain unverified rapidly become even longer and thereby difficult to handle without a formal calculator. Especially since there is apparently no general expression for an arbitrary component of $\B_j$, Proposition \ref{t3.2.7}, but twelve relations determine the matrix instead, whence many cases would have to be treated in order to formulate a proof for an arbitrary component.

One may further notice that Proposition \ref{t3.2.16} extends the results of \cite{1} for the Schlesinger equations associated to such systems. Nevertheless it is not yet clear how to generalize the trace properties of \cite{1} to our matrices.

\paragraph{Byproducts.} Now we mention two notable auxiliary results that we derive with the methods developed in order to prove Theorem \ref{t4.0.1} and Proposition \ref{t3.2.14}.

In the framework of RMT, the Fredholm determinant yields the \textit{gap probability}, such that no eigenvalue is found in the interval $\I$. Hereby associating $\F\big([s,\infty)\big)$ to the probability that the largest eigenvalue is found in $(-\infty,s)$, which is highly relevant in the context of two-dimensional quantum gravity. And the generalization to $\I=\sqcup_{j=1}^m[\tau_{2j-1};\tau_{2j}]$ would also provide relevant insights, therefore, in sections \sref{s2} and \sref{s3}, the entirety of the upshots which are obtained hold for any interval $\I=\sqcup_{j=1}^m[\tau_{2j-1};\tau_{2j}]$, providing a large basis of results for an eventual future generalization of Theorem \ref{t4.0.1} to a generic $\I$.

As a further observation, we manage to obtain the kernel $\Ki(\xi,\zeta)$ with a Christoffel-Darboux type formula, Proposition \ref{t2.1.8}, which is consistent as a double scaling limit of the discrete case given in \cite{11,14,12,24}.

Besides, still along the lines of double scaled random matrices, the diagonal values of the kernel $\Ki(\xi,\xi)$ bear relevance as they represent the spectral density \cite{15,27}. The following hereby notable upshot is derived in section \sref{s2}.

\begin{prop}[Proposition \ref{t2.1.10}]\label{t1.3.4}
The diagonal values of the kernel are
\begin{flalign}\label{1.3.10}
\begin{split}
\Ki(\xi,\xi)=\frac{\dot{u}_0}{\gamma}\big(\psi'(\xi)\big)^2-\frac{\gamma}{\dot{u}_0}\big(v_0+\xi\big)\big(\psi(\xi)\big)^2+\frac{\ddot{u}_0}{\dot{u}_0}\psi(\xi)\psi'(\xi)\,.
\end{split}
\end{flalign}
\end{prop}

Whilst on the side of isomonodromic systems, it was shown in \cite{19,2} that, considering the higher order Airy kernel, the dynamics of the finitely many analogues of the $\chi_{n,a}$ also admit a Hamiltonian description, which is related to the Lax formalism through trace properties of the analogues of $\A_j$ and $\B_\xi$. As we already mentioned, this relation through trace properties becomes unclear in our case. Nevertheless we manage to establish a set of Hamilton's equations obeyed by the $\chi_{p,a}$, $\forall p\in\{0,...,n-1\}\subset\Z_{\geq0}$, Corollary \ref{t3.1.5}. Thus we end up with an infinite number of Hamiltonians $\Ha_n$, each with an associated Poisson structure $\pbn{\cdot}{\cdot}$.

Whence \sref{s31} is mostly devoted to the evaluation of the Poisson brackets of the Hamiltonians, leading to the notable outcomes presented in Propositions \ref{t3.1.17} and \ref{t3.1.18}. Namely, setting $n,r\in\N$, $n>r$, we have $\pbn{\Ha_r}{\Ha_n}=0$, but the converse is not true in general as we shall show with a counter example
\begin{flalign}\label{1.3.11}
\begin{split}
\pbn[1]{\Ha_1}{\Ha_n}=\frac{n(n-1)}{\psi}\bigg(\frac{\qcv}{\psi}\big(\pcv'\partial_\pcv\psi+\qcv'\partial_\qcv\psi\big)-\qcv'\bigg)\fa{n-2}\Ha_1(\qcv,\pcv)\,,
\end{split}
\end{flalign}
where $\pcv\,(\coloneqq\chi_{0,1})$ and $\eta_n$ are defined in Notation \ref{t3.1.14}.

As a final note on these byproducts, let us emphasize the fact that the Hamiltonian formulations play an important role in the establishment of Theorem \ref{t4.0.1} for at least two reasons. Firstly because the basic idea is to express $\F(\I)$ in terms of $\Ha_1$ utilizing Proposition \ref{t3.1.2}, and then explicit $\Ha_1$ in terms of $\qcv$ only. Secondly because, in order to obtain\eref{1.3.11}, we shall need to prove Lemma \ref{t3.1.9} which in turn plays a crucial role in the proof of Theorem \ref{t4.0.1}, most notably for the derivation of\eref{1.3.2}.

\paragraph{Applications : Finite Temperature Kernels.} There is also an upshot of appendix \sref{sB} concerning the range of applicability of our results which we think is worth pointing out.

\begin{prop}[Proposition \ref{tb1}]\label{t1.3.5}
Let $\phi$ and $u$ be any twice differentiable functions, i.e. $u,\phi\in\operatorname{C}^2(\R)$, and $\gamma\in\R\char`\\\{0\}$. Upon probing the regularity Conditions \ref{t1.1.2}, \ref{t1.1.3}, \ref{t1.1.5}, \ref{t1.1.6} and \ref{t3.1.6} for $\varphi_\xi(x)\coloneqq\phi\big(u(x)+\gamma\xi\big)$ satisfying\eref{1.1.1}, Theorem \ref{t4.0.1} may be applied for any Fredholm determinant arising from a kernel of the following generic form,
\begin{flalign}\label{1.3.12}
\begin{split}
\Ki(\xi,\zeta)=\lim_{a\rightarrow\infty}\int_{u(0)}^{u(a)}\frac{\phi\big(\lambda+\gamma\xi\big)\phi\big(\lambda+\gamma\zeta\big)}{f(\lambda)}\,d\lambda\,,
\end{split}
\end{flalign}
where $f(u)\coloneqq\frac{du}{dx}(x)$ depends solely on $u$ and parameters not related to the integration.
\end{prop}

Such kernels arise in the context of the KPZ equation and are known under the moniker of finite temperature kernels. A specific example in which this upshot is applicable and bears relevance is also given in \sref{sB}. Essentially Proposition \ref{tb1} states that Condition \ref{t1.1.4} relates finite temperature kernels\eref{1.3.12} to kernels of the form\eref{1.1.3}.

\subsection{Outlook}\label{sD}

In the wake of our discussions and derivations, we mention seven yet uncharted directions beyond the ones we undertook which could be interesting to explore. The first two ones are directed towards a universal evaluation of Fredholm determinants arising from kernels of the form $\Ki(\xi,\zeta)=\int_{\R_+}f(x,\xi)\,g(x,\zeta)\,dx\,$, with $f$ and $g$ any two functions satisfying regularity properties and a differential equation. Then the third and fourth ones would characterize further the Fredholm determinants that we considered in this manuscript, whilst the fifth and sixth ones are related to the Lax formalism. Finally the seventh one is more specific and technical, it is more of a comment than a future direction to be explored.

\paragraph{1.} Relax Condition \ref{t1.1.4}. This amounts to the generalization of Theorem \ref{t4.0.1} to any kernel of the form\eref{1.1.3} with $\varphi_\xi$ satisfying the appropriate regularity conditions and a differential equation of the type\eref{1.1.1}. To go down this road, one would have to replace\eref{2.1.8} by a general chain rule $\partial_x\varphi_\xi(x)=\operatorname{ch}(x,\xi)\,\partial_\xi\varphi_\xi(x)$.

\paragraph{2.} Do not impose a Schrödinger-type equation. This amounts to replace\eref{1.1.1} by any differential equation $\So\varphi_\xi(x)=\xi\varphi_\xi(x)$ where $\So$ is any differential operator. A first stride along these lines would be to replace\eref{1.1.1} by $\big(\partial_x^{2p}-v(x;\xi)\big)\varphi_\xi(x)=\xi\varphi_\xi(x)$. Note that, when $\So$ contains odd powers in the derivatives, the operator looses self-adjointness requiring one to introduce new sets of functions in a similar fashion to \cite{2}. This topic could be related to RHP methods, since the kernels we consider are integrable kernels, as can be seen from\eref{2.1.1}, one could expect that this approach would yield a functional formula even without any differential equation. For instance, this is achieved for $\varphi_\xi(x)=\phi(x+\xi)$ or $\varphi_\xi(x)=\phi(x\xi)$ in \cite{35}. But, as we previously mentioned, the generalization of\eref{1.3.2} along these lines, without any differential equation, is seemingly not achievable.

\paragraph{3.} Generalize Theorem \ref{t4.0.1} to a generic $\I=\sqcup_{j=1}^m[\tau_{2j-1};\tau_{2j}]$. This would be highly relevant in the context of two-dimensional quantum gravity and more generally double scaled random matrices as this would yield the gap probability for any interval.

\paragraph{4.} Compute $\lim_{\tau\rightarrow-\infty}\F\big([\tau,\infty)\big)$. Amongst other interests, this would notably allow for a check of the compatibility with the Forrester–Chen–Eriksen–Tracy Proposition \cite{16,33} about the large gap behavior of the Fredholm determinant.

\paragraph{5.} Try to formulate a Schlesinger system of equations of finite size associated to $\af{0}{0}$ and $\af{0}{1}$ only. The systematic appearance of $\af{n+1}{\alpha}$ in $\af{n}{\alpha}'$ through Proposition \ref{t2.2.33} is what motivated our infinite size Lax matrices. Nonetheless the same could be said for the appearance of $\af{n}{\alpha+1}$, but we dealt with this employing Proposition \ref{t2.2.36}, thus it should also be possible to close the system "in the $n$ direction" utilizing Lemma \ref{t3.1.9}. This brings forth the following question. Does the resulting system with finite size matrices, accounting for $\af{0}{0}$ and $\af{0}{1}$ solely, satisfy a set of Schlesinger equations?

\paragraph{6.} Prove Proposition \ref{t3.2.16} for an arbitrary component without any formal argument and derive the associated trace properties \cite{1}. This would establish the connection with isomonodromic systems.

\paragraph{7.} Can Condition \ref{t3.1.6} be lifted? This amounts to a careful analysis of the behaviour of $\af{0}{0}$ at any $\tj$ where $\psi$ vanishes. In particular, it would be interesting to know whether or not it is possible to show that, wherever $\psi$ vanishes, $\af{0}{0}$ vanishes more rapidly.

\subsection{Outline}\label{s14}

To conclude this introduction, we turn to the plan for the remaining part of the present paper.

To begin with, section \sref{s2} aims at the derivation of several results which will most notably provide us with the essential technical tools in order to tackle sections \sref{s3} and \sref{s4}. In particular \sref{s21} focuses on properties of the kernel whilst \sref{s22} introduces in details the quantities which shall be at the heart of our main derivations.

In the wake of these discussions, section \sref{s3} undertakes the Hamiltonian and Lax formalisms, providing the justifications for the upshots mentioned above but also bringing forth several outcomes which will turn out to bear relevance in \sref{s4} too.

Ultimately the endeavor of section \sref{s4} is to provide a proof to Theorem \ref{t4.0.1}, constructing on results from the previous sections in order to yield our principal upshot. Specifically, \sref{s41} focuses on\eref{1.3.2} whilst \sref{s42} leads us to\eref{1.3.1} as well as\eref{1.3.3}, further commenting on the asymptotic behaviour of the auxiliary wave functions.

For the appendices, \sref{sA} and \sref{sB} respectively lay out generic and specific examples satisfying Condition \ref{t1.1.4}, whilst \sref{sC} brings forth a generalization of\eref{1.3.2} for any $\af{n}{0}\,,\,\,n\in\N\,$, briefly sketching the derivation.

Before we begin with the derivations we would like to emphasize the fact that Notations \ref{t2.1.6}, \ref{t2.2.3} and \ref{t2.2.8} are extensively used and are therefore worth bearing in mind.

\subsection*{Acknowledgements}

We are grateful to Thomas Bothner and Guilherme L. F. Silva for useful discussions about the project.
This work was supported by CNRS through MITI interdisciplinary programs, EIPHI Graduate School (No.~ANR-17-EURE-0002), and Bourgogne-Franche-Comté region.

\section{On the Integral Operator and its Resolvent}\label{s2}

\subsection{Evaluation of the Kernel}\label{s21}

First we evaluate the kernel $\Ki(\xi,\zeta)$ in terms of the wave function and its $x$-derivative, all evaluated at $x=0$. Then we translate these $x$-derivatives into $\xi$ and $\zeta$-derivatives using Condition \ref{t1.1.4}. Finally this allows us to express the derivatives of the kernel in a form which will be useful later on.

\subsubsection{Christoffel-Darboux Type Formula}\label{s211}

For the sake of conciseness we introduce the differential operator $\So f(x)\coloneqq \partial_x^2f(x)-v(x;\xi)f(x)$. Then Condition \ref{t1.1.3} together with\eref{1.1.1} entail $\So\varphi_\xi=\xi\varphi_\xi\in\leb^2(\R_+)$. And this leads us to the following result.

\begin{lem}\label{t2.1.1}
The kernel of $\Ko_0$ satisfies
\begin{flalign}\label{2.1.1}
\begin{split}
\Ki(\xi,\zeta)=\frac{\ipr{\varphi_\xi}{[\So,\Pi_{\R_+}]\varphi_\zeta}}{\xi-\zeta}\,,
\end{split}
\end{flalign}
where $\ipr{\cdot}{\cdot}$ denotes the $\leb^2(\R)$ inner product and $[\cdot,\cdot]$ the commutator.
\end{lem}

\begin{proof}
Denoting by $\OO^*$ the formal adjoint of any operator $\OO$, we notice that $\So^*\varphi_\xi=\So\varphi_\xi$. Therefore
\begin{flalign}\label{2.1.2}
\begin{split}
\ipr{\varphi_\xi}{[\So,\Pi_{\R_+}]\varphi_\zeta}=\ipr{\So\varphi_\xi}{\Pi_{\R_+}\varphi_\zeta}-\ipr{\varphi_\xi}{\Pi_{\R_+}\So\varphi_\zeta}=(\xi-\zeta)\ipr{\varphi_\xi}{\Pi_{\R_+}\varphi_\zeta}=(\xi-\zeta)\Ko(\xi,\zeta)\,,
\end{split}
\end{flalign}
which completes the proof.
\end{proof}

Hence we turn to the evaluation of the commutator of $\So$ and $\Pi_{\R_+}$. Let us introduce an operator which will be extensively used throughout this text.

\begin{defi}\label{t2.1.2}
For any $\operatorname{J}\subseteq\R$, we define the differentiation operator as
\begin{flalign}\label{2.1.3}
\begin{split}
\Po:\sob^1(\operatorname{J})\longrightarrow\leb^2(\operatorname{J})\,;\,\,\,\,\,\Po f(x)\coloneqq\partial_x f(x)\,.
\end{split}
\end{flalign}
\end{defi}

\begin{pro}\label{t2.1.3}
For any $f\in\sob^1(\operatorname{J})$, we have $\Po^*f=-\Po f$.
\end{pro}

\begin{pro}\label{t2.1.4}
Given $\I\coloneqq \bigsqcup_{j=1}^m[\tau_{2j-1};\tau_{2j}]$, the following relation holds for all $f\in\sob^1(\R)$ and $g\in\leb^2(\R)$,
\begin{flalign}\label{2.1.4}
\begin{split}
\int_\R g(x)\Po\Pi_{\I}f(x)\,dx=-\sum_{j=1}^{2m}(-1)^jg(\tj)f(\tj)+\int_\R g(x)\Pi_{\I}\Po f(x)\,dx\,.
\end{split}
\end{flalign}
\end{pro}

\begin{lem}\label{t2.1.5}
The kernel of $\Ko_0$ is given by
\begin{flalign}\label{2.1.5}
\begin{split}
\Ki(\xi,\zeta)=\frac{1}{\xi-\zeta}\Big(\varphi_\xi(x)\big(\partial_x\varphi_\zeta(x)\big)-\big(\partial_x\varphi_\xi(x)\big)\varphi_\zeta(x)\Big)\Big|_{x=0}\,.
\end{split}
\end{flalign}
\end{lem}

\begin{proof}
First we note that Property \ref{t2.1.4} entails $\int_\R g(x)[\Po,\Pi_{\R_+}]f(x)\,dx=g(0)f(0)$, for any $f\in\sob^1(\R)$ and $g\in\leb^2(\R)$. Thus, since $\Pi_{\R_+}$ commutes with multiplication by $v$, we obtain
\begin{flalign}\label{2.1.6}
\begin{split}
\ipr{\varphi_\xi}{[\So,\Pi_{\R_+}]\varphi_\zeta}&=\ipr{\varphi_\xi}{[\Po^2,\Pi_{\R_+}]\varphi_\zeta}=\ipr{\varphi_\xi}{\Big([\Po,\Pi_{\R_+}]\Po+\Po[\Po,\Pi_{\R_+}]\Big)\varphi_\zeta} \\
&=\int_\R\Big(\varphi_\xi(x)[\Po,\Pi_{\R_+}]\big(\partial_x\varphi_\zeta(x)\big)-\big(\partial_x\varphi_\xi(x)\big)[\Po,\Pi_{\R_+}]\varphi_\zeta(x)\Big)dx \\
&=\Big(\varphi_\xi(x)\big(\partial_x\varphi_\zeta(x)\big)-\big(\partial_x\varphi_\xi(x)\big)\varphi_\zeta(x)\Big)\Big|_{x=0}\,,
\end{split}
\end{flalign}
where we also used Property \ref{t2.1.3}. Finally, injecting this in Lemma \ref{t2.1.1} yields the result.
\end{proof}

Finally we convert the $x$-derivatives into $\xi$ and $\zeta$-derivatives. And to do so, since all the contributions to the kernel are evaluated at $x=0$, we introduce the following quantities.

\begin{nota}\label{t2.1.6}
We write $u_0\coloneqq u(0)$, $\ud\coloneqq \partial_x\big|_{x=0}u(x)$ as well as $\udd\coloneqq \partial_x^2\big|_{x=0}u(x)$, and similar for the potential $v(x;\xi)$, i.e. $v_0\coloneqq v(0;\xi)$, which does not depend on $\xi$ according to Condition \ref{t1.1.6}.
\end{nota}

%{\color{blue}{One may use $\dot{u}(x) = \partial_x u(x)$ and also $\ddot{u}(x) = \partial^2_x u(x)$}}

\begin{defi}\label{t2.1.7}
Let $\psi(\xi)$ denote the wave function at $x=0$, i.e. $\psi(\xi)\coloneqq \varphi_\xi(0)$, and $\psi^{(l)}(\xi)\coloneqq \partial_\xi^l\psi(\xi)$.
\end{defi}

\begin{prop}\label{t2.1.8}
The kernel of $\Ko_0$ obeys a Christoffel-Darboux type formula, namely
\begin{flalign}\label{2.1.7}
\begin{split}
\Ki(\xi,\zeta)=\frac{\dot{u}_0}{\gamma}\frac{\psi(\xi)\psi'(\zeta)-\psi'(\xi)\psi(\zeta)}{\xi-\zeta}\,.
\end{split}
\end{flalign}
\end{prop}

\begin{proof}
To begin with, under Condition \ref{t1.1.4}, we equate $\partial_x\varphi_\xi(x)=\dot{u}(x)\phi'\big(u(x)+\gamma \xi\big)$ with $\partial_\xi\varphi_\xi(x)=\gamma\phi'\big(u(x)+\gamma \xi\big)$, yielding
\begin{flalign}\label{2.1.8}
\begin{split}
\partial_x\varphi_\xi(x)=\frac{\dot{u}(x)}{\gamma}\partial_\xi\varphi_\xi(x)\,.
\end{split}
\end{flalign}
Thereby $\partial_x\big|_{x=0}\varphi_\xi(x)=\frac{\dot{u}_0}{\gamma}\psi'(\xi)$ and hence the result follows from Lemma \ref{t2.1.5}.
\end{proof}

\subsubsection{Diagonal Values and Derivatives of the Kernel}\label{s212}

So we are now equipped with a practical expression for the kernel. However this expression looks singular for the diagonal values of the kernel, whence we evaluate the limit as $\xi\rightarrow \zeta$.

\begin{lem}\label{t2.1.9}
The second derivative of $\psi$ reads
\begin{flalign}\label{2.1.9}
\begin{split}
\psi''(\xi)=\frac{\gamma^2}{\dot{u}_0^2}\bigg(\big(v_0+\xi\big)\psi(\xi)-\frac{\ddot{u}_0}{\gamma}\psi'(\xi)\bigg)\,.
\end{split}
\end{flalign}
\end{lem}

\begin{proof}
Observe that, under Condition \ref{t1.1.4} with $\omega\coloneqq u(x)+\gamma \xi$, we have $\partial_\xi\varphi_\xi(x)=\gamma\phi'(\omega)$ and $\partial_\xi^2\varphi_\xi(x)=\gamma^2\phi''(\omega)$. And this leads to
\begin{flalign}\label{2.1.10}
\begin{split}
\big(v(x;\xi)+\xi\big)\varphi_\xi(x)=\partial_x^2\varphi_\xi(x)=\big(\dot{u}(x)\big)^2\phi''(\omega)+\ddot{u}(x)\phi'(\omega)=\frac{\big(\dot{u}(x)\big)^2}{\gamma^2}\partial_\xi^2\varphi_\xi(x)+\frac{\ddot{u}(x)}{\gamma}\partial_\xi\varphi_\xi(x)\,,
\end{split}
\end{flalign}
where the first equality is\eref{1.1.1}. Thus, setting $x=0$, we obtain $(v_0+\xi)\psi(\xi)=\frac{\dot{u}_0^2}{\gamma^2}\psi''(\xi)+\frac{\ddot{u}_0}{\gamma}\psi'(\xi)$, and therefore, under Condition \ref{t1.1.5}, we reach the desired result.
\end{proof}

\begin{prop}\label{t2.1.10}
The diagonal values of the kernel are
\begin{flalign}\label{2.1.11}
\begin{split}
\Ki(\xi,\xi)=\frac{\dot{u}_0}{\gamma}\big(\psi'(\xi)\big)^2-\frac{\gamma}{\dot{u}_0}\big(v_0+\xi\big)\big(\psi(\xi)\big)^2+\frac{\ddot{u}_0}{\dot{u}_0}\psi(\xi)\psi'(\xi)\,.
\end{split}
\end{flalign}
\end{prop}

\begin{proof}
Let $h\in\R\char`\\\{0\}$, then Proposition \ref{t2.1.8} implies
\begin{flalign}\label{2.1.12}
\begin{split}
\Ki(\xi+h,\xi)=\frac{\dot{u}_0}{\gamma h}\Big(\psi(\xi+h)\psi'(\xi)-\psi'(\xi+h)\psi(\xi)\Big)=\frac{\dot{u}_0}{\gamma}\bigg(\frac{\psi(\xi+h)-\psi(\xi)}{h}\psi'(\xi)-\frac{\psi'(\xi+h)-\psi'(\xi)}{h}\psi(\xi)\bigg)\,.
\end{split}
\end{flalign}
Then taking the $h\rightarrow0$ limit yields $\Ki(\xi,\xi)=\frac{\dot{u}_0}{\gamma}\big(\psi'(\xi)\big)^2-\frac{\dot{u}_0}{\gamma}\psi''(\xi)\psi(\xi)$, so that Lemma \ref{t2.1.9} finalizes the derivation.
\end{proof}

Besides, another practical result is the following one.

\begin{prop}\label{t2.1.11}
The derivatives of the kernel obey
\begin{flalign}\label{2.1.13}
\begin{split}
\big(\partial_\xi+\partial_\zeta\big)\Ki(\xi,\zeta)=-\frac{\gamma}{\dot{u}_0}\bigg(\psi(\xi)\psi(\zeta)+\frac{\ddot{u}_0}{\dot{u}_0}\Ki(\xi,\zeta)\bigg)\,.
\end{split}
\end{flalign}
\end{prop}

\begin{proof}
Using Proposition \ref{t2.1.8} and then Lemma \ref{t2.1.9}, we obtain
\begin{flalign}\label{2.1.14}
\begin{split}
\big(\partial_\xi+\partial_\zeta\big)\Ki(\xi,\zeta)&=\frac{\dot{u}_0}{\gamma(\xi-\zeta)}\Big(\psi(\xi)\psi''(\zeta)-\psi''(\xi)\psi(\zeta)\Big) \\
&=\frac{\gamma}{\dot{u}_0(\xi-\zeta)}\bigg(\psi(\xi)\psi(\zeta)\big(\zeta-\xi\big)-\frac{\ddot{u}_0}{\gamma}\Big(\psi(\xi)\psi'(\zeta)-\psi'(\xi)\psi(\zeta)\Big)\bigg)\,,
\end{split}
\end{flalign}
completing the proof.
\end{proof}

\subsection{Corresponding Resolvent}\label{s22}

Here we focus on the following quantities which will play a central role in the derivation of the results of sections \sref{s3} and \sref{s4}.

\begin{defi}\label{t2.2.1}
We refer to the integral operator
\begin{flalign}\label{2.2.1}
\begin{split}
\Ro:\leb^2(\R)&\longrightarrow\leb^2(\R) \\
f&\longmapsto\Ko\big(\id-\Ko\big)^{-1}f
\end{split}
\end{flalign}
as the resolvent.
\end{defi}

\begin{defi}\label{t2.2.2}
The $(n,k)\in\Z_{\geq0}\times\Z_{\geq0}$ Auxiliary Wave Function (AWF) is defined by
\begin{flalign}\label{2.2.2}
\begin{split}
\af{n}{k}\coloneqq \Ro^n\big(\id-\Ko\big)^{-1}\Po^k\psi\,.
\end{split}
\end{flalign}
And, for a fixed $n$, we refer to $\{\af{n}{\alpha}\,\,;\,\,\alpha\in\{0,1\}\,\}$ as the $n^{\operatorname{th}}$ order AWF.
\end{defi}

\begin{nota}\label{t2.2.3}
From now on $a,b\in\{0,1,2\}$ and $\alpha,\beta\in\{0,1\}$.
\end{nota}

\begin{pro}\label{t2.2.4}
According to Condition \ref{t1.1.2}, $\af{n}{a}\in\leb^2(\R)$, $\forall n\in\Z_{\geq0}$.
\end{pro}

Our first main interest will be the kernels of $\Ro^n$, $\forall n\in\N$. Whence we first determine, in terms of the AWF, the resolvent kernel, $n=1$, and then obtain by induction the kernel of $\Ro^n$ for any $n$. These results will be at the heart of section \sref{s3}. Afterwards our interest will shift to the AWF, first evaluating their derivatives and then expressing $\af{n}{2}$ in terms of $\af{n}{0}$ and $\af{n}{1}$. Doing so shall involve quite a few inductions too for $n\in\N$, and thereby the case $n=0$ will be treated first so that we can get a feeling on how to tackle the derivations for higher $n$. And in turn these results will be at the heart of section \sref{s4}.

\subsubsection{Resolvent Kernels}\label{s221}

\begin{defi}\label{t2.2.5}
We denote by $\Ri_n:\R^2\rightarrow\R$ the kernel of $\Ro^n$, namely, $\forall f\in\leb^2(\R)$,
\begin{flalign}\label{2.2.3}
\begin{split}
\Ro^n:f\longmapsto\int_\R\Ri_n(\cdot,\zeta)f(\zeta)\,d\zeta\,.
\end{split}
\end{flalign}
And we call $\Ri_n$ the $n^{\operatorname{th}}$ order resolvent kernel.
\end{defi}

We now relate $\Ri_n$ to the AWF. Achieving this will require the following mappings.

\begin{defi}\label{t2.2.6}
Given any two functions $g\in\leb^2(\R)$ and $h\in\Hil$ for some Hilbert space $\Hil$, we define
\begin{flalign}\label{2.2.4}
\begin{split}
\T_{h,g}:\leb^2(\R)&\longrightarrow\Hil \\
f&\longmapsto\ipr{g}{f}h\,.
\end{split}
\end{flalign}
And if $g=h$, then we write $\T_g\coloneqq \T_{g,g}$.
\end{defi}

\begin{defi}\label{t2.2.7}
Let $\D(\Qo)\coloneqq \{f\in\leb^2(\R)\,\,;\,\,\leb^2(\R)\ni g:\xi\mapsto\xi f(\xi)\,\}$, then we define the multiplication operator as
\begin{flalign}\label{2.2.5}
\begin{split}
\Qo:\D(\Qo)&\longrightarrow\leb^2(\R)\,;\,\,\,\,\,\Qo f(\xi)\coloneqq \xi f(\xi)\,.
\end{split}
\end{flalign}
\end{defi}

\begin{nota}\label{t2.2.8}
Let us write $\rho\coloneqq \big(\id-\Ko\big)^{-1}$.
\end{nota}

\begin{pro}\label{t2.2.9}
Neumann series entail $\rho=\id+\Ro$.
\end{pro}

Our first step is to determine $\Ri_1$ in terms of the $0^{\operatorname{th}}$ order AWF. During this section and the next one, \sref{s3}, we consider
\begin{flalign}\label{2.2.6}
\begin{split}
\I=\bigsqcup_{j=1}^m[\tau_{2j-1};\tau_{2j}]\,.
\end{split}
\end{flalign}

\begin{lem}\label{t2.2.10}
The commutator of $\Qo$ and $\Ko$ admits the following representation,
\begin{flalign}\label{2.2.7}
\begin{split}
[\Qo,\Ko]=\frac{\dot{u}_0}{\gamma}\big(\T_{\psi,\psi'}-\T_{\psi',\psi}\big)\Pi_{\I}\,.
\end{split}
\end{flalign}
\end{lem}

\begin{proof}
Noticing that $\Ko:f\mapsto\int_\R\Ki(\cdot,\zeta)\ind{\I}(\zeta)f(\zeta)\,d\zeta$, we then observe that, $\forall f\in\D(\Qo)$,
\begin{flalign}\label{2.2.8}
\begin{split}
[\Qo,\Ko]f(\xi)&=\int_{\I}\Big(\xi\Ki(\xi,\zeta)f(\zeta)-\Ki(\xi,\zeta)\zeta f(\zeta)\Big)\,d\zeta=\frac{\dot{u}_0}{\gamma}\int_{\I}\Big(\psi(\xi)\psi'(\zeta)-\psi'(\xi)\psi(\zeta)\Big)f(\zeta)\,d\zeta \\
&=\frac{\dot{u}_0}{\gamma}\Big(\ipr{\psi'}{\Pi_{\I}f}\psi(\xi)-\ipr{\psi}{\Pi_{\I}f}\psi'(\xi)\Big)=\frac{\dot{u}_0}{\gamma}\big(\T_{\psi,\psi'}-\T_{\psi',\psi}\big)\Pi_{\I}f(\xi)\,,
\end{split}
\end{flalign}
where we used Proposition \ref{t2.1.8}.
\end{proof}

\begin{lem}\label{t2.2.11}
Given any operator $\OO$, its commutator with $\Ro$ agrees with its commutator with $\rho$, namely
\begin{flalign}\label{2.2.9}
\begin{split}
[\OO,\Ro]=\rho[\OO,\Ko]\rho=[\OO,\rho]\,.
\end{split}
\end{flalign}
\end{lem}

\begin{proof}
Using $\rho^{-1}=\big(\id-\Ko\big)$ and $\Ro=\Ko\rho$, we compute
\begin{flalign}\label{2.2.10}
\begin{split}
[\OO,\Ro]=\rho\Big(\big(\id-\Ko\big)\OO\Ko-\Ko\OO\big(\id-\Ko\big)\Big)\rho\,,
\end{split}
\end{flalign}
and observe that the same trick can be applied for $[\OO,\rho]$ and arrives at the same result.
\end{proof}

\begin{lem}\label{t2.2.12}
The operators $\Ko$, $\rho$, and $\Ro$ are related to their adjoint by the same relation, namely
\begin{flalign}\label{2.2.11}
\begin{split}
\Pi_{\I}\Ko=\Ko^*\Pi_{\I}\,,\,\,\,\,\,\Pi_{\I}\rho=\rho^*\Pi_{\I}\,,\,\,\,\,\,\Pi_{\I}\Ro=\Ro^*\Pi_{\I}\,.
\end{split}
\end{flalign}
\end{lem}

\begin{proof}
The definition of $\Ki$,\eref{1.1.3}, directly entails $\Ki(\xi,\zeta)=\Ki(\zeta,\xi)$, ensuring $\Ko_0^*=\Ko_0$. Therefore, since $\Pi_{\I}^*=\Pi_{\I}$, it follows that $\Ko^*=\Pi_{\I}\Ko_0$, which is yielding
\begin{flalign}\label{2.2.12}
\begin{split}
\Pi_{\I}\Ko=\Pi_{\I}\Ko_0\Pi_{\I}=\Ko^*\Pi_{\I}\,.
\end{split}
\end{flalign}
Furthermore this leads to same relation for $\rho$ as a consequence of Neumann series, $\rho=\sum_{j=0}^\infty\big(\Ko_0\Pi_{\I}\big)^j$. And, since $\rho$ commutes with $\Ko$, this in turn leads us to the same relation for $\Ro$.
\end{proof}

\begin{nota}\label{t2.2.13}
For any two operators $\OO_1$ and $\OO_2$ such that their commutator is an integral operator, we denote its kernel by $\Ci\big(\OO_1,\OO_2\,;\,\xi,\zeta\big)$, i.e.
\begin{flalign}\label{2.2.13}
\begin{split}
[\OO_1,\OO_2]:f\longmapsto\int_\R\Ci\big(\OO_1,\OO_2\,;\,\,\cdot\,,\zeta\big)f(\zeta)\,d\zeta\,.
\end{split}
\end{flalign}
\end{nota}

\begin{lem}\label{t2.2.14}
The first order resolvent kernel can be expressed in terms of the zeroth order AWF as follows,
\begin{flalign}\label{2.2.14}
\begin{split}
\Ri_1(\xi,\zeta)=\frac{\dot{u}_0}{\gamma}\frac{\af{0}{0}(\xi)\af{0}{1}(\zeta)-\af{0}{1}(\xi)\af{0}{0}(\zeta)}{\xi-\zeta}\ind{\I}(\zeta)\,.
\end{split}
\end{flalign}
\end{lem}

\begin{proof}
To begin with we have, $\forall f\in\D(\Qo)$,
\begin{flalign}\label{2.2.15}
\begin{split}
[\Qo,\Ro]f(\xi)=\int_\R\big(\xi-\zeta\big)\Ri_1(\xi,\zeta)f(\zeta)\,d\zeta\,.
\end{split}
\end{flalign}
Hence $[\Qo,\Ro]$ is an integral operator and its kernel satisfies the following kernel relation,
\begin{flalign}\label{2.2.16}
\begin{split}
\Ri_1(\xi,\zeta)=\frac{\Ci\big(\Qo,\Ro\,;\,\xi,\zeta\big)}{\xi-\zeta}\,.
\end{split}
\end{flalign}
Besides, using Lemmas \ref{t2.2.11} and \ref{t2.2.10}, we are led to, $\forall f\in\D(\Qo)$,
\begin{flalign}\label{2.2.17}
\begin{split}
[\Qo,\Ro]f(\xi)&=\frac{\dot{u}_0}{\gamma}\rho\big(\T_{\psi,\psi'}-\T_{\psi',\psi}\big)\Pi_{\I}\rho f(\xi)=\frac{\dot{u}_0}{\gamma}\Big(\ipr{\psi'}{\Pi_{\I}\rho f}\rho\psi(\xi)-\ipr{\psi}{\Pi_{\I}\rho f}\rho\psi'(\xi)\Big) \\
&=\frac{\dot{u}_0}{\gamma}\Big(\ipr{\rho\psi'}{\Pi_{\I}f}\rho\psi(\xi)-\ipr{\rho\psi}{\Pi_{\I}f}\rho\psi'(\xi)\Big)=\frac{\dot{u}_0}{\gamma}\Big(\ipr{\af{0}{1}}{\Pi_{\I}f}\af{0}{0}(\xi)-\ipr{\af{0}{0}}{\Pi_{\I}f}\af{0}{1}(\xi)\Big) \\
&=\frac{\dot{u}_0}{\gamma}\int_\R\Big(\af{0}{0}(\xi)\af{0}{1}(\zeta)-\af{0}{1}(\xi)\af{0}{0}(\zeta)\Big)\ind{\I}(\zeta)f(\zeta)\,d\zeta\,,
\end{split}
\end{flalign}
where we also used Lemma \ref{t2.2.12} and Definition \ref{t2.2.2}. Finally, injecting this in\eref{2.2.16} yields the result.
\end{proof}

Once again this is a Christoffel-Darboux type formula which looks singular on the diagonal, whence we still need to evaluate the limit $\xi\rightarrow\zeta$.

\begin{lem}\label{t2.2.15}
The diagonal values of the first order resolvent kernel can be expressed as
\begin{flalign}\label{2.2.18}
\begin{split}
\Ri_1(\xi,\xi)=\frac{\dot{u}_0}{\gamma}\bigg(\big(\partial_\xi\af{0}{0}(\xi)\big)\af{0}{1}(\xi)-\big(\partial_\xi\af{0}{1}(\xi)\big)\af{0}{0}(\xi)\bigg)\ind{\I}(\xi)\,.
\end{split}
\end{flalign}
\end{lem}

\begin{proof}
We proceed similarly to the proof of Proposition \ref{t2.1.10}, setting $h\in\R\char`\\\{0\}$, Lemma \ref{t2.2.14} implies
\begin{flalign}\label{2.2.19}
\begin{split}
\Ri_1(\xi+h,\xi)=\frac{\dot{u}_0}{\gamma}\bigg(\frac{\af{0}{0}(\xi+h)-\af{0}{0}(\xi)}{h}\af{0}{1}(\xi)-\frac{\af{0}{1}(\xi+h)-\af{0}{1}(\xi)}{h}\af{0}{0}(\xi)\bigg)\ind{\I}(\xi)\,.
\end{split}
\end{flalign}
And the $h\rightarrow0$ limit completes the proof.
\end{proof}

We turn to the evaluation of $\Ri_n$, proceeding by induction.

\begin{rem}\label{t2.2.16}
For any $F$ and $G$ we have
\begin{flalign}\label{2.2.20}
\begin{split}
\sum_{k=0}^{n-1}F(n-k-1)\,\,G(k)=\sum_{k=1}^nF(n-k)\,\,G(k-1)=\sum_{k=1}^nF(k-1)\,\,G(n-k)\,.
\end{split}
\end{flalign}
\end{rem}

\begin{lem}\label{t2.2.17}
The commutator of $\Qo$ and $\Ro^n$ can be evaluated by
\begin{flalign}\label{2.2.21}
\begin{split}
[\Qo,\Ro^n]=\frac{\dot{u}_0}{\gamma}\rho\Bigg(\sum_{k=0}^{n-1}\Ro^{n-k-1}\big(\T_{\psi,\psi'}-\T_{\psi',\psi}\big)\Pi_{\I}\Ro^k\Bigg)\rho\,.
\end{split}
\end{flalign}
\end{lem}

\begin{proof}
Lemmas \ref{t2.2.10} and \ref{t2.2.11} initiate the induction, $[\Qo,\Ro]=\frac{\dot{u}_0}{\gamma}\rho\big(\T_{\psi,\psi'}-\T_{\psi',\psi}\big)\Pi_{\I}\rho$. Then assuming that the result holds for $n-1$ allows us to compute
\begin{flalign}\label{2.2.22}
\begin{split}
[\Qo,\Ro^n]=\Ro[\Qo,\Ro^{n-1}]+[\Qo,\Ro]\Ro^{n-1}=\frac{\dot{u}_0}{\gamma}\rho\Bigg(\sum_{k=0}^{n-2}\Ro^{n-k-1}\big(\T_{\psi,\psi'}-\T_{\psi',\psi}\big)\Pi_{\I}\Ro^k\Bigg)\rho+\frac{\dot{u}_0}{\gamma}\rho\big(\T_{\psi,\psi'}-\T_{\psi',\psi}\big)\Pi_{\I}\Ro^{n-1}\rho\,,
\end{split}
\end{flalign}
where we also used the initiation again.
\end{proof}

\begin{prop}\label{t2.2.18}
The $n^{\operatorname{th}}$ order resolvent kernel can be expressed as
\begin{flalign}\label{2.2.23}
\begin{split}
\Ri_n(\xi,\zeta)=\frac{\dot{u}_0\ind{\I}(\zeta)}{\gamma(\xi-\zeta)}\sum_{k=1}^n\Big(\af{n-k}{0}(\xi)\af{k-1}{1}(\zeta)-\af{n-k}{1}(\xi)\af{k-1}{0}(\zeta)\Big)\,.
\end{split}
\end{flalign}
\end{prop}

\begin{proof}
We proceed similarly to Lemma \ref{t2.2.14}, in particular\eref{2.2.15} and\eref{2.2.16} generalize to
\begin{flalign}\label{2.2.24}
\begin{split}
\Ri_n(\xi,\zeta)=\frac{\Ci\big(\Qo,\Ro^n\,;\,\xi,\zeta\big)}{\xi-\zeta}\,.
\end{split}
\end{flalign}
And then Lemma \ref{t2.2.17} enables us to determine $\Ci\big(\Qo,\Ro^n\,;\,\xi,\zeta\big)$ as follows, $\forall f\in\D(\Qo)$,
\begin{flalign}\label{2.2.25}
\begin{split}
[\Qo,\Ro^n]f(\xi)&=\frac{\dot{u}_0}{\gamma}\sum_{k=0}^{n-1}\Ro^{n-k-1}\rho\Big(\psi(\xi)\ipr{\psi'}{\Pi_{\I}\Ro^k\rho f}-\psi'(\xi)\ipr{\psi}{\Pi_{\I}\Ro^k\rho f}\Big) \\
&=\frac{\dot{u}_0}{\gamma}\sum_{k=0}^{n-1}\Ro^{n-k-1}\rho\Big(\psi(\xi)\ipr{\Ro^k\rho\psi'}{\Pi_{\I}f}-\psi'(\xi)\ipr{\Ro^k\rho\psi}{\Pi_{\I}f}\Big) \\
&=\frac{\dot{u}_0}{\gamma}\sum_{k=0}^{n-1}\Big(\af{n-k-1}{0}(\xi)\ipr{\af{k}{1}}{\Pi_{\I}f}-\af{n-k-1}{1}(\xi)\ipr{\af{k}{0}}{\Pi_{\I}f}\Big) \\
&=\frac{\dot{u}_0}{\gamma}\int_\R\sum_{k=0}^{n-1}\Big(\af{n-k-1}{0}(\xi)\af{k}{1}(\zeta)-\af{n-k-1}{1}(\xi)\af{k}{0}(\zeta)\Big)\ind{\I}(\zeta)f(\zeta)\,d\zeta\,,
\end{split}
\end{flalign}
where we used Lemma \ref{t2.2.12}. Therefore
\begin{flalign}\label{2.2.26}
\begin{split}
\Ri_n(\xi,\zeta)=\frac{\dot{u}_0\ind{\I}(\zeta)}{\gamma(\xi-\zeta)}\sum_{k=0}^{n-1}\Big(\af{n-k-1}{0}(\xi)\af{k}{1}(\zeta)-\af{n-k-1}{1}(\xi)\af{k}{0}(\zeta)\Big)\,,
\end{split}
\end{flalign}
so that Remark \ref{t2.2.16} yields the result.
\end{proof}

\begin{cor}\label{t2.2.19}
Remark \ref{t2.2.16} also entails that\eref{2.2.23} can be rewritten as
\begin{flalign}\label{2.2.27}
\begin{split}
\Ri_n(\xi,\zeta)=\frac{\dot{u}_0\ind{\I}(\zeta)}{\gamma(\xi-\zeta)}\sum_{k=1}^n\Big(\af{n-k}{0}(\xi)\af{k-1}{1}(\zeta)-\af{k-1}{1}(\xi)\af{n-k}{0}(\zeta)\Big)\,.
\end{split}
\end{flalign}
\end{cor}

\begin{prop}\label{t2.2.20}
The diagonal values of the $n^{\operatorname{th}}$ order resolvent kernel are given by
\begin{flalign}\label{2.2.28}
\begin{split}
\Ri_n(\xi,\xi)=\frac{\dot{u}_0}{\gamma}\sum_{k=1}^{n}\bigg(\big(\partial_\xi\af{n-k}{0}(\xi)\big)\af{k-1}{1}(\xi)-\big(\partial_\xi\af{k-1}{1}(\xi)\big)\af{n-k}{0}(\xi)\bigg)\ind{\I}(\xi)\,.
\end{split}
\end{flalign}
\end{prop}

\begin{proof}
With the same trick as before for $h\in\R\char`\\\{0\}$ and using Corollary \ref{t2.2.19}, we obtain
\begin{flalign}\label{2.2.29}
\begin{split}
\Ri_n(\xi+h,\xi)=\frac{\dot{u}_0}{\gamma}\sum_{k=1}^n\bigg(\frac{\af{n-k}{0}(\xi+h)-\af{n-k}{0}(\xi)}{h}\af{k-1}{1}(\xi)-\frac{\af{k-1}{1}(\xi+h)-\af{k-1}{1}(\xi)}{h}\af{n-k}{0}(\xi)\bigg)\ind{\I}(\xi)\,.
\end{split}
\end{flalign}
And the $h\rightarrow0$ limit indeed leads us to\eref{2.2.28}.
\end{proof}

\subsubsection{Derivatives of the Auxiliary Wave Functions}\label{s222}

Notice that, in the Definition \ref{t2.2.2} of the AWF, $\Ko$ acts on $\psi$ and hence, since $\Ko$ contains $\Pi_{\I}$, the AWF depend on our choice of the interval $\I$, in particular they depend on the parameters $\tau_j$ appearing in\eref{2.2.6}. Thus we evaluate both $\partial_j\af{n}{a}(\xi)\coloneqq \partial_{\tau_j}\af{n}{a}(\xi)$ and $\partial_\xi\af{n}{a}(\xi)$. And to achieve this, we begin with the zeroth order AWF, then generalizing the results to $\af{n}{a}$, $\forall n\in\N$.

\begin{pro}\label{t2.2.21}
Recall that, similarly to Property \ref{t2.1.4}, we have $\int_\R g(x)\big(\partial_j\Pi_{\I}\big)f(x)\,dx=(-1)^jg(\tau_j)f(\tau_j)$, which holds for any $j\in\{1,2,...,2m\}\subset\N$, and $f,g\in\leb^2(\R)$.
\end{pro}

\begin{lem}\label{t2.2.22}
The following relations hold $\forall j\in\{1,2,...,2m\}\subset\N$,
\begin{flalign}\label{2.2.30}
\begin{split}
\partial_j\Ro=\partial_j\rho=\Ro\big(\partial_j\Pi_{\I}\big)\rho\,.
\end{split}
\end{flalign}
\end{lem}

\begin{proof}
The first equality follows from Property \ref{t2.2.9}. Whilst for the second one, recall that, in general $\partial_a\OO^{-1}=-\OO^{-1}\big(\partial_a\OO\big)\OO^{-1}$ for any $\OO$ depending on $a$, and thereby
\begin{flalign}\label{2.2.31}
\begin{split}
\partial_j\rho=-\rho\Big(\partial_j\big(\id-\Ko\big)\Big)\rho=\rho\Ko_0\big(\partial_j\Pi_{\I}\big)\rho\,.
\end{split}
\end{flalign}
Then observe that, $\forall f,g\in\leb^2(\R)$,
\begin{flalign}\label{2.2.32}
\begin{split}
\int_\R g(\zeta)\Ko_0\big(\partial_j\Pi_{\I}\big)f(\zeta)\,d\zeta=(-1)^j\Ko_0^*g(\tau_j)f(\tau_j)=(-1)^j\Pi_{\I}\Ko_0^*g(\tau_j)f(\tau_j)=\int_\R g(\zeta)\Ko_0\Pi_{\I}\big(\partial_j\Pi_{\I}\big)f(\zeta)\,d\zeta\,,
\end{split}
\end{flalign}
where we made use of the fact that $\ind{\I}(\tau_j)=1$. So that $\Ko_0\big(\partial_j\Pi_{\I}\big)=\Ko\big(\partial_j\Pi_{\I}\big)$, completing the proof.
\end{proof}

\begin{lem}\label{t2.2.23}
The parameter dependence of the zeroth order AWF is given by
\begin{flalign}\label{2.2.33}
\begin{split}
\partial_j\af{0}{a}(\xi)=(-1)^j\Ri_1(\xi,\tau_j)\af{0}{a}(\tau_j)\,.
\end{split}
\end{flalign}
\end{lem}

\begin{proof}
Using Lemma \ref{t2.2.22}, we compute
\begin{flalign}\label{2.2.34}
\begin{split}
\partial_j\af{0}{a}(\xi)=\Ro\big(\partial_j\Pi_{\I}\big)\rho\psi^{(a)}(\xi)=\int_\R\Ri_1(\xi,\zeta)\big(\partial_j\Pi_{\I}\big)\af{0}{a}(\zeta)\,d\zeta\,,
\end{split}
\end{flalign}
and hence Property \ref{t2.2.21} finalizes the derivation.
\end{proof}

\begin{lem}\label{t2.2.24}
According to Definition \ref{t2.2.6}, the commutator of $\Po$ and $\Ko$ obeys
\begin{flalign}\label{2.2.35}
\begin{split}
[\Po,\Ko]=-\frac{\gamma}{\dot{u}_0}\bigg(\T_{\psi}\Pi_{\I}+\frac{\ddot{u}_0}{\dot{u}_0}\Ko\bigg)-\sum_{j=1}^{2m}\Ko\big(\partial_j\Pi_{\I}\big)\,.
\end{split}
\end{flalign}
\end{lem}

\begin{proof}
First we note that, $\forall f\in\sob^1(\R)\cap\D(\Ko_0)$,
\begin{flalign}\label{2.2.36}
\begin{split}
[\Po,\Ko_0]f(\xi)=\int_\R\Big(\partial_\xi\Ki(\xi,\zeta)-\Ki(\xi,\zeta)\partial_\zeta\Big)f(\zeta)\,d\zeta\,,
\end{split}
\end{flalign}
therefore, integrating by parts, we obtain $\Ci\big(\Po,\Ko_0\,;\,\xi,\zeta\big)=\big(\partial_\xi+\partial_\zeta\big)\Ki(\xi,\zeta)$. Henceforth, using Proposition \ref{t2.1.11}, we can rewrite the above commutator as follows, $\forall f\in\sob^1(\R)\cap\D(\Ko_0)$,
\begin{flalign}\label{2.2.37}
\begin{split}
-\frac{\gamma}{\dot{u}_0}\bigg(\T_{\psi}+\frac{\ddot{u}_0}{\dot{u}_0}\Ko_0\bigg)f(\xi)=-\frac{\gamma}{\dot{u}_0}\int_\R\bigg(\psi(\xi)\psi(\zeta)f(\zeta)+\frac{\ddot{u}_0}{\dot{u}_0}\Ki(\xi,\zeta)f(\zeta)\bigg)d\zeta=\int_\R f(\zeta)\big(\partial_\xi+\partial_\zeta\big)\Ki(\xi,\zeta)\,d\zeta\,.
\end{split}
\end{flalign}
Yielding $[\Po,\Ko_0]=-\frac{\gamma}{\dot{u}_0}\Big(\T_{\psi}+\frac{\ddot{u}_0}{\dot{u}_0}\Ko_0\Big)$. Besides, Properties \ref{t2.1.4} and \ref{t2.2.21} are leading to, $\forall f,g\in\leb^2(\R)$,
\begin{flalign}\label{2.2.38}
\begin{split}
\int_\R g(x)[\Po,\Pi_{\I}]f(x)\,dx=-\sum_{j=1}^{2m}(-1)^jg(\tau_j)f(\tau_j)=-\sum_{j=1}^{2m}\int_\R g(x)\big(\partial_j\Pi_{\I}\big)f(x)\,dx\,.
\end{split}
\end{flalign}
So that we end up with
\begin{flalign}\label{2.2.39}
\begin{split}
[\Po,\Ko]=[\Po,\Ko_0]\Pi_{\I}+\Ko_0[\Po,\Pi_{\I}]=-\frac{\gamma}{\dot{u}_0}\bigg(\T_{\psi}\Pi_{\I}+\frac{\ddot{u}_0}{\dot{u}_0}\Ko\bigg)-\sum_{j=1}^{2m}\Ko_0\big(\partial_j\Pi_{\I}\big)\,,
\end{split}
\end{flalign}
and injecting\eref{2.2.32} into this relation leads us to\eref{2.2.35}.
\end{proof}

\begin{defi}\label{t2.2.25}
We introduce the following quantities, $\au{n}{a}\coloneqq \ipr{\psi}{\Pi_{\I}\Ro^n\rho\Po^a\psi}$, $\forall n\in\Z_{\geq0}$.
\end{defi}

\begin{pro}\label{t2.2.26}
Lemma \ref{t2.2.12} entails $\ipr{\Po^a\psi}{\Pi_{\I}\Ro^n\rho\psi}=\au{n}{a}$.
\end{pro}

\begin{lem}\label{t2.2.27}
The zeroth order AWF satisfy
\begin{flalign}\label{2.2.40}
\begin{split}
\partial_\xi\af{0}{\alpha}(\xi)=\af{0}{\alpha+1}(\xi)-\frac{\gamma}{\dot{u}_0}\bigg(\au{0}{\alpha}\af{0}{0}(\xi)+\frac{\ddot{u}_0}{\dot{u}_0}\af{1}{\alpha}(\xi)\bigg)-\sum_{j=1}^{2m}\partial_j\af{0}{\alpha}(\xi)\,.
\end{split}
\end{flalign}
\end{lem}

\begin{proof}
Firstly we notice that Lemma \ref{t2.2.11} ensures
\begin{flalign}\label{2.2.41}
\begin{split}
\Po\,\af{0}{\alpha}(\xi)=\af{0}{\alpha+1}(\xi)+[\Po,\rho]\Po^\alpha\psi(\xi)=\af{0}{\alpha+1}(\xi)+\rho[\Po,\Ko]\af{0}{\alpha}(\xi)\,.
\end{split}
\end{flalign}
Whilst Lemma \ref{t2.2.24} leads to
\begin{flalign}\label{2.2.42}
\begin{split}
\rho[\Po,\Ko]\af{0}{\alpha}(\xi)&=-\frac{\gamma}{\dot{u}_0}\rho\bigg(\psi(\xi)\ipr{\psi}{\Pi_{\I}\rho\Po^\alpha\psi}+\frac{\ddot{u}_0}{\dot{u}_0}\Ko\rho\Po^\alpha\psi(\xi)\bigg)-\sum_{j=1}^{2m}\Ro\big(\partial_j\Pi_{\I}\big)\af{0}{\alpha}(\xi) \\
&=-\frac{\gamma}{\dot{u}_0}\bigg(\af{0}{0}(\xi)\au{0}{\alpha}+\frac{\ddot{u}_0}{\dot{u}_0}\af{1}{\alpha}(\xi)\bigg)-\sum_{j=1}^{2m}(-1)^j\Ri_1(\xi,\tau_j)\af{0}{\alpha}(\tau_j)\,,
\end{split}
\end{flalign}
where we used Property \ref{t2.2.21}. Finally, inserting Lemma \ref{t2.2.23} in this expression and in turn injecting the result in\eref{2.2.41}, we indeed obtain\eref{2.2.40}.
\end{proof}

\begin{cor}\label{t2.2.28}
Lemma \ref{t2.2.27} imply
\begin{flalign}\label{2.2.43}
\begin{split}
\frac{d}{d\tau_j}\af{0}{\alpha}(\tau_j)=\big(\partial_\xi+\partial_j\big)\big|_{\xi=\tau_j}\af{0}{\alpha}(\xi)=\af{0}{\alpha+1}(\tau_j)-\frac{\gamma}{\dot{u}_0}\bigg(\au{0}{\alpha}\af{0}{0}(\tau_j)+\frac{\ddot{u}_0}{\dot{u}_0}\af{1}{\alpha}(\tau_j)\bigg)-\sum_{\substack{i=1\\i\neq j}}^{2m}\partial_i\af{0}{\alpha}(\tau_j)\,.
\end{split}
\end{flalign}
\end{cor}

Henceforth we are ready to discuss the generalization of these results for $\af{n}{a}$, $\forall n\in\Z_{\geq0}$. Since our expressions will become somewhat lengthy, we introduce the following notation which shall allow us to shorten our results.

\begin{nota}\label{t2.2.29}
We write $\Gamma\coloneqq \T_{\psi}\Pi_{\I}+\frac{\ddot{u}_0}{\dot{u}_0}\Ko$.
\end{nota}

So that, for instance, Lemma \ref{t2.2.24} becomes $[\Po,\Ko]=-\frac{\gamma}{\dot{u}_0}\Gamma-\sum_{j=1}^{2m}\Ko\dl[\I]{j}$.

\begin{lem}\label{t2.2.30}
The following relation is satisfied, $\forall n\in\Z_{\geq0}$\,,
\begin{flalign}\label{2.2.44}
\begin{split}
\partial_j\Ro^n\rho=\Ro\big(\partial_j\Pi_{\I}\big)\Ro^n\rho+\sum_{k=0}^{n-1}\bigg(\Ro^{n-k+1}\big(\partial_j\Pi_{\I}\big)+\Ro^{n-k}\big(\partial_j\Pi_{\I}\big)\bigg)\Ro^k\rho\,.
\end{split}
\end{flalign}
\end{lem}

\begin{proof}
The induction is initiated, for $n=0$, by Lemma \ref{t2.2.22}. And then, assuming that this holds for $n-1$, we may obtain
\begin{flalign}\label{2.2.45}
\begin{split}
\partial_j\Ro^n\rho&=\Ro\partial_j\Ro^{n-1}\rho+\big(\partial_j\Ro\big)\Ro^{n-1}\rho \\
&=\Ro^2\big(\partial_j\Pi_{\I}\big)\Ro^{n-1}\rho+\sum_{k=0}^{n-2}\bigg(\Ro^{n-k+1}\big(\partial_j\Pi_{\I}\big)+\Ro^{n-k}\big(\partial_j\Pi_{\I}\big)\bigg)\Ro^k\rho+\Ro\big(\partial_j\Pi_{\I}\big)\big(\id+\Ro\big)\Ro^{n-1}\rho\,,
\end{split}
\end{flalign}
where we used Lemma \ref{t2.2.22} and Property \ref{t2.2.9} together.
\end{proof}

\begin{prop}\label{t2.2.31}
The parameter dependence of the AWF, $\forall n\in\Z_{\geq0}$, is given by
\begin{flalign}\label{2.2.46}
\begin{split}
\partial_j\af{n}{a}(\xi)=(-1)^j\Ri_1(\xi,\tau_j)\af{n}{a}(\tau_j)+(-1)^j\sum_{k=0}^{n-1}\bigg(\Ri_{n-k+1}(\xi,\tau_j)+\Ri_{n-k}(\xi,\tau_j)\bigg)\af{k}{a}(\tau_j)\,.
\end{split}
\end{flalign}
\end{prop}

\begin{proof}
This is a consequence of Lemma \ref{t2.2.30},
\begin{flalign}\label{2.2.47}
\begin{split}
\partial_j\Ro^n\rho\psi^{(a)}(\xi)=(-1)^j\Ri_1(\xi,\tau_j)\Ro^n\rho\psi^{(a)}(\tau_j)+(-1)^j\sum_{k=0}^{n-1}\bigg(\Ri_{n-k+1}(\xi,\tau_j)+\Ri_{n-k}(\xi,\tau_j)\bigg)\Ro^k\rho\psi^{(a)}(\tau_j)\,,
\end{split}
\end{flalign}
where Property \ref{t2.2.21} is used again.
\end{proof}

\begin{lem}\label{t2.2.32}
The commutator of $\Po$ and $\Ro^n\rho$, $\forall n\in\Z_{\geq0}$, can be written as
\begin{flalign}\label{2.2.48}
\begin{split}
[\Po,\Ro^n\rho]=-\frac{\gamma}{\dot{u}_0}\rho\Bigg(\Gamma\Ro^n+\sum_{k=0}^{n-1}\Big(\Ro^{n-k}+\Ro^{n-k-1}\Big)\Gamma\Ro^k\Bigg)\rho-\sum_{j=1}^{2m}\partial_j\Ro^n\rho\,.
\end{split}
\end{flalign}
\end{lem}

\begin{proof}
Using Lemma \ref{t2.2.22}, Lemmas \ref{t2.2.11} and \ref{t2.2.24} initiate the induction for $n=0$. So that, assuming that this relation holds for $n-1$, the result is yielded by
\begin{flalign}\label{2.2.49}
\begin{split}
[\Po,\Ro^n\rho]&=\Ro[\Po,\Ro^{n-1}\rho]+[\Po,\Ro]\Ro^{n-1}\rho \\
&=-\frac{\gamma}{\dot{u}_0}\rho\Bigg(\Ro\Gamma\Ro^{n-1}+\sum_{k=0}^{n-2}\Big(\Ro^{n-k}+\Ro^{n-k-1}\Big)\Gamma\Ro^k\Bigg)\rho-\sum_{j=1}^{2m}\Ro\partial_j\Ro^{n-1}\rho-\frac{\gamma}{\dot{u}_0}\rho\Gamma\rho\Ro^{n-1}\rho-\sum_{j=1}^{2m}\Ro\dl{j}\rho\Ro^{n-1}\rho \\
&=-\frac{\gamma}{\dot{u}_0}\rho\Bigg(\Gamma\Ro^{n}+\sum_{k=0}^{n-1}\Big(\Ro^{n-k}+\Ro^{n-k-1}\Big)\Gamma\Ro^k\Bigg)\rho-\sum_{j=1}^{2m}\bigg(\Ro\partial_j\Ro^{n-1}\rho+\big(\partial_j\Ro\big)\Ro^{n-1}\rho\bigg)\,,
\end{split}
\end{flalign}
where we used again Lemmas \ref{t2.2.11} and \ref{t2.2.24} to get to the second line, whilst Lemma \ref{t2.2.22} and Property \ref{t2.2.9} result in the third line.
\end{proof}

\begin{prop}\label{t2.2.33}
The AWF obey, $\forall n\in\Z_{\geq0}$,
\begin{flalign}\label{2.2.50}
\begin{split}
\partial_\xi\af{n}{\alpha}(\xi)=\af{n}{\alpha+1}(\xi)-\frac{\gamma}{\dot{u}_0}\Bigg(&\au{0}{\alpha}\af{n}{0}(\xi)+\sum_{k=0}^{n-1}\big(\au{n-k}{\alpha}+\au{n-k-1}{\alpha}\big)\af{k}{0}(\xi)\Bigg) \\
&-\frac{\gamma\ddot{u}_0}{\dot{u}_0^2}\bigg(n\af{n}{\alpha}(\xi)+(n+1)\af{n+1}{\alpha}(\xi)\bigg)-\sum_{j=1}^{2m}\partial_j\af{n}{\alpha}(\xi)\,.
\end{split}
\end{flalign}
\end{prop}

\begin{proof}
Firstly we observe that
\begin{flalign}\label{2.2.51}
\begin{split}
\Gamma\Ro^{n}+\sum_{k=0}^{n-1}\Big(\Ro^{n-k}+\Ro^{n-k-1}\Big)\Gamma\Ro^k=\Gamma\Ro^{n}+\sum_{k=1}^{n}\Ro^k\Gamma\Ro^{n-k}+\sum_{k=0}^{n-1}\Ro^k\Gamma\Ro^{n-k-1}=\Ro^n\Gamma+\sum_{k=0}^{n-1}\Ro^k\Gamma\Big(\Ro^{n-k}+\Ro^{n-k-1}\Big)\,.
\end{split}
\end{flalign}
Furthermore we notice $\Ro^n\Ro+\sum_{k=0}^{n-1}\Ro^k\Ro\Big(\Ro^{n-k}+\Ro^{n-k-1}\Big)=n\Ro^n+(n+1)\Ro^{n+1}$, and, using $\rho\Gamma=\rho\T_{\psi}\Pi_{\I}+\frac{\ddot{u}_0}{\dot{u}_0}\Ro$, this leads to
\begin{flalign}\label{2.2.52}
\begin{split}
\partial_\xi\af{n}{\alpha}(\xi)&=\Po\Ro^n\rho\Po^\alpha\psi(\xi)=\af{n}{\alpha+1}(\xi)+[\Po,\Ro^n\rho]\psi^{(\alpha)}(\xi) \\
&=\af{n}{\alpha+1}(\xi)-\frac{\gamma}{\dot{u}_0}\rho\Bigg(\Ro^n\Gamma+\sum_{k=0}^{n-1}\Ro^k\Gamma\Big(\Ro^{n-k}+\Ro^{n-k-1}\Big)\Bigg)\rho\psi^{(\alpha)}(\xi)-\sum_{j=1}^{2m}\partial_j\Ro^n\rho\psi^{(\alpha)}(\xi) \\
&=\af{n}{\alpha+1}(\xi)-\frac{\gamma}{\dot{u}_0}\Bigg(\au{n}{\alpha}\Ro^n\rho\psi(\xi)+\sum_{k=0}^{n-1}\Big(\au{n-k}{\alpha}+\au{n-k-1}{\alpha}\Big)\Ro^k\rho\psi(\xi)\Bigg) \\
&\,\,\,\,\,\,\,\,\,\,\,\,\,\,\,\,\,\,\,\,\,\,\,\,\,\,\,\,\,\,\,\,-\frac{\gamma\ddot{u}_0}{\dot{u}_0^2}\bigg(n\Ro^n\rho\psi^{(\alpha)}(\xi)+(n+1)\Ro^{n+1}\rho\psi^{(\alpha)}(\xi)\bigg)-\sum_{j=1}^{2m}\partial_j\Ro^n\rho\psi^{(\alpha)}(\xi)\,,
\end{split}
\end{flalign}
which is a consequence of Lemma \ref{t2.2.32}.
\end{proof}

\subsubsection{Closure Relations}\label{s223}

We now demonstrate that $\af{n}{2}(\xi)$ is related to $\af{n}{0}(\xi)$ and $\af{n}{1}(\xi)$ through Lemma \ref{t2.1.9}, $\forall n\in\Z_{\geq0}$.

\begin{lem}\label{t2.2.34}
The zeroth order AWFs are related by
\begin{flalign}\label{2.2.53}
\begin{split}
\af{0}{2}(\xi)=\frac{\gamma^2}{\dot{u}_0^2}\big(v_0+\xi\big)\af{0}{0}(\xi)-\frac{\gamma\ddot{u}_0}{\dot{u}_0^2}\af{0}{1}(\xi)+\frac{\gamma}{\dot{u}_0}\Big(\au{0}{0}\af{0}{1}(\xi)-\au{0}{1}\af{0}{0}(\xi)\Big)\,.
\end{split}
\end{flalign}
\end{lem}

\begin{proof}
We rewrite Lemma \ref{t2.1.9} as
\begin{flalign}\label{2.2.54}
\begin{split}
\Po^2\psi=\frac{\gamma^2}{\dot{u}_0^2}\bigg(v_0+\Qo-\frac{\ddot{u}_0}{\gamma}\Po\bigg)\psi\,,
\end{split}
\end{flalign}
and then Lemmas \ref{t2.2.10} and \ref{t2.2.11} entail
\begin{flalign}\label{2.2.55}
\begin{split}
\af{0}{2}(\xi)&=\frac{\gamma^2}{\dot{u}_0^2}\rho\bigg(v_0+\Qo-\frac{\ddot{u}_0}{\gamma}\Po\bigg)\psi(\xi)=\frac{\gamma^2}{\dot{u}_0^2}\big(v_0+\xi\big)\af{0}{0}(\xi)-\frac{\gamma\ddot{u}_0}{\dot{u}_0^2}\af{0}{1}(\xi)-\frac{\gamma^2}{\dot{u}_0^2}[\Qo,\rho]\psi(\xi) \\
&=\frac{\gamma^2}{\dot{u}_0^2}\big(v_0+\xi\big)\af{0}{0}(\xi)-\frac{\gamma\ddot{u}_0}{\dot{u}_0^2}\af{0}{1}(\xi)-\frac{\gamma}{\dot{u}_0}\bigg(\ipr{\Po\psi}{\Pi_{\I}\rho\psi}\af{0}{0}(\xi)-\ipr{\psi}{\Pi_{\I}\rho\psi}\af{0}{1}(\xi)\bigg)\,,
\end{split}
\end{flalign}
completing the proof.
\end{proof}

It remains to generalize this result for any AWF.

\begin{lem}\label{t2.2.35}
One may express the commutator of $\Qo$ and $\Ro^n\rho$ for any $n\in\Z_{\geq0}$ as follows,
\begin{flalign}\label{2.2.56}
\begin{split}
[\Ro^n\rho,\Qo]=\frac{\dot{u}_0}{\gamma}\rho\Bigg(\big(\T_{\psi',\psi}-\T_{\psi,\psi'}\big)\Pi_{\I}\Ro^n+\sum_{k=0}^{n-1}\big(\Ro^{n-k}+\Ro^{n-k-1}\big)\big(\T_{\psi',\psi}-\T_{\psi,\psi'}\big)\Pi_{\I}\Ro^k\Bigg)\rho\,.
\end{split}
\end{flalign}
\end{lem}

\begin{proof}
For $n=0$ this result is ensured by Lemmas \ref{t2.2.10} and \ref{t2.2.11} together. And then, assuming\eref{2.2.56} holds for $n-1$, we evaluate
\begin{flalign}\label{2.2.57}
\begin{split}
[\Ro^n\rho,\Qo]=\frac{\dot{u}_0}{\gamma}\rho\Bigg(\Ro\big(\T_{\psi',\psi}-\T_{\psi,\psi'}\big)\Pi_{\I}\Ro^{n-1}+\sum_{k=0}^{n-2}\big(\Ro^{n-k}&+\Ro^{n-k-1}\big)\big(\T_{\psi',\psi}-\T_{\psi,\psi'}\big)\Pi_{\I}\Ro^k\Bigg)\rho \\
&+\frac{\dot{u}_0}{\gamma}\rho\big(\T_{\psi',\psi}-\T_{\psi,\psi'}\big)\Pi_{\I}\big(\id+\Ro\big)\Ro^{n-1}\rho\,.
\end{split}
\end{flalign}
This argument is highly similar to the proof of Lemma \ref{t2.2.32} to which we refer for more detailed intermediate steps.
\end{proof}

\begin{prop}\label{t2.2.36}
Defining for the sake of conciseness $\au{-p}{a}\coloneqq 0$, $\forall p\in\N$, one may show that the following relation is satisfied $\forall n\in\Z_{\geq0}\,$,
\begin{flalign}\label{2.2.58}
\begin{split}
\af{n}{2}(\xi)=\frac{\gamma^2}{\dot{u}_0^2}\big(v_0+\xi\big)\af{n}{0}(\xi)-\frac{\gamma\ddot{u}_0}{\dot{u}_0^2}\af{n}{1}(\xi)+\frac{\gamma}{\dot{u}_0}\sum_{k=0}^n\bigg(\big(\au{n-k}{0}+\au{n-k-1}{0}\big)\af{k}{1}(\xi)-\big(\au{n-k}{1}+\au{n-k-1}{1}\big)\af{k}{0}(\xi)\bigg)\,.
\end{split}
\end{flalign}
\end{prop}

\begin{proof}
Similarly to\eref{2.2.51}, we have
\begin{flalign}\label{2.2.59}
\begin{split}
\big(\T_{\psi',\psi}-\T_{\psi,\psi'}\big)\Pi_{\I}\Ro^n+\sum_{k=0}^{n-1}\big(\Ro^{n-k}+&\Ro^{n-k-1}\big)\big(\T_{\psi',\psi}-\T_{\psi,\psi'}\big)\Pi_{\I}\Ro^k \\
&=\Ro^n\big(\T_{\psi',\psi}-\T_{\psi,\psi'}\big)\Pi_{\I}+\sum_{k=0}^{n-1}\Ro^k\big(\T_{\psi',\psi}-\T_{\psi,\psi'}\big)\Pi_{\I}\big(\Ro^{n-k}+\Ro^{n-k-1}\big)\,,
\end{split}
\end{flalign}
which is yielding
\begin{flalign}\label{2.2.60}
\begin{split}
\af{n}{2}(\xi)&=\frac{\gamma^2}{\dot{u}_0^2}\Ro^n\rho\bigg(v_0+\Qo-\frac{\ddot{u}_0}{\gamma}\Po\bigg)\psi(\xi)=\frac{\gamma^2}{\dot{u}_0^2}\big(v_0+\xi\big)\af{n}{0}(\xi)-\frac{\gamma\ddot{u}_0}{\dot{u}_0^2}\af{n}{1}(\xi)+\frac{\gamma^2}{\dot{u}_0^2}[\Ro^n\rho,\Qo]\psi(\xi) \\
&=\frac{\gamma^2}{\dot{u}_0^2}\big(v_0+\xi\big)\af{n}{0}(\xi)-\frac{\gamma\ddot{u}_0}{\dot{u}_0^2}\af{n}{1}(\xi) \\
&\,\,\,\,\,\,\,\,\,\,\,\,\,\,\,\,\,\,\,\,\,\,\,\,\,\,\,\,\,\,+\frac{\gamma}{\dot{u}_0}\Bigg(\au{0}{0}\af{n}{1}(\xi)-\au{0}{1}\af{n}{0}(\xi)+\sum_{k=0}^{n-1}\Big((\au{n-k}{0}+\au{n-k-1}{0})\af{k}{1}(\xi)-(\au{n-k}{1}+\au{n-k-1}{1})\af{k}{0}(\xi)\Big)\Bigg)\,,
\end{split}
\end{flalign}
where\eref{2.2.54} and Lemma \ref{t2.2.35} have been used.
\end{proof}

\section{Associated Hamiltonian System and Lax Formalism}\label{s3}

\subsection{Hamiltonian Formulation}\label{s31}

We now relate the diagonal values of the resolvent kernel to the Fredholm determinant, and then discuss the Hamiltonian formulations of the dynamics of the AWF that the $n^{\operatorname{th}}$ order resolvent kernel gives rise to, in particular we ask whether or not the Hamiltonians commute in a Poisson brackets sense. Whilst this discussion aims solely at providing insights, it turns out that it will require technical results which shall play an important role in section \sref{s4}. Throughout this section we still consider a generic interval\eref{2.2.6}, and we define the main quantity of interest as follows.

\begin{defi}\label{t3.1.1}
We refer to $\Ha_n(\tau_j)\coloneqq \frac{\gamma}{\dot{u}_0}\Ri_n(\tau_j,\tau_j)$ as the $n^{\operatorname{th}}$ order Hamiltonian, for any $n\in\Z_{\geq0}$.
\end{defi}

Since $\Ri_n$ is the kernel associated with the $\nth$ power of the resolvent, we would like to point out an eventual analogy with the Lax framework, in which one may construct the commuting Hamiltonians as $\Ha_n=\frac{1}{n}\tra\operatorname{L}^n$, where their conservation is ensured by the fact that $\operatorname{L}$ satisfies a Lax equation. Although, strictly speaking, the Hamiltonians of Definition \ref{t3.1.1} do not commute with each other.

\subsubsection{Hamilton's Equations for the Auxiliary Wave Functions}\label{s311}

\begin{prop}\label{t3.1.2}
For any $j\in\{1,...,2m\}\subset\N$, the Fredholm determinant is related to the first order Hamiltonian by
\begin{flalign}\label{3.1.1}
\begin{split}
\frac{d}{d\tj}\lo\F(\I)=\frac{\dot{u}_0}{\gamma}\Ha_1(\tj)\,.
\end{split}
\end{flalign}
\end{prop}

\begin{proof}
We indeed compute
\begin{flalign}\label{3.1.2}
\begin{split}
\frac{d}{d\tj}\lo\de\big(\id-\Ko\big)_{\leb^2(\R)}&=\tra\bigg(\frac{d}{d\tj}\lo\big(\id-\Ko\big)\bigg)_{\leb^2(\R)}=\tra\bigg(\big(\id-\Ko\big)^{-1}\Ko_0\dl{j}\bigg)_{\leb^2(\R)} \\
&=\tra\bigg(\Ro\dl{j}\bigg)_{\leb^2(\R)}=\int_{\I}\Ri_1(\xi,\xi)\dl{j}\,d\xi=\Ri_1(\tj,\tj)=\frac{\dot{u}_0}{\gamma}\Ha_1(\tj)\,,
\end{split}
\end{flalign}
where we used\eref{2.2.32} and Property \ref{t2.2.21}.
\end{proof}

\begin{rem}\label{t3.1.3}
This result will play a crucial role in the derivation of an explicit expression for the Fredholm determinant as we shall see in section \sref{s4}. However for the moment we focus on the intrinsic properties of the $\Ha_n$, in particular on the Hamiltonian structure that they provide.
\end{rem}

\begin{lem}\label{t3.1.4}
The $\nth$ order Hamiltonian admits the following expression,
\begin{flalign}\label{3.1.3}
\begin{split}
\Ha_n(\tj)=\sum_{k=1}^n\Big(\af{n-k}{0}'(\tj)\af{k-1}{1}(\tj)-\af{k-1}{1}'(\tj)\af{n-k}{0}(\tj)\Big)\,,
\end{split}
\end{flalign}
where $\af{n}{a}^{(l)}(\tj)\coloneqq \frac{d^l}{d\tj^l}\af{n}{a}(\tj)$.
\end{lem}

\begin{proof}
Firstly, since $\af{n}{a}'(\tj)=\big(\partial_\xi+\partial_j\big)\big|_{\xi=\tj}\af{n}{a}(\xi)$, Proposition \ref{t2.2.20} implies
\begin{flalign}\label{3.1.4}
\begin{split}
\sum_{k=1}^n\Big(\af{n-k}{0}'\af{k-1}{1}-\af{k-1}{1}'\af{n-k}{0}\Big)=\Ha_n(\tj)+\sum_{k=1}^n\Big(\af{k-1}{1}\partial_j\big|_{\xi=\tj}\af{n-k}{0}(\xi)-\af{n-k}{0}\partial_j\big|_{\xi=\tj}\af{k-1}{1}(\xi)\Big)\,.
\end{split}
\end{flalign}
Where $\af{n}{a}^{(l)}\coloneqq \af{n}{a}^{(l)}(\tj)$. And then Proposition \ref{t2.2.31} entails
\begin{flalign}\label{3.1.5}
\begin{split}
\sum_{k=1}^n&\Big(\af{k-1}{1}\partial_j\big|_{\xi=\tj}\af{n-k}{0}(\xi)-\af{n-k}{0}\partial_j\big|_{\xi=\tj}\af{k-1}{1}(\xi)\Big) \\
&=\sum_{k=1}^n(-1)^j\Bigg(\sum_{l=0}^{n-k-1}\Big(\Ri_{n-k-l+1}(\tj,\tj)+\Ri_{n-k-l}(\tj,\tj)\Big)\af{l}{0}\af{k-1}{1}-\sum_{l=0}^{k-2}\Big(\Ri_{k-l}(\tj,\tj)+\Ri_{k-l-1}(\tj,\tj)\Big)\af{n-k}{0}\af{l}{1}\Bigg) \\
&=-\sum_{k=1}^n\sum_{l=0}^{k-2}\Big(\Ri_{k-l}(\tj,\tj)+\Ri_{k-l-1}(\tj,\tj)\Big)\Big(\af{n-k}{0}\af{l}{1}-\af{l}{0}\af{n-k}{1}\Big) \\
&=-\sum_{k=2}^n\sum_{l'=2}^k\Big(\Ri_{l'}(\tj,\tj)+\Ri_{l'-1}(\tj,\tj)\Big)\Big(\af{n-k}{0}\af{j-l'}{1}-\af{j-l'}{0}\af{n-k}{1}\Big)\,.
\end{split}
\end{flalign}
One notices that the third line is obtained using Remark \ref{t2.2.16}, whilst the fourth one is reached with $l'=k-l$. Besides, defining for any $k\in\Z_{\geq2}$, $\Ns{2}{k}\coloneqq\{2,...,k\}\subset\N$, we may observe that, for any $F$,
\begin{flalign}\label{3.1.6}
\begin{split}
\sum_{k=2}^n\sum_{l=2}^kF(l,k,n)=\sum_{l=2}^n\sum_{k=2}^nF(l,k,n)\ind{\Ns{2}{k}}(l)=\sum_{l=2}^n\sum_{k=l}^nF(l,k,n)\,.
\end{split}
\end{flalign}
And therefore
\begin{flalign}\label{3.1.7}
\begin{split}
\sum_{k=1}^n\Big(\af{k-1}{1}\partial_j\big|_{\xi=\tj}\af{n-k}{0}(\xi)-&\af{n-k}{0}\partial_j\big|_{\xi=\tj}\af{k-1}{1}(\xi)\Big) \\
&=-\sum_{l=2}^n\Big(\Ri_{l}(\tj,\tj)+\Ri_{l-1}(\tj,\tj)\Big)\sum_{k=l}^n\Big(\af{n-k}{0}\af{j-l}{1}-\af{j-l}{0}\af{n-k}{1}\Big)=0\,,
\end{split}
\end{flalign}
yielding the result when injected in\eref{3.1.4}.
\end{proof}

\begin{cor}\label{t3.1.5}
The AWF satisfy the following Hamilton's equations, $\forall n\in\N$ and $\forall k\in\{1,...,n\}\subset\N$,
\begin{flalign}\label{3.1.8}
\begin{split}
\af{n-k}{0}'=\frac{\partial\Ha_n}{\partial\af{k-1}{1}}\,,\,\,\,\,\,\,\,\,\,\,\,\,\,\,\af{k-1}{1}'=-\frac{\partial\Ha_n}{\partial\af{n-k}{0}}\,.
\end{split}
\end{flalign}
\end{cor}

Since each Hamiltonian has a different set of associated canonical variables, evaluating the Poisson brackets of the Hamiltonians will demand to relate the AWF order by order. In turn this shall be highly relevant in \sref{s4}. And in order to go any further, we need the following condition to be satisfied.

\begin{cond}\label{t3.1.6}
We ask that the ratio $\frac{\af{0}{0}(\tj)}{\psi(\tj)}$ is well defined $\forall j\in\{1,...,2m\}\subset\N$.
\end{cond}

\begin{defi}\label{t3.1.7}
We define $\ipro{n}{a}\coloneqq\ipr{\psi}{\Pi_{\operatorname{J}}\Ro^n\rho\Po^a\psi}\in\R$, $\forall n\in\Z_{\geq0}$ and $\operatorname{J}\coloneqq\bigsqcup_{j=1}^m(\tau_{2j-1};\tau_{2j})\subset\I$.
\end{defi}

\begin{lem}\label{t3.1.8}
We have, $\forall n\in\Z_{\geq0}\,$,
\begin{flalign}\label{3.1.9}
\begin{split}
\ipro{n}{a}=\int_{\operatorname{J}}\,\Bigg((-1)^n\psi(\sigma)\af{0}{a}(\sigma)-\sum_{k=0}^{n-1}(-1)^{n-k}\af{k}{a}(\sigma)\af{0}{0}(\sigma)\Bigg)d\sigma\,.
\end{split}
\end{flalign}
\end{lem}

\begin{proof}
We proceed by induction on $n$, with the initiation for $n=0$ being direct, we assume that this relation stands for $n-1$ in order to derive
\begin{flalign}\label{3.1.10}
\begin{split}
\ipro{n}{a}&=\ipr{\psi}{\Pi_{\operatorname{J}}\big(\rho-\id\big)\Ro^{n-1}\rho\Po^a\psi}=\ipr{\psi}{\Pi_{\operatorname{J}}\rho\af{n-1}{a}}-\ipro{n-1}{a}=\int_{\operatorname{J}}\rho\psi(\sigma)\,\af{n-1}{a}(\sigma)\,d\sigma-\ipro{n-1}{a} \\
&=\int_{\operatorname{J}}\Bigg(\af{0}{0}(\sigma)\af{n-1}{a}(\sigma)+(-1)^n\psi(\sigma)\af{0}{a}(\sigma)-\sum_{k=0}^{n-2}(-1)^{n-k}\af{k}{a}(\sigma)\af{0}{0}(\sigma)\Bigg)d\sigma\,,
\end{split}
\end{flalign}
where Property \ref{t2.2.9} and Lemma \ref{t2.2.12} have been used.
\end{proof}

Henceforth we are able to relate the AWF order by order.

\begin{lem}\label{t3.1.9}
For any $n,p\in\Z_{\geq0}$, $n\geq p$, the $\nth$ order AWF are related to the $p^{\operatorname{th}}$ order AWF by
\begin{flalign}\label{3.1.11}
\begin{split}
\af{n}{a}(\tj)=\bigg(\frac{\af{0}{0}(\tj)}{\psi(\tj)}-1\bigg)^{n-p}\af{p}{a}(\tj)\,,\,\,\,\,\,\,\,\,\forall j\in\{1,...,2m\}\subset\N\,.
\end{split}
\end{flalign}
\end{lem}

\begin{proof}
Firstly one observes that, according to Definition \ref{t3.1.7},
\begin{flalign}\label{3.1.12}
\begin{split}
\ipro{n}{a}=\int_{\operatorname{J}}\psi(\sigma)\af{n}{a}(\sigma)\,d\sigma\,.
\end{split}
\end{flalign}
Then one notes that equating this with Lemma \ref{t3.1.8} and taking the $\tj$-derivative yields
\begin{flalign}\label{3.1.13}
\begin{split}
\psi(\tj)\af{n}{a}(\tj)=(-1)^n\psi(\tj)\af{0}{a}(\tj)-\sum_{k=0}^{n-1}(-1)^{n-k}\af{k}{a}(\tj)\af{0}{0}(\tj)\,.
\end{split}
\end{flalign}
Therefore, under Condition \ref{t3.1.6}, we obtain
\begin{flalign}\label{3.1.14}
\begin{split}
\af{n}{a}(\tj)&=(-1)^n\af{0}{a}(\tj)-\sum_{k=0}^{n-2}(-1)^{n-k}\frac{\af{0}{0}(\tj)}{\psi(\tj)}\af{k}{a}(\tj)+\frac{\af{0}{0}(\tj)}{\psi(\tj)}\af{n-1}{a}(\tj) \\
&=-\af{n-1}{a}(\tj)+\frac{\af{0}{0}(\tj)}{\psi(\tj)}\af{n-1}{a}(\tj)=\bigg(\frac{\af{0}{0}(\tj)}{\psi(\tj)}-1\bigg)\af{n-1}{a}(\tj)\,.
\end{split}
\end{flalign}
Finally, a direct iteration on $n$ leads to the result.
\end{proof}

\subsubsection{Associated Poisson Structure}\label{s312}

\begin{defi}\label{t3.1.10}
Let $n\in\N$. We shall refer to $\big(\qcv_{n,k},\pcv_k\big)\coloneqq\big(\af{n-k}{0},\af{k-1}{1}\big)$, $\forall k\in\{1,...,n\}\subset\N$, as the $\nth$ order canonical variables (associated to the $\nth$ order Hamiltonian).
\end{defi}

\begin{defi}\label{t3.1.11}
We define the $\nth$ order Poisson bracket (associated to the $\nth$ order Hamiltonian) as
\begin{flalign}\label{3.1.15}
\begin{split}
\pbn{f}{g}\coloneqq\sum_{k=1}^n\Bigg(\frac{\partial f}{\partial\qcv_{n,k}}\frac{\partial g}{\partial\pcv_k}-\frac{\partial f}{\partial\pcv_k}\frac{\partial g}{\partial\qcv_{n,k}}\Bigg)\,,
\end{split}
\end{flalign}
for any differentiable mappings $f$ and $g$ depending on the $\nth$ order canonical variables.
\end{defi}

\begin{nota}\label{t3.1.12}
Throughout \sref{s312}, we let $n,r\in\N\,,\,\,n>r\,,\,\,k\in\{1,...,r\}\subset\N$, and $d\coloneqq n-r\geq1$.
\end{nota}

\begin{pro}\label{t3.1.13}
The $r^{\operatorname{th}}$ and $\nth$ order canonical variables are related by
\begin{flalign}\label{3.1.16}
\begin{split}
\qcv_{r,k}=\af{n-(k+n-r)}{0}=\qcv_{n,k+d}\,.
\end{split}
\end{flalign}
\end{pro}

\begin{nota}\label{t3.1.14}
We write $\qcv\coloneqq\qcv_{n,n}=\af{0}{0}\,$ and $\,\pcv\coloneqq\pcv_1=\af{0}{1}$. Moreover, since this quantity will frequently appear, we also denote $\eta_n(\tj)\coloneqq\bigg(\frac{\qcv(\tj)}{\psi(\tj)}-1\bigg)^n$, $\forall n\in\Z_{\geq0}$.
\end{nota}

\begin{lem}\label{t3.1.15}
One may express the $\nth$ order Hamiltonian in terms of the first order Hamiltonian as follows,
\begin{flalign}\label{3.1.17}
\begin{split}
\Ha_n(\qcv,\pcv)=n\,\eta_{n-1}\Ha_1(\qcv,\pcv)\,,
\end{split}
\end{flalign}
where $\big(\qcv,\pcv,\eta_{n-1}\big)\coloneqq\big(\qcv(\tj),\pcv(\tj),\eta_{n-1}(\tj)\big)$, and this result holds for any $j\in\{1,...,2m\}\subset\N$.
\end{lem}

\begin{proof}
As before we denote the $l^{\operatorname{th}}$ total derivative with respect to $\tj$ by a superscript $(l)$. First we need to compute
\begin{flalign}\label{3.1.18}
\begin{split}
\eta_n'=\frac{n}{\psi}\bigg(\qcv'-\frac{\psi'}{\psi}\qcv\bigg)\,\eta_{n-1}\,,
\end{split}
\end{flalign}
then, observing that $\eta_n\eta_k=\eta_{n+k}$, Lemmas \ref{t3.1.4} and \ref{t3.1.9} are leading to
\begin{flalign}\label{3.1.19}
\begin{split}
\Ha_n(\qcv,\pcv)=\sum_{k=1}^n\Bigg(\big(\eta_{n-k}\qcv\big)'\eta_{k-1}\pcv-\big(\eta_{k-1}\pcv\big)'\eta_{n-k}\qcv\Bigg)=\sum_{k=1}^n\Bigg(\Big(\qcv'\pcv-\pcv'\qcv\Big)\fa{n-1}+\frac{n-2k+1}{\psi}\bigg(\qcv'\qcv\pcv-\frac{\psi'}{\psi}\qcv^2\pcv\bigg)\fa{n-2}\Bigg)\,.
\end{split}
\end{flalign}
Besides, we notice that $\sum_{k=1}^nk=n(n+1)/2$ entails $\sum_{k=1}^n(n+1-2k)=0$, and hence
\begin{flalign}\label{3.1.20}
\begin{split}
\Ha_n(\qcv,\pcv)=n\Big(\qcv'\pcv-\pcv'\qcv\Big)\fa{n-1}\,,
\end{split}
\end{flalign}
which completes the proof.
\end{proof}

\begin{nota}\label{t3.1.16}
We keep on writing $f\coloneqq f(\tj)$ for any function depending on $\tj$, and a superscript $(l)$ denotes the $l^{\operatorname{th}}$ total derivative with respect to $\tj$.
\end{nota}

\begin{prop}\label{t3.1.17}
Utilizing the conventions of Notation \ref{t3.1.12}, the Hamiltonians $\Ha_r$ and $\Ha_n$ commute with respect to the $\nth$ order Poisson bracket, namely
\begin{flalign}\label{3.1.21}
\begin{split}
\pbn{\Ha_r}{\Ha_n}=0\,.
\end{split}
\end{flalign}
\end{prop}

\begin{proof}
To begin with we shall express $\Ha_r$ in terms of the $\nth$ order canonical variables. And, according to Property \ref{t3.1.13}, Lemma \ref{t3.1.4} yields
\begin{flalign}\label{3.1.22}
\begin{split}
\Ha_r=\sum_{k=1}^r\Big(\qcv_{r,k}'\pcv_k-\pcv_k'\qcv_{r,k}\Big)=\sum_{k=1}^r\Big(\qcv_{n,k+d}'\pcv_k-\pcv_k'\qcv_{n,k+d}\Big)\,.
\end{split}
\end{flalign}
Henceforth we can compute
\begin{flalign}\label{3.1.23}
\begin{split}
\pbn{\Ha_r}{\Ha_n}=\sum_{k=1}^n\Bigg(\frac{\partial\Ha_r}{\partial\qcv_{n,k}}\qcv_{n,k}'+\frac{\partial\Ha_r}{\partial\pcv_k}\pcv_k'\Bigg)=\sum_{k=1-d}^r\frac{\partial\Ha_r}{\partial\qcv_{n,k+d}}\qcv_{n,k+d}'+\sum_{k=1}^r\qcv_{n,k+d}'\pcv_k'=\sum_{k=1}^r\bigg(\qcv_{n,k+d}'\pcv_k'-\pcv_k'\qcv_{n,k+d}'\bigg)\,,
\end{split}
\end{flalign}
hereby completing the argument.
\end{proof}

\begin{prop}\label{t3.1.18}
The first order Poisson bracket of $\Ha_1$ and $\Ha_n$ is given by
\begin{flalign}\label{3.1.24}
\begin{split}
\pbn[1]{\Ha_1}{\Ha_n}=\frac{n(n-1)}{\psi}\bigg(\frac{\qcv}{\psi}\big(\pcv'\partial_\pcv\psi+\qcv'\partial_\qcv\psi\big)-\qcv'\bigg)\fa{n-2}\Ha_1(\qcv,\pcv)\,.
\end{split}
\end{flalign}
\end{prop}

\begin{proof}
Lemma \ref{t3.1.15} enables us to write $\Ha_n$ only in terms of $\qcv$ and $\pcv$, but in\eref{3.1.17} one has to notice that $\psi$ is related to $\qcv$ and $\pcv$. To be more precise, from the definition of $\fa{n-1}$, Notation \ref{t3.1.16}, it follows that
\begin{flalign}\label{3.1.25}
\begin{split}
\partial_\qcv\fa{n-1}=\frac{n-1}{\psi}\bigg(1-\frac{\qcv}{\psi}\partial_\qcv\psi\bigg)\fa{n-2}\,,\,\,\,\,\,\,\,\,\,\,\partial_\pcv\fa{n-1}=-\frac{n-1}{\psi}\frac{\qcv}{\psi}\big(\partial_\pcv\psi\big)\fa{n-2}\,.
\end{split}
\end{flalign}
Injecting this in the $\qcv$ and $\pcv$ derivatives of\eref{3.1.17} leads to
\begin{flalign}\label{3.1.26}
\begin{split}
&\partial_\qcv\Ha_n(\qcv,\pcv)=\frac{n(n-1)}{\psi}\bigg(1-\frac{\qcv}{\psi}\partial_\qcv\psi\bigg)\fa{n-2}\Ha_1(\qcv,\pcv)-n\fa{n-1}\pcv'\,,\\
&\partial_\pcv\Ha_n(\qcv,\pcv)=-\frac{n(n-1)}{\psi}\frac{\qcv}{\psi}\big(\partial_\pcv\psi\big)\fa{n-2}\Ha_1(\qcv,\pcv)+n\fa{n-1}\qcv'\,.
\end{split}
\end{flalign}
And henceforth we are able to compute
\begin{flalign}\label{3.1.27}
\begin{split}
\pbn[1]{\Ha_1}{\Ha_n}&=-\pcv'\partial_\pcv\Ha_n(\qcv,\pcv)-\qcv'\partial_\qcv\Ha_n(\qcv,\pcv) \\
&=\frac{n(n-1)}{\psi}\frac{\qcv}{\psi}\big(\pcv'\partial_\pcv\psi\big)\fa{n-2}\Ha_1(\qcv,\pcv)+\frac{n(n-1)}{\psi}\bigg(\frac{\qcv}{\psi}\qcv'\partial_\qcv\psi-\qcv'\bigg)\fa{n-2}\Ha_1(\qcv,\pcv)\,,
\end{split}
\end{flalign}
which is indeed leading to\eref{3.1.24}.
\end{proof}

Hence, recalling that we set $n>r$, it turns out that $\Ha_n$ and $\Ha_r$ commute with respect to $\pbn{\cdot}{\cdot}$, but in general they do not commute with respect to $\pbn[r]{\cdot}{\cdot}$.

\subsection{Schlesinger Equations}\label{s32}

Turning to the Lax formulation, we provide evidence that the dynamics of the AWF are determined by a set of infinite size matrices which formally satisfy a Schlesinger system of equations. Once again the aim of this discussion is to provide insights, nonetheless we shall derive technical results following in particular from the closure relations which will turn out to be useful in section \sref{s4}. Throughout \sref{s32}, we let $\xi\in\I$.

\subsubsection{Infinite Size Lax Matrices}\label{s321}

\begin{defi}\label{t3.2.1}
Let $\Af\coloneqq\begin{pmatrix}\af{0}{0}&\af{0}{1}&\af{1}{0}&\af{1}{1}&\cdots\end{pmatrix}^t$, more precisely any component of $\Af$ is given by
\begin{flalign}\label{3.2.1}
\begin{split}
\compo{\Af}{2n+\alpha}(\xi)\coloneqq\af{n}{\alpha}(\xi)\,,
\end{split}
\end{flalign}
where $n\in\Z_{\geq0}$ and $\alpha\in\{0,1\}$ as was emphasized by Notation \ref{t2.2.3}.
\end{defi}

So, in \sref{s321}, we try to find infinite size matrices $\A_j$ and $\B_\xi$ such that, $\forall j\in\{1,...,2m\}\subset\N$,
\begin{flalign}\label{3.2.2}
\begin{split}
\partial_j\Af(\xi)=-\frac{\A_j}{\xi-\tj}\Af(\xi)\,,\,\,\,\,\,\,\,\,\,\,\partial_\xi\Af(\xi)=\Bigg(\B_\xi+\sum_{j=1}^{2m}\frac{\A_j}{\xi-\tj}\Bigg)\Af(\xi)\,.
\end{split}
\end{flalign}

\begin{nota}\label{t3.2.2}
For all $n\in\Z_{\geq0}$, let $\af{n}{-0}\coloneqq\af{n}{1}$ and $\af{n}{-1}\coloneqq\af{n}{0}$. Whilst $\forall p\in\N$, $\af{-p}{a}\coloneqq0$.
\end{nota}

\begin{prop}\label{t3.2.3}
For any $j\in\{1,...,2m\}\subset\N$, the infinite size matrix whose components are, $\forall n,p\in\Z_{\geq0}$,
\begin{flalign}\label{3.2.3}
\begin{split}
\comps{\A_j}{2n+\alpha}{2p+\beta}=-\udf(-1)^{j+\beta}\sum_{k=0}^{n-p}\big(\af{n-p-k}{-\beta}+\af{n-p-k-1}{-\beta}\big)\af{k}{\alpha}\,,
\end{split}
\end{flalign}
satisfies the first equality of\eref{3.2.2}, where we use Notations \ref{t3.1.16} and \ref{t3.2.2}.
\end{prop}

\begin{proof}
We let $j$ be any integer between $1$ and $2m$ and $n\in\Z_{\geq0}$. We begin with the following observation, Propositions \ref{t2.2.31} and \ref{t2.2.18} (in particular written as\eref{2.2.26}) yield
\begin{flalign}\label{3.2.4}
\begin{split}
\partial_j\compo{\Af}{2n+\alpha}(\xi)&=(-1)^j\Ri_1(\xi,\tau_j)\af{n}{a}+(-1)^j\sum_{k=0}^{n-1}\bigg(\Ri_{n-k+1}(\xi,\tau_j)+\Ri_{n-k}(\xi,\tau_j)\bigg)\af{k}{a} \\
&=\udf\frac{(-1)^j}{\xi-\tj}\Bigg[\big(\af{0}{0}(\xi)\af{0}{1}-\af{0}{1}(\xi)\af{0}{0}\big)\af{n}{\alpha}+\sum_{k=0}^{n-1}\Bigg(\sum_{l=0}^{n-k}\big(\af{n-k-l}{0}(\xi)\af{l}{1}-\af{n-k-l}{1}(\xi)\af{l}{0}\big)\af{k}{\alpha} \\
&\,\,\,\,\,\,\,\,\,\,\,\,\,\,\,\,\,\,\,\,\,\,\,\,\,\,\,\,\,\,\,\,\,\,\,\,\,\,\,\,\,\,\,\,\,\,\,\,\,\,\,\,\,\,\,\,\,\,\,\,\,\,\,\,\,\,\,\,\,\,\,\,\,\,\,\,\,\,\,\,\,\,\,\,\,\,\,\,\,\,\,\,\,\,\,\,\,\,\,\,\,\,\,\,\,\,\,\,\,\,\,\,\,\,\,\,\,\,\,\,\,\,\,\,+\sum_{l=0}^{n-k-1}\big(\af{n-k-l-1}{0}(\xi)\af{l}{1}-\af{n-k-l-1}{1}(\xi)\af{l}{0}\big)\af{k}{\alpha}\Bigg)\Bigg] \\
&=\udf\frac{(-1)^j}{\xi-\tj}\sum_{k=0}^{n}\Bigg(\sum_{l=0}^{n-k}\big(\af{l}{0}(\xi)\af{n-k-l}{1}-\af{l}{1}(\xi)\af{n-k-l}{0}\big)\af{k}{\alpha} \\
&\,\,\,\,\,\,\,\,\,\,\,\,\,\,\,\,\,\,\,\,\,\,\,\,\,\,\,\,\,\,\,\,\,\,\,\,\,\,\,\,\,\,\,\,\,\,\,\,\,\,\,\,\,\,\,\,\,\,\,\,\,\,\,\,\,\,\,\,\,\,\,\,\,\,\,\,\,\,\,\,\,\,\,\,\,\,\,\,\,\,\,\,\,\,\,\,\,\,\,\,\,\,\,\,\,\,\,\,\,\,\,\,\,\,\,\,\,\,\,\,\,\,\,\,\,\,\,+\sum_{l=0}^{n-k-1}\big(\af{l}{0}(\xi)\af{n-k-l-1}{1}-\af{l}{1}(\xi)\af{n-k-l-1}{0}\big)\af{k}{\alpha}\Bigg) \\
&=\udf\frac{(-1)^j}{\xi-\tj}\sum_{k=0}^{n}\sum_{l=0}^{n-k}\Bigg(\big(\af{l}{0}(\xi)\af{n-k-l}{1}-\af{l}{1}(\xi)\af{n-k-l}{0}\big)\af{k}{\alpha} \\
&\,\,\,\,\,\,\,\,\,\,\,\,\,\,\,\,\,\,\,\,\,\,\,\,\,\,\,\,\,\,\,\,\,\,\,\,\,\,\,\,\,\,\,\,\,\,\,\,\,\,\,\,\,\,\,\,\,\,\,\,\,\,\,\,\,\,\,\,\,\,\,\,\,\,\,\,\,\,\,\,\,\,\,\,\,\,\,\,\,\,\,\,\,\,\,\,\,\,\,\,\,\,\,\,\,\,\,\,\,\,\,\,\,\,\,\,\,\,\,\,\,\,\,\,\,\,\,\,\,\,\,\,\,\,\,\,\,\,\,+\big(\af{l}{0}(\xi)\af{n-k-l-1}{1}-\af{l}{1}(\xi)\af{n-k-l-1}{0}\big)\af{k}{\alpha}\Bigg)\,,
\end{split}
\end{flalign}
where Notation \ref{t3.2.2} is used in the last line. Whilst it follows from Definition \ref{t3.2.1} that
\begin{flalign}\label{3.2.5}
\begin{split}
\compo{\A_j\Af}{2n+\alpha}(\xi)&=\sum_{p=0}^\infty\bigg(\comps{\A_j}{2n+\alpha}{2p}\compo{\Af}{2p}(\xi)+\comps{\A_j}{2n+\alpha}{2p+1}\compo{\Af}{2p+1}(\xi)\bigg) \\
&=\sum_{p=0}^\infty\bigg(\comps{\A_j}{2n+\alpha}{2p}\af{p}{0}(\xi)+\comps{\A_j}{2n+\alpha}{2p+1}\af{p}{1}(\xi)\bigg)\,.
\end{split}
\end{flalign}
Finally it remains to equate these two expressions. For the sake of clarity we achieve this considering two separate cases. Firstly we note that the highest order AWF which depends on $\xi$ that appears in\eref{3.2.4} is $\af{n}{\alpha}(\xi)$, and therefore we obtain $\comps{\A_j}{2n+\alpha}{2p+\beta}\big|_{p>n}=0$, which is indeed consistent with\eref{3.2.3}. Secondly we begin with a procedure similar to\eref{3.1.6}, defining $\Ns{a}{b}\coloneqq\{a,...,b\}\subset\Z$ we notice that, for any $F$,
\begin{flalign}\label{3.2.6}
\begin{split}
\sum_{k=0}^{n}\sum_{l=0}^{n-k}F(l,k,n)=\sum_{l=0}^{n}\sum_{k=0}^{n}F(l,k,n)\ind{\Ns{0}{n-k}}(l)=\sum_{l=0}^{n}\sum_{k=0}^{n}F(l,k,n)\ind{\Ns{k}{n}}(n-l)=\sum_{l=0}^{n}\sum_{k=0}^{n-l}F(l,k,n)\,,
\end{split}
\end{flalign}
enabling us to rewrite\eref{3.2.4} as
\begin{flalign}\label{3.2.7}
\begin{split}
\partial_j\compo{\Af}{2n+\alpha}(\xi)=\udf\frac{(-1)^j}{\xi-\tj}\sum_{l=0}^{n}\sum_{k=0}^{n-l}\bigg(\af{l}{0}(\xi)\big(\af{n-k-l}{1}+\af{n-k-l-1}{1}\big)\af{k}{\alpha}-\af{l}{1}(\xi)\big(\af{n-k-l}{0}+\af{n-k-l-1}{0}\big)\af{k}{\alpha}\bigg)\,.
\end{split}
\end{flalign}
And a comparison with\eref{3.2.5} accordingly to the first equality of\eref{3.2.2} allows us to identify
\begin{flalign}\label{3.2.8}
\begin{split}
&\comps{\A_j}{2n+\alpha}{2p}\big|_{p\leq n}=-\udf(-1)^j\sum_{k=0}^{n-l}\big(\af{n-k-l}{1}+\af{n-k-l-1}{1}\big)\af{k}{\alpha}\,, \\
&\comps{\A_j}{2n+\alpha}{2p+1}\big|_{p\leq n}=\udf(-1)^j\sum_{k=0}^{n-l}\big(\af{n-k-l}{0}+\af{n-k-l-1}{0}\big)\af{k}{\alpha}\,.
\end{split}
\end{flalign}
So that, according to Notation \ref{t3.2.2}, all the components indeed agree with\eref{3.2.3}.
\end{proof}

\begin{rem}\label{t3.2.4}
According to Proposition \ref{t3.2.3}, the first few components of $\A_j$ are
\begin{flalign}\label{3.2.9}
\begin{split}
\A_j=-\udf(-1)^j\begin{pmatrix}\af{0}{1}\af{0}{0}&-\af{0}{0}^2&0&0&0&\cdots\\\af{0}{1}^2&-\af{0}{1}\af{0}{0}&0&0&0&\cdots\\\af{0}{1}\af{0}{0}+\af{0}{1}\af{1}{0}+\af{0}{0}\af{1}{1}&-\af{0}{0}^2-2\af{0}{0}\af{1}{0}&\af{0}{1}\af{0}{0}&-\af{0}{0}^2&0&\cdots\\\af{0}{1}^2+2\af{0}{1}\af{1}{1}&-\af{0}{1}\af{0}{0}-\af{0}{1}\af{1}{0}-\af{0}{0}\af{1}{1}&\af{0}{1}^2&-\af{0}{1}\af{0}{0}&0&\cdots\\\vdots&\vdots&\vdots&\vdots&\vdots&\ddots\end{pmatrix}\,.
\end{split}
\end{flalign}
One may notice that this matrix has a $2\times2$ block lower triangular structure. And one may further observe that this matrix is formally traceless, which would coincide with \cite{1}, with the diagonal blocks repeating themselves.
\end{rem}

The computation of $\B_\xi$ is more involved, and it requires the following Lemma which will in turn play a central role in \sref{s4}. We begin with a notation allowing for the shortening of the incoming computations.

\begin{nota}\label{t3.2.5}
We write $\an{n}{a}\coloneqq\au{n}{a}+\au{n-1}{a}$, for any $n\in\Z_{\geq0}$.
\end{nota}

In terms of which we recall Proposition \ref{t2.2.33},
\begin{flalign}\label{3.2.10}
\begin{split}
\partial_\xi\af{n}{\alpha}(\xi)=\af{n}{\alpha+1}(\xi)-\frac{\gamma}{\dot{u}_0^2}\Bigg(\ud\au{0}{\alpha}\af{n}{0}(\xi)+\udd n\af{n}{\alpha}(\xi)&+\udd(n+1)\af{n+1}{\alpha}(\xi)\Bigg) \\
&-\frac{\gamma}{\dot{u}_0}\sum_{k=0}^{n-1}\an{n-k}{\alpha}\af{k}{0}(\xi)-\sum_{j=1}^{2m}\partial_j\af{n}{\alpha}(\xi)\,.
\end{split}
\end{flalign}

\begin{lem}\label{t3.2.6}
The $\xi$-derivative of $\af{n}{1}(\xi)$, $\forall n\in\Z_{\geq0}$, can alternatively be written as
\begin{flalign}\label{3.2.11}
\begin{split}
\partial_\xi\af{n}{1}(\xi)=\frac{\gamma}{\ud^2}\Bigg(\Big(\gamma\big(v_0+\xi\big)-2\ud\au{0}{1}\Big)\af{n}{0}(\xi)+&\Big(\ud\au{0}{0}-\udd(n+1)\Big)\af{n}{1}(\xi)-\udd(n+1)\af{n+1}{1}(\xi)\Bigg) \\
&+\frac{\gamma}{\ud}\sum_{k=0}^{n-1}\Bigg(\an{n-k}{0}\af{k}{1}(\xi)-2\an{n-k}{1}\af{k}{0}(\xi)\Bigg)-\sum_{j=1}^{2m}\partial_j\af{n}{1}(\xi)\,.
\end{split}
\end{flalign}
\end{lem}

\begin{proof}
Recall that, using Notation \ref{t3.2.5}, Proposition \ref{t2.2.36} provides us with the following closure relation, $\forall n\in\Z_{\geq0}$,
\begin{flalign}\label{3.2.12}
\begin{split}
\af{n}{2}(\xi)&=\frac{\gamma}{\dot{u}_0^2}\Bigg(\gamma\big(v_0+\xi\big)\af{n}{0}(\xi)-\ddot{u}_0\af{n}{1}(\xi)\Bigg)+\frac{\gamma}{\dot{u}_0}\sum_{k=0}^n\Bigg(\an{n-k}{0}\af{k}{1}(\xi)-\an{n-k}{1}\af{k}{0}(\xi)\Bigg) \\
&=\frac{\gamma}{\dot{u}_0^2}\Bigg(\gamma\big(v_0+\xi\big)\af{n}{0}(\xi)-\ddot{u}_0\af{n}{1}(\xi)\Bigg)+\frac{\gamma}{\dot{u}_0}\sum_{k=0}^{n-1}\Bigg(\an{n-k}{0}\af{k}{1}(\xi)-\an{n-k}{1}\af{k}{0}(\xi)\Bigg) \\
&\,\,\,\,\,\,\,\,\,\,\,\,\,\,\,\,\,\,\,\,\,\,\,\,\,\,\,\,\,\,\,\,\,\,\,\,\,\,\,\,\,\,\,\,\,\,\,\,\,\,\,\,\,\,\,\,\,\,\,\,\,\,\,\,\,\,\,\,\,\,\,\,\,\,\,\,\,\,\,\,\,\,\,\,\,\,\,\,\,\,\,\,\,\,\,\,\,\,\,\,\,\,\,+\frac{\gamma}{\dot{u}_0^2}\Bigg(\ud\au{0}{0}\af{n}{1}(\xi)-\ud\au{0}{1}\af{n}{0}(\xi)\Bigg)\,.
\end{split}
\end{flalign}
And then injecting this into\eref{3.2.10} for $\alpha=1$ yields the result.
\end{proof}

Henceforth we can evaluate $\B_\xi$, but there is no formula for a generic component as\eref{3.2.3} for $\A_j$, instead we end up with twelve relations which completely determine the matrix.

\begin{prop}\label{t3.2.7}
The infinite size matrix whose components are, $\forall n,p\in\Z_{\geq0}$,
\begin{flalign}\label{3.2.13}
\begin{split}
&\comps{\B}{2n}{2p}\big|_{p>n+1}=0\,,\,\,\,\,\,\,\,\,\comps{\B}{2n}{2p+1}=\delta_{n,p}\,,\,\,\,\,\,\,\,\,\comps{\B}{2n}{2n}=\frac{\gamma}{\ud^2}\Big(-\ud\au{0}{0}-\udd n\Big)\,,\,\,\,\,\,\,\,\,\comps{\B}{2n}{2(n+1)}=\frac{\gamma}{\ud^2}\Big(-\udd(n+1)\Big)\,, \\
&\comps{\B}{2n}{2p}=\frac{\gamma}{\ud^2}\Big(-\ud\an{n-p}{0}\Big)\,,\,\,\,\,\,\,\,\,\comps{\B}{2n+1}{2p}\big|_{p>n}=0\,,\,\,\,\,\,\,\,\,\comps{\B}{2n+1}{2p+1}\big|_{p>n+1}=0\,, \\
&\comps{\B}{2n+1}{2n}=\frac{\gamma}{\ud^2}\Big(\gamma\big(v_0+\xi\big)-2\ud\au{0}{1}\Big)\,,\,\,\,\,\,\,\,\,\comps{\B}{2n+1}{2n+1}=\frac{\gamma}{\ud^2}\Big(\ud\au{0}{0}-\udd(n+1)\Big)\,, \\
&\comps{\B}{2n+1}{2(n+1)+1}=\frac{\gamma}{\ud^2}\Big(-\udd(n+1)\Big)\,,\,\,\,\,\,\,\,\,\comps{\B}{2n+1}{2p}=\frac{\gamma}{\ud^2}\Big(-2\ud\an{n-p}{1}\Big)\,,\,\,\,\,\,\,\,\,\comps{\B}{2n+1}{2p+1}=\frac{\gamma}{\ud^2}\Big(\ud\an{n-p}{0}\Big)\,,
\end{split}
\end{flalign}
satisfies the second equality of\eref{3.2.2}. Where we wrote $\B\coloneqq\B_\xi$ for compactness.
\end{prop}

\begin{proof}
Rewriting the second equality of\eref{3.2.2} as
\begin{flalign}\label{3.2.14}
\begin{split}
\B\Af(\xi)=\partial_\xi\Af(\xi)-\sum_{j=1}^{2m}\frac{\A_j}{\xi-\tj}\Af(\xi)=\partial_\xi\Af(\xi)-\sum_{j=1}^{2m}\partial_j\Af(\xi)\,,
\end{split}
\end{flalign}
we obtain for each $n\in\Z_{\geq0}$
\begin{flalign}\label{3.2.15}
\begin{split}
\compo{\B\Af}{2n+\alpha}(\xi)=\partial_\xi\af{n}{\alpha}(\xi)-\sum_{j=1}^{2m}\partial_j\af{n}{\alpha}(\xi)\,.
\end{split}
\end{flalign}
Hence, writing the left hand side in a similar fashion to\eref{3.2.5} and expressing explicitly the right hand side through\eref{3.2.10} for $\alpha=0$, this leads to
\begin{flalign}\label{3.2.16}
\begin{split}
\sum_{p=0}^\infty\Bigg(\comps{\B}{2n}{2p}\af{p}{0}(\xi)+\comps{\B}{2n}{2p+1}\af{p}{1}(\xi)\Bigg)=\af{n}{1}(\xi)-\frac{\gamma}{\dot{u}_0^2}\Bigg(\Big(\ud&\au{0}{0}+\udd n\Big)\af{n}{0}(\xi) \\
&+\udd(n+1)\af{n+1}{0}(\xi)\Bigg)-\frac{\gamma}{\dot{u}_0}\sum_{k=0}^{n-1}\an{n-k}{0}\af{k}{0}(\xi)\,.
\end{split}
\end{flalign}
Whilst for $\alpha=1$ we have to use Lemma \ref{t3.2.6} in order to explicit the right hand side, yielding
\begin{flalign}\label{3.2.17}
\begin{split}
\compo{\B\Af}{2n+1}(\xi)=\frac{\gamma}{\ud^2}\Bigg(\Big(\gamma\big(v_0+\xi\big)-2\ud\au{0}{1}\Big)\af{n}{0}(\xi)+\Big(\ud\au{0}{0}-\udd(n+&1)\Big)\af{n}{1}(\xi)-\udd(n+1)\af{n+1}{1}(\xi)\Bigg) \\
&+\frac{\gamma}{\ud}\sum_{k=0}^{n-1}\Bigg(\an{n-k}{0}\af{k}{1}(\xi)-2\an{n-k}{1}\af{k}{0}(\xi)\Bigg)\,.
\end{split}
\end{flalign}
Finally, solving\eref{3.2.16} for the components of $\B$, we obtain the first five relations of\eref{3.2.13}. Whilst the seven remaining ones are obtained by expanding
\begin{flalign}\label{3.2.18}
\begin{split}
\compo{\B\Af}{2n+1}(\xi)=\sum_{p=0}^\infty\Bigg(\comps{\B}{2n+1}{2p}\af{p}{0}(\xi)+\comps{\B}{2n+1}{2p+1}\af{p}{1}(\xi)\Bigg)
\end{split}
\end{flalign}
in\eref{3.2.17}, and then again solving for the components of $\B$.
\end{proof}

\begin{rem}\label{t3.2.8}
In a matrix form, the first few components yielded by Proposition \ref{t3.2.7} are
\begin{flalign}\label{3.2.19}
\begin{split}
\B_\xi=\frac{\gamma}{\ud}\begin{pmatrix}-\au{0}{0}&\frac{\ud}{\gamma}&-\frac{\udd}{\ud}&0&0&0&\cdots\\\frac{\gamma}{\ud}w(\xi)-2\au{0}{1}&\au{0}{0}-\frac{\udd}{\ud}&0&-\frac{\udd}{\ud}&0&0&\cdots\\-\an{1}{0}&0&-\au{0}{0}-\frac{\udd}{\ud}&\frac{\ud}{\gamma}&-2\frac{\udd}{\ud}&0&\cdots\\-2\an{1}{1}&\an{1}{0}&\frac{\gamma}{\ud}w(\xi)-2\au{0}{1}&\au{0}{0}-2\frac{\udd}{\ud}&0&-2\frac{\udd}{\ud}&\cdots\\-\an{2}{0}&0&-\an{1}{0}&0&-\au{0}{0}-2\frac{\udd}{\ud}&\frac{\ud}{\gamma}&\cdots\\-2\an{2}{1}&\an{2}{0}&-2\an{1}{1}&\an{1}{0}&\frac{\gamma}{\ud}w(\xi)-2\au{0}{1}&\au{0}{0}-3\frac{\udd}{\ud}&\cdots\\\vdots&\vdots&\vdots&\vdots&\vdots&\vdots&\ddots\end{pmatrix}\,,
\end{split}
\end{flalign}
where we wrote $w(\xi)\coloneqq v_0+\xi$ for compactness.
\end{rem}

\subsubsection{Dynamics of the Lax Matrices}\label{s322}

We try to evaluate the $\tj$-derivatives of the matrices $\A_i$, $\forall j,i\in\{1,...,2m\}$, in terms of the commutators of these matrices and $\B_j\coloneqq\B_{\tj}$. It turns out that, for the components that we manage to compute, we are led to a Schlesinger system of equations, which we then formally formulate from the zero curvature equation presented in \sref{s13}. To be more precise, the components that we manage to compute are all the components labeled by $\big(n,p\big)\in\Z_{\geq0}^2$ such that $p\geq n$, i.e. the $4 \times 4$ blocks on the diagonal and the zeros.

\begin{lem}\label{t3.2.9}
For any $i,j\in\{1,...,2m\}\subset\N$ and $n,p\in\Z_{\geq0}\,$, the following expressions hold, $[\A_i,\A_j]_{2n+\alpha,2p+\beta}\big|_{p>n}=0\,$, and
\begin{flalign}\label{3.2.20}
\begin{split}
[\A_i,\A_j]_{2n+\alpha,2n+\beta}=-\frac{\ud^2}{\gamma^2}(-1)^{i+j+\beta}\bigg(\af{0}{0}(\tau_i)\af{0}{1}(\tj)-\af{0}{1}(\tau_i)\af{0}{0}(\tj)\bigg)\bigg(\af{0}{\alpha}(\tau_i)\af{0}{-\beta}(\tj)+\af{0}{-\beta}(\tau_i)\af{0}{\alpha}(\tj)\bigg)\,,
\end{split}
\end{flalign}
where we used Notation \ref{t3.2.2}.
\end{lem}

\begin{proof}
Let $n,p\in\Z_{\geq0}$ and $i,j\in\{1,...,2m\}$. Using $\comps{\A_j}{2n+\alpha}{2p+\beta}\big|_{p>n}=0$, we observe
\begin{flalign}\label{3.2.21}
\begin{split}
[\A_i,\A_j]_{2n+\alpha,2p+\beta}&=\sum_{q=0}^\infty\Bigg(\comps{\A_i}{2n+\alpha}{2q}\comps{\A_j}{2q}{2p+\beta}+\comps{\A_i}{2n+\alpha}{2q+1}\comps{\A_j}{2q+1}{2p+\beta} \\
&\,\,\,\,\,\,\,\,\,\,\,\,\,\,\,\,\,\,\,\,\,\,\,\,\,\,\,\,\,\,\,\,\,\,\,\,\,\,\,\,\,\,\,\,\,\,\,\,\,\,\,\,\,\,\,\,\,\,\,\,\,\,\,\,\,\,\,\,\,\,\,\,\,\,\,\,\,\,\,\,\,\,\,\,-\comps{\A_j}{2n+\alpha}{2q}\comps{\A_i}{2q}{2p+\beta}-\comps{\A_j}{2n+\alpha}{2q+1}\comps{\A_i}{2q+1}{2p+\beta}\Bigg) \\
&=\sum_{q=p}^n\Bigg(\comps{\A_i}{2n+\alpha}{2q}\comps{\A_j}{2q}{2p+\beta}+\comps{\A_i}{2n+\alpha}{2q+1}\comps{\A_j}{2q+1}{2p+\beta} \\
&\,\,\,\,\,\,\,\,\,\,\,\,\,\,\,\,\,\,\,\,\,\,\,\,\,\,\,\,\,\,\,\,\,\,\,\,\,\,\,\,\,\,\,\,\,\,\,\,\,\,\,\,\,\,\,\,\,\,\,\,\,\,\,\,\,\,\,\,\,\,\,\,\,\,\,\,\,\,\,\,\,\,\,\,-\comps{\A_j}{2n+\alpha}{2q}\comps{\A_i}{2q}{2p+\beta}-\comps{\A_j}{2n+\alpha}{2q+1}\comps{\A_i}{2q+1}{2p+\beta}\Bigg)\,,
\end{split}
\end{flalign}
where we used again an expansion similar to\eref{3.2.5}. Thereby it follows that $[\A_i,\A_j]_{2n+\alpha,2p+\beta}\big|_{p>n}=0$, moreover this also entails
\begin{flalign}\label{3.2.22}
\begin{split}
[\A_i,\A_j]_{2n+\alpha,2n+\beta}&=\comps{\A_i}{2n+\alpha}{2n}\comps{\A_j}{2n}{2n+\beta}+\comps{\A_i}{2n+\alpha}{2n+1}\comps{\A_j}{2n+1}{2n+\beta} \\
&\,\,\,\,\,\,\,\,\,\,\,\,\,\,\,\,\,\,\,\,\,\,\,\,\,\,\,\,\,\,\,\,\,\,\,\,\,\,\,\,\,\,\,\,\,\,\,\,\,\,\,\,\,\,\,\,\,\,\,\,\,\,\,\,\,\,\,\,\,\,\,\,\,\,-\comps{\A_j}{2n+\alpha}{2n}\comps{\A_i}{2n}{2n+\beta}-\comps{\A_j}{2n+\alpha}{2n+1}\comps{\A_i}{2n+1}{2n+\beta} \\
&=\comps{\A_i}{\alpha}{0}\comps{\A_j}{0}{\beta}+\comps{\A_i}{\alpha}{1}\comps{\A_j}{1}{\beta}-\comps{\A_j}{\alpha}{0}\comps{\A_i}{0}{\beta}-\comps{\A_j}{\alpha}{1}\comps{\A_i}{1}{\beta}\,,
\end{split}
\end{flalign}
which is a consequence of Proposition \ref{t3.2.3}, in particular
\begin{flalign}\label{3.2.23}
\begin{split}
\comps{\A_j}{2n+\alpha}{2n+\beta}=-\udf(-1)^{j+\beta}\af{0}{-\beta}(\tj)\af{0}{\alpha}(\tj)=\comps{\A_j}{\alpha}{\beta}\,.
\end{split}
\end{flalign}
Hence\eref{3.2.22} explicitly gives
\begin{flalign}\label{3.2.24}
\begin{split}
[\A_i,\A_j]_{2n+\alpha,2n+\beta}=\frac{\ud^2}{\gamma^2}(-1)^{i+j+\beta}\Bigg(\af{0}{1}(\tau_i)&\af{0}{\alpha}(\tau_i)\af{0}{-\beta}(\tj)\af{0}{0}(\tj)-\af{0}{0}(\tau_i)\af{0}{\alpha}(\tau_i)\af{0}{-\beta}(\tj)\af{0}{1}(\tj) \\
&-\af{0}{1}(\tj)\af{0}{\alpha}(\tj)\af{0}{-\beta}(\tau_i)\af{0}{0}(\tau_i)+\af{0}{0}(\tj)\af{0}{\alpha}(\tj)\af{0}{-\beta}(\tau_i)\af{0}{1}(\tau_i)\Bigg)\,,
\end{split}
\end{flalign}
which can indeed be factorized as\eref{3.2.20}.
\end{proof}

\begin{lem}\label{t3.2.10}
Let $j\in\{1,...,2m\}\subset\N\,$, then, for $n,p\in\Z_{\geq0}\,$, we have $[\B_j,\A_j]_{2n,2p}\big|_{p>n}=0\,$, and
\begin{flalign}\label{3.2.25}
\begin{split}
-(-1)^j[\B_j,\A_j]_{2n,2n}=\udf\af{0}{1}^2+\bigg(\frac{\gamma}{\ud}\big(v_0+\tj\big)-2\au{0}{1}\bigg)\af{0}{0}^2-\frac{\udd}{\ud}\Big(\af{0}{1}\af{0}{0}+\af{0}{1}\af{1}{0}+\af{0}{0}\af{1}{1}\Big)\,,
\end{split}
\end{flalign}
where we utilize Notation \ref{t3.1.16}.
\end{lem}

\begin{proof}
Set $j\in\{1,...,2m\}$ and $n,p\in\Z_{\geq0}$. Proceeding similarly to\eref{3.2.21}, employing the zeros of $\B_j$ and $\A_j$ given by Propositions \ref{t3.2.3} and \ref{t3.2.7}, we are led to
\begin{flalign}\label{3.2.26}
\begin{split}
[\B_j,\A_j]_{2n,2p}&=\sum_{q=0}^\infty\Bigg(\comps{\B_j}{2n}{2q}\comps{\A_j}{2q}{2p}+\comps{\B_j}{2n}{2q+1}\comps{\A_j}{2q+1}{2p}-\comps{\A_j}{2n}{2q}\comps{\B_j}{2q}{2p}-\comps{\A_j}{2n}{2q+1}\comps{\B_j}{2q+1}{2p}\Bigg) \\
&=\sum_{q=p}^{n+1}\comps{\B_j}{2n}{2q}\comps{\A_j}{2q}{2p}+\comps{\A_j}{2n+1}{2p}-\sum_{q=p-1}^n\comps{\A_j}{2n}{2q}\comps{\B_j}{2q}{2p}-\sum_{q=p}^n\comps{\A_j}{2n}{2q+1}\comps{\B_j}{2q+1}{2p} \\
&=\comps{\A_j}{2n+1}{2p}+\comps{\B_j}{2n}{2(n+1)}\comps{\A_j}{2(n+1)}{2p}-\comps{\A_j}{2n}{2(p-1)}\comps{\B_j}{2(p-1)}{2p} \\
&\,\,\,\,\,\,\,\,\,\,\,\,\,\,\,\,\,\,\,\,\,\,\,\,\,\,\,\,\,\,\,\,\,\,\,\,\,\,\,\,\,\,\,\,\,\,\,\,\,\,\,\,\,\,\,\,\,\,\,\,\,\,\,\,\,\,\,\,\,\,+\sum_{q=p}^n\Bigg(\comps{\B_j}{2n}{2q}\comps{\A_j}{2q}{2p}-\comps{\A_j}{2n}{2q}\comps{\B_j}{2q}{2p}-\comps{\A_j}{2n}{2q+1}\comps{\B_j}{2q+1}{2p}\Bigg)\,,
\end{split}
\end{flalign}
where, in the second line, the relation $\comps{\B_j}{2n}{2p+1}=\delta_{n,p}$ of Proposition \ref{t3.2.7} was utilized in the second term. To begin with we treat
\begin{flalign}\label{3.2.27}
\begin{split}
[\B_j,\A_j]_{2n,2p}\big|_{p>n}&=\comps{\B_j}{2n}{2(n+1)}\comps{\A_j}{2(n+1)}{2p}-\comps{\A_j}{2n}{2(p-1)}\comps{\B_j}{2(p-1)}{2p} \\
&=-\frac{\gamma}{\ud}\bigg(\frac{\udd}{\ud}(n+1)\comps{\A_j}{2(n+1)}{2p}-\frac{\udd}{\ud}p\comps{\A_j}{2n}{2(p-1)}\bigg)\,.
\end{split}
\end{flalign}
Thus we immediately have $[\B_j,\A_j]_{2n,2p}\big|_{p>n+1}=0$, whilst
\begin{flalign}\label{3.2.28}
\begin{split}
[\B_j,\A_j]_{2n,2(n+1)}=-\frac{\gamma}{\ud}\bigg(\frac{\udd}{\ud}(n+1)\comps{\A_j}{2(n+1)}{2(n+1)}-\frac{\udd}{\ud}(n+1)\comps{\A_j}{2n}{2n}\bigg)=0
\end{split}
\end{flalign}
is ensured by\eref{3.2.23}. Thereby $[\B_j,\A_j]_{2n,2p}\big|_{p>n}=0$. Besides, Proposition \ref{t3.2.7} together with\eref{3.2.23} also yields
\begin{flalign}\label{3.2.29}
\begin{split}
[\B_j,\A_j]_{2n,2n}&=\comps{\A_j}{2n+1}{2n}-\frac{\gamma}{\ud}\bigg(\frac{\udd}{\ud}(n+1)\comps{\A_j}{2(n+1)}{2n}-\frac{\udd}{\ud}n\comps{\A_j}{2n}{2(n-1)}\bigg)-\comps{\A_j}{2n}{2n+1}\comps{\B_j}{2n+1}{2n} \\
&=\comps{\A_j}{1}{0}-\frac{\gamma}{\ud}\frac{\udd}{\ud}\comps{\A_j}{2(n+1)}{2n}-\frac{\gamma}{\ud}\bigg(\frac{\gamma}{\ud}\big(v_0+\tj\big)-2\au{0}{1}\bigg)\comps{\A_j}{0}{1}\,,
\end{split}
\end{flalign}
as a consequence of\eref{3.2.26}. Note that, in order to obtain this expression, we noticed that Proposition \ref{t3.2.3} entails
\begin{flalign}\label{3.2.30}
\begin{split}
\comps{\A_j}{2(n+1)}{2n}=-\udf(-1)^j\Big(\af{0}{1}\af{0}{0}+\af{0}{1}\af{1}{0}+\af{0}{0}\af{1}{1}\Big)=\comps{\A_j}{2n}{2(n-1)}\,.
\end{split}
\end{flalign}
Finally, injecting the $\A_j$ components accordingly to Proposition \ref{t3.2.3} leads to the result.
\end{proof}

\begin{lem}\label{t3.2.11}
Once again set $j\in\{1,...,2m\}$ and let $n,p\in\Z_{\geq0}\,$. The matrices $\B_j$ and $\A_j$ obey $[\B_j,\A_j]_{2n+1,2p}\big|_{p>n}=0\,$, and
\begin{flalign}\label{3.2.31}
\begin{split}
-(-1)^j[\B_j,\A_j]_{2n+1,2n}&=2\bigg(\frac{\gamma}{\ud}\big(v_0+\tj\big)-2\au{0}{1}\bigg)\af{0}{1}\af{0}{0}+\bigg(2\au{0}{0}-\frac{\udd}{\ud}\bigg)\af{0}{1}^2-\frac{\udd}{\ud}\Big(\af{0}{1}^2+2\af{0}{1}\af{1}{1}\Big) \\
&=2\Bigg(\bigg(\frac{\gamma}{\ud}\big(v_0+\tj\big)-2\au{0}{1}\bigg)\af{0}{1}\af{0}{0}+\bigg(\au{0}{0}-\frac{\udd}{\ud}\bigg)\af{0}{1}^2-\frac{\udd}{\ud}\af{0}{1}\af{1}{1}\Bigg)\,.
\end{split}
\end{flalign}
\end{lem}

\begin{proof}
We reach these results in a similar fashion to Lemma \ref{t3.2.10}, beginning with
\begin{flalign}\label{3.2.32}
\begin{split}
[\B_j,\A_j]_{2n+1,2p}&=\sum_{q=0}^\infty\Bigg(\comps{\B_j}{2n+1}{2q}\comps{\A_j}{2q}{2p}+\comps{\B_j}{2n+1}{2q+1}\comps{\A_j}{2q+1}{2p} \\
&\,\,\,\,\,\,\,\,\,\,\,\,\,\,\,\,\,\,\,\,\,\,\,\,\,\,\,\,\,\,\,\,\,\,\,\,\,\,\,\,\,\,\,\,\,\,\,\,\,\,\,\,\,\,\,\,\,\,\,\,\,\,\,\,\,\,\,\,\,\,-\comps{\A_j}{2n+1}{2q}\comps{\B_j}{2q}{2p}-\comps{\A_j}{2n+1}{2q+1}\comps{\B_j}{2q+1}{2p}\Bigg) \\
&=\sum_{q=p}^n\comps{\B_j}{2n+1}{2q}\comps{\A_j}{2q}{2p}+\sum_{q=p}^{n+1}\comps{\B_j}{2n+1}{2q+1}\comps{\A_j}{2q+1}{2p} \\
&\,\,\,\,\,\,\,\,\,\,\,\,\,\,\,\,\,\,\,\,\,\,\,\,\,\,\,\,\,\,\,\,\,\,\,\,\,\,\,\,\,\,\,\,\,\,\,\,\,\,\,\,\,\,\,\,\,\,\,\,\,\,\,\,\,\,\,\,\,\,-\sum_{q=p-1}^n\comps{\A_j}{2n+1}{2q}\comps{\B_j}{2q}{2p}-\sum_{q=p}^n\comps{\A_j}{2n+1}{2q+1}\comps{\B_j}{2q+1}{2p} \\
&=\comps{\B_j}{2n+1}{2(n+1)+1}\comps{\A_j}{2(n+1)+1}{2p}-\comps{\A_j}{2n+1}{2(p-1)}\comps{\B_j}{2(p-1)}{2p} \\
&\,\,\,\,\,\,\,\,\,\,\,\,\,\,\,\,\,\,\,\,\,\,\,\,\,\,\,\,\,\,\,\,\,\,\,\,\,\,\,\,\,\,\,\,\,\,\,\,\,\,\,\,\,\,\,\,\,\,\,\,\,\,\,\,\,\,\,\,\,\,+\sum_{q=p}^n\Bigg(\comps{\B_j}{2n+1}{2q}\comps{\A_j}{2q}{2p}+\comps{\B_j}{2n+1}{2q+1}\comps{\A_j}{2q+1}{2p} \\
&\,\,\,\,\,\,\,\,\,\,\,\,\,\,\,\,\,\,\,\,\,\,\,\,\,\,\,\,\,\,\,\,\,\,\,\,\,\,\,\,\,\,\,\,\,\,\,\,\,\,\,\,\,\,\,\,\,\,\,\,\,\,\,\,\,\,\,\,\,\,\,\,\,\,\,\,\,\,\,\,\,\,\,\,\,\,\,\,\,\,\,\,\,\,\,\,\,\,\,\,\,-\comps{\A_j}{2n+1}{2q}\comps{\B_j}{2q}{2p}-\comps{\A_j}{2n+1}{2q+1}\comps{\B_j}{2q+1}{2p}\Bigg)\,.
\end{split}
\end{flalign}
In particular this leads to
\begin{flalign}\label{3.2.33}
\begin{split}
[\B_j,\A_j]_{2n+1,2p}\big|_{p>n}=-\frac{\gamma}{\ud}\bigg(\frac{\udd}{\ud}(n+1)\comps{\A_j}{2(n+1)+1}{2p}-\frac{\udd}{\ud}p\comps{\A_j}{2n+1}{2(p-1)}\bigg)\,,
\end{split}
\end{flalign}
ensuring $[\B_j,\A_j]_{2n+1,2p}\big|_{p>n}=0$, and
\begin{flalign}\label{3.2.34}
\begin{split}
[\B_j,\A_j]_{2n+1,2n}&=-\frac{\gamma}{\ud}\frac{\udd}{\ud}\comps{\A_j}{2(n+1)+1}{2n}+\comps{\B_j}{2n+1}{2n}\bigg(\comps{\A_j}{2n}{2n}-\comps{\A_j}{2n+1}{2n+1}\bigg) \\
&,\,\,\,\,\,\,\,\,\,\,\,\,\,\,\,\,\,\,\,\,\,\,\,\,\,\,\,\,\,\,\,\,\,\,\,\,\,\,\,\,\,\,\,\,\,\,\,\,\,\,\,\,\,\,\,\,\,\,\,\,\,\,\,\,\,\,\,\,\,\,\,\,\,\,\,\,\,\,\,\,\,\,\,\,\,\,\,+\comps{\A_j}{2n+1}{2n}\bigg(\comps{\B_j}{2n+1}{2n+1}-\comps{\B_j}{2n}{2n}\bigg) \\
&=\frac{\gamma}{\ud}\Bigg[\bigg(\frac{\gamma}{\ud}\big(v_0+\tj\big)-2\au{0}{1}\bigg)\bigg(\comps{\A_j}{0}{0}-\comps{\A_j}{1}{1}\bigg)+\bigg(2\au{0}{0}-\frac{\udd}{\ud}\bigg)\comps{\A_j}{1}{0}-\frac{\udd}{\ud}\comps{\A_j}{2(n+1)+1}{2n}\Bigg]\,,
\end{split}
\end{flalign}
where we also utilized
\begin{flalign}\label{3.2.35}
\begin{split}
\comps{\A_j}{2(n+1)+1}{2n}=-\udf(-1)^j\Big(\af{0}{1}^2+2\af{0}{1}\af{1}{1}\Big)=\comps{\A_j}{2n+1}{2(n-1)}\,,
\end{split}
\end{flalign}
following from Proposition \ref{t3.2.3}. It then only remains to inject the $\A_j$ components.
\end{proof}

\begin{lem}\label{t3.2.12}
The following relations are satisfied for any $j\in\{1,...,2m\}$ and $n,p\in\Z_{\geq0}\,$, $[\B_j,\A_j]_{2n,2p+1}\big|_{p>n}=0\,$, and
\begin{flalign}\label{3.2.36}
\begin{split}
-(-1)^j[\B_j,\A_j]_{2n,2n+1}&=\bigg(2\au{0}{0}-\frac{\udd}{\ud}\bigg)\af{0}{0}^2+\frac{\udd}{\ud}\Big(\af{0}{0}^2+2\af{0}{0}\af{1}{0}\Big)-2\udf\af{0}{1}\af{0}{0} \\
&=2\Bigg(\au{0}{0}\af{0}{0}^2+\frac{\udd}{\ud}\af{0}{0}\af{1}{0}-\udf\af{0}{1}\af{0}{0}\Bigg)\,.
\end{split}
\end{flalign}
\end{lem}

\begin{proof}
As before we first compute
\begin{flalign}\label{3.2.37}
\begin{split}
[\B_j,\A_j]_{2n,2p+1}&=\sum_{q=0}^\infty\Bigg(\comps{\B_j}{2n}{2q}\comps{\A_j}{2q}{2p+1}+\comps{\B_j}{2n}{2q+1}\comps{\A_j}{2q+1}{2p+1} \\
&\,\,\,\,\,\,\,\,\,\,\,\,\,\,\,\,\,\,\,\,\,\,\,\,\,\,\,\,\,\,\,\,\,\,\,\,\,\,\,\,\,\,\,\,\,\,\,\,\,\,\,\,\,\,\,\,\,\,\,\,\,\,\,\,\,\,\,\,\,\,\,\,\,\,\,\,\,\,\,\,\,\,\,-\comps{\A_j}{2n}{2q}\comps{\B_j}{2q}{2p+1}-\comps{\A_j}{2n}{2q+1}\comps{\B_j}{2q+1}{2p+1}\Bigg) \\
&=\sum_{q=p}^{n+1}\comps{\B_j}{2n}{2q}\comps{\A_j}{2q}{2p+1}+\comps{\A_j}{2n+1}{2p+1}-\comps{\A_j}{2n}{2p}-\sum_{q=p-1}^n\comps{\A_j}{2n}{2q+1}\comps{\B_j}{2q+1}{2p+1} \\
&=\comps{\A_j}{2n+1}{2p+1}-\comps{\A_j}{2n}{2p}+\comps{\B_j}{2n}{2(n+1)}\comps{\A_j}{2(n+1)}{2p+1}-\comps{\A_j}{2n}{2(p-1)+1}\comps{\B_j}{2(p-1)+1}{2p+1} \\
&\,\,\,\,\,\,\,\,\,\,\,\,\,\,\,\,\,\,\,\,\,\,\,\,\,\,\,\,\,\,\,\,\,\,\,\,\,\,\,\,\,\,\,\,\,\,\,\,\,\,\,\,\,\,\,\,\,\,\,\,\,\,\,\,\,\,\,\,\,\,\,\,\,\,\,\,\,\,\,\,\,\,\,\,\,\,\,\,\,\,\,\,\,\,\,\,\,\,+\sum_{q=p}^n\Bigg(\comps{\B_j}{2n}{2q}\comps{\A_j}{2q}{2p+1}-\comps{\A_j}{2n}{2q+1}\comps{\B_j}{2q+1}{2p+1}\Bigg)\,,
\end{split}
\end{flalign}
which is yielding
\begin{flalign}\label{3.2.38}
\begin{split}
[\B_j,\A_j]_{2n,2p+1}\big|_{p>n}=-\frac{\gamma}{\ud}\bigg(\frac{\udd}{\ud}(n+1)\comps{\A_j}{2(n+1)}{2p+1}-\frac{\udd}{\ud}p\comps{\A_j}{2n}{2(p-1)+1}\bigg)\,,
\end{split}
\end{flalign}
so that we indeed have $[\B_j,\A_j]_{2n,2p+1}\big|_{p>n}=0$. Whilst
\begin{flalign}\label{3.2.39}
\begin{split}
[\B_j,\A_j]_{2n,2n+1}&=\comps{\A_j}{2n+1}{2n+1}-\comps{\A_j}{2n}{2n}-\frac{\gamma}{\ud}\frac{\udd}{\ud}\comps{\A_j}{2(n+1)}{2n+1}+\comps{\A_j}{2n}{2n+1}\bigg(\comps{\B_j}{2n}{2n}-\comps{\B_j}{2n+1}{2n+1}\bigg) \\
&=\comps{\A_j}{1}{1}-\comps{\A_j}{0}{0}-\frac{\gamma}{\ud}\frac{\udd}{\ud}\comps{\A_j}{2(n+1)}{2n+1}-\frac{\gamma}{\ud}\bigg(2\au{0}{0}-\frac{\udd}{\ud}\bigg)\comps{\A_j}{0}{1}\,,
\end{split}
\end{flalign}
where we employed the following consequence of Proposition \ref{t3.2.3},
\begin{flalign}\label{3.2.40}
\begin{split}
\comps{\A_j}{2(n+1)}{2n+1}=-\udf(-1)^j\Big(\af{0}{0}^2+2\af{0}{0}\af{1}{0}\Big)=\comps{\A_j}{2n}{2(n-1)+1}\,.
\end{split}
\end{flalign}
And the components of $\A_j$ give the result.
\end{proof}

\begin{lem}\label{t3.2.13}
The commutator of $\B_j$ and $\A_j$ also satisfies, $\forall j\in\{1,...,2m\}$ and $\forall n,p\in\Z_{\geq0}\,$, $[\B_j,\A_j]_{2n+1,2p+1}\big|_{p>n}=0\,$, and
\begin{flalign}\label{3.2.41}
\begin{split}
-(-1)^j[\B_j,\A_j]_{2n+1,2n+1}=\frac{\udd}{\ud}\Big(\af{0}{1}\af{0}{0}+\af{0}{1}\af{1}{0}+\af{0}{0}\af{1}{1}\Big)-\bigg(\frac{\gamma}{\ud}\big(v_0+\tj\big)-2\au{0}{1}\bigg)\af{0}{0}^2-\udf\af{0}{1}^2\,.
\end{split}
\end{flalign}
\end{lem}

\begin{proof}
Once more we begin with
\begin{flalign}\label{3.2.42}
\begin{split}
[\B_j,\A_j]_{2n+1,2p+1}&=\sum_{q=0}^\infty\Bigg(\comps{\B_j}{2n+1}{2q}\comps{\A_j}{2q}{2p+1}+\comps{\B_j}{2n+1}{2q+1}\comps{\A_j}{2q+1}{2p+1} \\
&\,\,\,\,\,\,\,\,\,\,\,\,\,\,\,\,\,\,\,\,\,\,\,\,\,\,\,\,\,\,\,\,\,\,\,\,\,\,\,\,\,\,\,\,\,\,\,-\comps{\A_j}{2n+1}{2q}\comps{\B_j}{2q}{2p+1}-\comps{\A_j}{2n+1}{2q+1}\comps{\B_j}{2q+1}{2p+1}\Bigg) \\
&=\comps{\B_j}{2n+1}{2(n+1)+1}\comps{\A_j}{2(n+1)+1}{2p+1}-\comps{\A_j}{2n+1}{2p}-\comps{\A_j}{2n+1}{2(p-1)+1}\comps{\B_j}{2(p-1)+1}{2p+1} \\
&\,\,\,\,\,\,\,\,\,\,\,\,\,\,\,\,+\sum_{q=p}^n\Bigg(\comps{\B_j}{2n+1}{2q}\comps{\A_j}{2q}{2p+1}+\comps{\B_j}{2n+1}{2q+1}\comps{\A_j}{2q+1}{2p+1}-\comps{\A_j}{2n+1}{2q+1}\comps{\B_j}{2q+1}{2p+1}\Bigg)\,,
\end{split}
\end{flalign}
which entails
\begin{flalign}\label{3.2.43}
\begin{split}
[\B_j,\A_j]_{2n+1,2p+1}\big|_{p>n}=-\comps{\A_j}{2n+1}{2p}-\frac{\gamma}{\ud}\bigg(\frac{\udd}{\ud}(n+1)\comps{\A_j}{2(n+1)+1}{2p+1}-\frac{\udd}{\ud}p\comps{\A_j}{2n+1}{2(p-1)+1}\bigg)=0\,.
\end{split}
\end{flalign}
Besides, since
\begin{flalign}\label{3.2.44}
\begin{split}
\comps{\A_j}{2(n+1)+1}{2n+1}=\udf(-1)^j\Big(\af{0}{1}\af{0}{0}+\af{0}{1}\af{1}{0}+\af{0}{0}\af{1}{1}\Big)=\comps{\A_j}{2n+1}{2(n-1)+1}\,,
\end{split}
\end{flalign}
it follows that
\begin{flalign}\label{3.2.45}
\begin{split}
[\B_j,\A_j]_{2n+1,2n+1}&=-\comps{\A_j}{1}{0}-\frac{\gamma}{\ud}\frac{\udd}{\ud}\comps{\A_j}{2(n+1)+1}{2n+1}+\comps{\B_j}{2n+1}{2n}\comps{\A_j}{2n}{2n+1} \\
&=\frac{\gamma}{\ud}\bigg(\frac{\gamma}{\ud}\big(v_0+\tj\big)-2\au{0}{1}\bigg)\comps{\A_j}{0}{1}-\frac{\gamma}{\ud}\frac{\udd}{\ud}\comps{\A_j}{2(n+1)+1}{2n+1}-\comps{\A_j}{1}{0}\,.
\end{split}
\end{flalign}
And whence Proposition \ref{t3.2.3} completes the proof.
\end{proof}

Henceforth we can tackle the Schlesinger equations.

\begin{prop}\label{t3.2.14}
Let $i,j\in\{1,...,2m\}\subset\N$ such that $i\neq j$, and $n,p\in\Z_{\geq0}\,$. The infinite size matrices $\A_j$ and $\B_j\coloneqq\B_{\tj}$, that solve\eref{3.2.2} thereby describing the dynamics of the AWF, satisfy the following set of equations,
\begin{flalign}\label{3.2.46}
\begin{split}
\comps{\partial_j\A_i}{2n+\alpha}{2p+\beta}\big|_{p\geq n}=\comps{\frac{[\A_i,\A_j]}{\tau_i-\tj}}{2n+\alpha}{2p+\beta}\,,\,\,\,\,\,\,\,\,\,\,\comps{\partial_j\A_j}{2n+\alpha}{2p+\beta}\big|_{p\geq n}=\comps{\sum_{\substack{i=1\\i\neq j}}^{2m}\frac{[\A_i,\A_j]}{\tj-\tau_i}-[\A_j,\B_j]}{2n+\alpha}{2p+\beta}\,.
\end{split}
\end{flalign}
\end{prop}

\begin{proof}
Throughout this proof $i\neq j$, whilst $i,j\in\{1,...,2m\}$, and $n,p\in\Z_{\geq0}$. First notice that, according to Proposition \ref{t3.2.3},
\begin{flalign}\label{3.2.47}
\begin{split}
\comps{\partial_j\A_i}{2n+\alpha}{2p+\beta}\big|_{p>n}=0=\comps{\partial_j\A_j}{2n+\alpha}{2p+\beta}\big|_{p>n}\,,
\end{split}
\end{flalign}
which coincides with Lemmas \ref{t3.2.9}, \ref{t3.2.10}, \ref{t3.2.11}, \ref{t3.2.12}, and \ref{t3.2.13} for $p>n$. Therefore it only remains to check these relations for $p=n$, to which we turn next, beginning with the first equality. Using Proposition \ref{t3.2.3} and Lemmas \ref{t2.2.23}, \ref{t2.2.14}, and \ref{t3.2.9}, we evaluate
\begin{flalign}\label{3.2.48}
\begin{split}
\comps{\partial_j\A_i}{2n+\alpha}{2n+\beta}&=-\frac{\ud}{\gamma}(-1)^{i+\beta}\partial_j\af{0}{-\beta}(\tau_i)\af{0}{\alpha}(\tau_i) \\
&=-\frac{\ud}{\gamma}(-1)^{i+j+\beta}\bigg(\Ri_1(\tau_i,\tj)\af{0}{-\beta}(\tj)\af{0}{\alpha}(\tau_i)+\Ri_1(\tau_i,\tj)\af{0}{-\beta}(\tau_i)\af{0}{\alpha}(\tj)\bigg) \\
&=-\frac{\ud^2}{\gamma^2}\frac{(-1)^{i+j+\beta}}{\tau_i-\tj}\bigg(\af{0}{0}(\tau_i)\af{0}{1}(\tj)-\af{0}{1}(\tau_i)\af{0}{0}(\tj)\bigg)\bigg(\af{0}{\alpha}(\tau_i)\af{0}{-\beta}(\tj)+\af{0}{-\beta}(\tau_i)\af{0}{\alpha}(\tj)\bigg) \\
&=\comps{\frac{[\A_i,\A_j]}{\tau_i-\tj}}{2n+\alpha}{2n+\beta}\,.
\end{split}
\end{flalign}
Turning to the second equality, the following observation follows from Corollary \ref{t2.2.28} and\eref{3.2.48},
\begin{flalign}\label{3.2.49}
\begin{split}
\comps{\partial_j\A_j}{2n+\alpha}{2n+\beta}&=-\frac{\ud}{\gamma}(-1)^{j+\beta}\frac{d}{d\tj}\af{0}{-\beta}\af{0}{\alpha} \\
&=-\frac{\ud}{\gamma}(-1)^{j+\beta}\Bigg[\Bigg(\af{0}{-\beta+1}-\frac{\gamma}{\dot{u}_0}\bigg(\au{0}{-\beta}\af{0}{0}+\frac{\ddot{u}_0}{\dot{u}_0}\af{1}{-\beta}\bigg)-\sum_{\substack{i=1\\i\neq j}}^{2m}\partial_i\af{0}{-\beta}\Bigg)\af{0}{\alpha} \\
&\,\,\,\,\,\,\,\,\,\,\,\,\,\,\,\,\,\,\,\,\,\,\,\,\,\,\,\,\,\,\,\,\,\,\,\,\,\,\,\,\,\,\,\,\,\,\,\,\,\,\,\,\,\,\,\,\,\,\,\,\,\,\,\,\,\,\,\,\,\,\,\,+\af{0}{-\beta}\Bigg(\af{0}{\alpha+1}-\frac{\gamma}{\dot{u}_0}\bigg(\au{0}{\alpha}\af{0}{0}+\frac{\ddot{u}_0}{\dot{u}_0}\af{1}{\alpha}\bigg)-\sum_{\substack{i=1\\i\neq j}}^{2m}\partial_i\af{0}{\alpha}\Bigg)\Bigg] \\
&=-\frac{\ud}{\gamma}(-1)^{j+\beta}\Bigg[\Bigg(\af{0}{-\beta+1}-\frac{\gamma}{\dot{u}_0}\bigg(\au{0}{-\beta}\af{0}{0}+\frac{\ddot{u}_0}{\dot{u}_0}\af{1}{-\beta}\bigg)\Bigg)\af{0}{\alpha} \\
&\,\,\,\,\,\,\,\,\,\,\,\,\,\,\,\,\,\,\,\,\,\,\,\,\,\,\,\,\,\,\,\,\,\,\,\,\,\,\,\,\,\,\,\,\,\,\,\,\,\,\,\,\,\,+\af{0}{-\beta}\Bigg(\af{0}{\alpha+1}-\frac{\gamma}{\dot{u}_0}\bigg(\au{0}{\alpha}\af{0}{0}+\frac{\ddot{u}_0}{\dot{u}_0}\af{1}{\alpha}\bigg)\Bigg)\Bigg]+\frac{\ud}{\gamma}(-1)^{j+\beta}\sum_{\substack{i=1\\i\neq j}}^{2m}\partial_i\af{0}{-\beta}\af{0}{\alpha} \\
&=-\frac{\ud}{\gamma}(-1)^{j+\beta}\Bigg[\Bigg(\af{0}{-\beta+1}-\frac{\gamma}{\dot{u}_0}\bigg(\au{0}{-\beta}\af{0}{0}+\frac{\ddot{u}_0}{\dot{u}_0}\af{1}{-\beta}\bigg)\Bigg)\af{0}{\alpha} \\
&\,\,\,\,\,\,\,\,\,\,\,\,\,\,\,\,\,\,\,\,\,\,\,\,\,\,\,\,\,\,\,\,\,\,\,\,\,\,\,\,\,\,\,\,\,\,\,\,\,\,\,\,\,\,\,\,\,\,\,\,\,\,\,\,\,\,\,\,\,\,\,\,\,+\af{0}{-\beta}\Bigg(\af{0}{\alpha+1}-\frac{\gamma}{\dot{u}_0}\bigg(\au{0}{\alpha}\af{0}{0}+\frac{\ddot{u}_0}{\dot{u}_0}\af{1}{\alpha}\bigg)\Bigg)\Bigg]+\comps{\sum_{\substack{i=1\\i\neq j}}^{2m}\frac{[\A_i,\A_j]}{\tj-\tau_i}}{2n+\alpha}{2n+\beta}\,.
\end{split}
\end{flalign}
Therefore, in order to complete the proof it remains to show
\begin{flalign}\label{3.2.50}
\begin{split}
[\B_j,\A_j]_{2n+\alpha,2n+\beta}=(-1)^{j+\beta}\Bigg[\af{0}{\alpha}\Bigg(\au{0}{-\beta}\af{0}{0}+\frac{\ddot{u}_0}{\dot{u}_0}\af{1}{-\beta}-&\udf\af{0}{-\beta+1}\Bigg) \\
&+\af{0}{-\beta}\Bigg(\au{0}{\alpha}\af{0}{0}+\frac{\ddot{u}_0}{\dot{u}_0}\af{1}{\alpha}-\udf\af{0}{\alpha+1}\Bigg)\Bigg]\,.
\end{split}
\end{flalign}
For $(\alpha,\beta)=(0,1)$ this does coincide with Lemma \ref{t3.2.12}. But the three other cases are more involved since we need to employ the closure relation Lemma \ref{t2.2.34}. Namely, for $(\alpha,\beta)=(0,0)$,\eref{3.2.50} is satisfied if and only if
\begin{flalign}\label{3.2.51}
\begin{split}
[\B_j,\A_j]_{2n,2n}&=(-1)^j\Bigg[\af{0}{0}\Bigg(\au{0}{1}\af{0}{0}+\frac{\ddot{u}_0}{\dot{u}_0}\af{1}{1}-\udf\bigg(\frac{\gamma^2}{\dot{u}_0^2}\big(v_0+\tj\big)\af{0}{0}-\frac{\gamma\ddot{u}_0}{\dot{u}_0^2}\af{0}{1} \\
&\,\,\,\,\,\,\,\,\,\,\,\,\,\,\,\,\,\,\,\,\,\,\,\,\,\,\,\,\,\,\,\,\,\,\,\,\,\,\,\,\,\,\,\,\,\,\,\,\,\,\,\,\,\,\,\,\,\,\,\,\,\,\,\,+\frac{\gamma}{\dot{u}_0}\Big(\au{0}{0}\af{0}{1}-\au{0}{1}\af{0}{0}\Big)\bigg)\Bigg)+\af{0}{1}\Bigg(\au{0}{0}\af{0}{0}+\frac{\ddot{u}_0}{\dot{u}_0}\af{1}{0}-\udf\af{0}{1}\Bigg)\Bigg] \\
&=(-1)^j\Bigg(2\au{0}{1}\af{0}{0}^2+\frac{\ddot{u}_0}{\dot{u}_0}\af{0}{0}\af{1}{1}-\frac{\gamma}{\dot{u}_0}\big(v_0+\tj\big)\af{0}{0}^2+\frac{\ddot{u}_0}{\dot{u}_0}\af{0}{0}\af{0}{1}+\frac{\ddot{u}_0}{\dot{u}_0}\af{0}{1}\af{1}{0}-\udf\af{0}{1}^2\Bigg)\,,
\end{split}
\end{flalign}
which indeed agrees with Lemma \ref{t3.2.10}. The same argument for $(\alpha,\beta)=(1,0)$ leads to
\begin{flalign}\label{3.2.52}
\begin{split}
[\B_j,\A_j]_{2n+1,2n}&=2(-1)^j\Bigg[\af{0}{1}\Bigg(\au{0}{1}\af{0}{0}+\frac{\ddot{u}_0}{\dot{u}_0}\af{1}{1}-\udf\bigg(\frac{\gamma^2}{\dot{u}_0^2}\big(v_0+\tj\big)\af{0}{0}-\frac{\gamma\ddot{u}_0}{\dot{u}_0^2}\af{0}{1}+\frac{\gamma}{\dot{u}_0}\Big(\au{0}{0}\af{0}{1}-\au{0}{1}\af{0}{0}\Big)\bigg)\Bigg)\Bigg] \\
&=2(-1)^j\Bigg(2\au{0}{1}\af{0}{0}\af{0}{1}+\frac{\ddot{u}_0}{\dot{u}_0}\af{0}{1}\af{1}{1}-\frac{\gamma}{\dot{u}_0}\big(v_0+\tj\big)\af{0}{0}\af{0}{1}+\frac{\ddot{u}_0}{\dot{u}_0}\af{0}{1}^2-\au{0}{0}\af{0}{1}^2\Bigg)\,,
\end{split}
\end{flalign}
in accordance with Lemma \ref{t3.2.11}. Finally, in order to ensure that\eref{3.2.50} holds for $(\alpha,\beta)=(1,1)$, we need
\begin{flalign}\label{3.2.53}
\begin{split}
[\B_j,\A_j]_{2n+1,2n+1}&=-(-1)^j\Bigg[\af{0}{1}\Bigg(\au{0}{0}\af{0}{0}+\frac{\ddot{u}_0}{\dot{u}_0}\af{1}{0}-\udf\af{0}{1}\Bigg) \\
&\,\,\,\,\,\,\,\,\,\,\,\,\,\,\,\,\,\,\,+\af{0}{0}\Bigg(\au{0}{1}\af{0}{0}+\frac{\ddot{u}_0}{\dot{u}_0}\af{1}{1}-\udf\bigg(\frac{\gamma^2}{\dot{u}_0^2}\big(v_0+\tj\big)\af{0}{0}-\frac{\gamma\ddot{u}_0}{\dot{u}_0^2}\af{0}{1}+\frac{\gamma}{\dot{u}_0}\Big(\au{0}{0}\af{0}{1}-\au{0}{1}\af{0}{0}\Big)\bigg)\Bigg)\Bigg] \\
&=-(-1)^j\Bigg(\frac{\ddot{u}_0}{\dot{u}_0}\af{0}{1}\af{1}{0}-\udf\af{0}{1}^2+2\au{0}{1}\af{0}{0}^2+\frac{\ddot{u}_0}{\dot{u}_0}\af{0}{0}\af{1}{1}-\frac{\gamma}{\dot{u}_0}\big(v_0+\tj\big)\af{0}{0}^2+\frac{\ddot{u}_0}{\dot{u}_0}\af{0}{0}\af{0}{1}\Bigg)\,,
\end{split}
\end{flalign}
and this is consistent with Lemma \ref{t3.2.13}, finalizing the argument.
\end{proof}

\begin{rem}\label{t3.2.15}
Henceforth, we formally consider the zero curvature equation $[D_j,D_\xi]\Af(\xi)=0$, which follows from\eref{zce}, enabling us to derive the following upshot.
\end{rem}

\begin{prop}[Schlesinger Equations]\label{t3.2.16}
Let $i,j\in\{1,...,2m\}\subset\N$ such that $i\neq j\,$. The infinite size matrices $\A_j$ and $\B_j\coloneqq\B_{\tj}\,,\,$ that solve\eref{3.2.2} thereby describing the dynamics of the AWF, formally satisfy the following Schlesinger equations,
\begin{flalign}\label{3.2.54}
\begin{split}
\partial_j\A_i=\frac{[\A_i,\A_j]}{\tau_i-\tj}\,,\,\,\,\,\,\,\,\,\,\,\partial_j\A_j=\sum_{\substack{i=1\\i\neq j}}^{2m}\frac{[\A_i,\A_j]}{\tj-\tau_i}-[\A_j,\B_j]\,.
\end{split}
\end{flalign}
\end{prop}

\begin{proof}
All along this proof, $\,\Af\coloneqq\Af(\xi)$. We consider the covariant derivatives defined in\eref{zce} and we demonstrate this formal statement by proving that $[D_j,D_\xi]\Af=0$ $\Leftrightarrow$\eref{3.2.54}. Thus we shall begin with the following computations, $\forall\xi\in\I$,
\begin{flalign}\label{sep1}
\begin{split}
D_jD_\xi\Af&=\left(\partial_j+\cdj{j}\right)\left(\partial_\xi-\cdx\right)\Af \\
&=\partial_j\partial_\xi\Af-\B_\xi\partial_j\Af-\sum_{i=1}^{2m}\left(\frac{\big(\partial_j\A_i\big)}{\xi-\tau_i}\Af+\cdj{i}\partial_j\Af\right)-\frac{\A_j}{\big(\xi-\tj\big)^2}\Af+\cdj{j}\partial_\xi\Af-\cdj{j}\cdx\Af
\end{split}
\end{flalign}
and, still for all $\xi\in\I$, we also have
\begin{flalign}\label{sep2}
\begin{split}
D_\xi D_j\Af&=\left(\partial_\xi-\cdx\right)\left(\partial_j+\cdj{j}\right)\Af \\
&=\partial_\xi\partial_j\Af+\cdj{j}\partial_\xi\Af-\frac{\A_j}{\big(\xi-\tj\big)^2}\Af-\B_\xi\partial_j\Af-\sum_{i=1}^{2m}\cdj{i}\partial_j\Af-\cdx\cdj{j}\Af\,.
\end{split}
\end{flalign}
Enabling us to evaluate the commutator of the covariant derivatives, $\forall\xi\in\I$,
\begin{flalign}\label{sep3}
\begin{split}
[D_j,D_\xi]\Af&=-\sum_{i=1}^{2m}\frac{\big(\partial_j\A_i\big)}{\xi-\tau_i}\Af-\cdj{j}\cdx\Af+\cdx\cdj{j}\Af \\
&=\left(\rule{0cm}{0.75cm}\frac{1}{\xi-\tj}\left([\B_\xi,\A_j]+\sum_{i=1}^{2m}\frac{[\A_i,\A_j]}{\xi-\tau_i}\right)-\sum_{i=1}^{2m}\frac{\big(\partial_j\A_i\big)}{\xi-\tau_i}\right)\Af \\
&=\left(\frac{1}{\xi-\tj}\left([\B_\xi,\A_j]+\sum_{\substack{i=1\\i\neq j}}^{2m}\frac{[\A_i,\A_j]}{\xi-\tau_i}-\partial_j\A_j\right)-\sum_{\substack{i=1\\i\neq j}}^{2m}\frac{\big(\partial_j\A_i\big)}{\xi-\tau_i}\right)\Af=0\,.
\end{split}
\end{flalign}
Whence\eref{3.2.54} amounts to two different equations which together are equivalent to this relation, the former being
\begin{flalign}\label{sep4}
\begin{split}
\partial_j\A_j=[\B_\xi,\A_j]+\sum_{\substack{i=1\\i\neq j}}^{2m}\frac{[\A_i,\A_j]}{\xi-\tau_i}-\big(\xi-\tj\big)\sum_{\substack{i=1\\i\neq j}}^{2m}\frac{\big(\partial_j\A_i\big)}{\xi-\tau_i}\,,\quad\forall\xi\in\I\,,
\end{split}
\end{flalign}
whilst the latter one is, for any $i\in\{1,...,2m\}\,,\,\,i\neq j$,
\begin{flalign}\label{sep5}
\begin{split}
\partial_j\A_i=\frac{[\A_i,\A_j]}{\xi-\tau_j}+\frac{\xi-\tau_i}{\xi-\tj}\left([\B_\xi,\A_j]+\sum_{\substack{k=1\\k\neq i,j}}^{2m}\frac{[\A_k,\A_j]}{\xi-\tau_k}-\partial_j\A_j\right)-\big(\xi-\tau_i\big)\sum_{\substack{k=1\\k\neq i,j}}^{2m}\frac{\big(\partial_j\A_k\big)}{\xi-\tau_k}\,,\quad\forall\xi\in\I\,.
\end{split}
\end{flalign}
Ultimately, the $\xi\rightarrow\tj$ limit of\eref{sep4} leads to the second equation of\eref{3.2.54}. Meanwhile, the $\xi\rightarrow\tau_i$ limit of\eref{sep5} reads $\big(\tau_i-\tj\big)\partial_j\A_i=[\A_i,\A_j]$, ensuring our upshot.
\end{proof}

\begin{rem}\label{t3.2.17}
According to Proposition \ref{t3.2.14}, in order to rigorously prove Proposition \ref{t3.2.16} for an arbitrary component, it suffices to show that
\begin{flalign}\label{3.2.55}
\begin{split}
\comps{\partial_j\A_i}{2n+\alpha}{2p+\beta}\big|_{p<n}=\comps{\frac{[\A_i,\A_j]}{\tau_i-\tj}}{2n+\alpha}{2p+\beta}\,,\,\,\,\,\,\,\,\,\,\,\comps{\partial_j\A_j}{2n+\alpha}{2p+\beta}\big|_{p<n}=\comps{\sum_{\substack{i=1\\i\neq j}}^{2m}\frac{[\A_i,\A_j]}{\tj-\tau_i}-[\A_j,\B_j]}{2n+\alpha}{2p+\beta}\,.
\end{split}
\end{flalign}
It is worth pointing out that, since Propositions \ref{t3.2.3} and \ref{t3.2.7} completely determine these matrices, we have all the information necessary to verify\eref{3.2.55}. But the computations rapidly become rather lengthy so that it becomes quite difficult to compare expressions. Whence $p<n$ components remain a (purely) technical challenge.
\end{rem}

\section{Deformation of Tracy-Widom Distribution}%{Generalized Tracy-Widom Distribution}
\label{s4}

This section is dedicated to the finalization of the proof for our main result, which we formulate as follows in terms of the quantities defined in Definitions \ref{t1.1.1}.

\begin{tm}\label{t4.0.1}
If Condition \ref{t1.1.4} as well as Conditions \ref{t1.1.2}, \ref{t1.1.3}, 
\ref{t1.1.5}, \ref{t1.1.6} and \ref{t3.1.6} are satisfied, then the Fredholm determinant for $\I=[\tau,\infty)$ is explicitly given only in terms of $\psi(\xi)\coloneqq\varphi_\xi(0)$ and $\qcv(\xi)\coloneqq\big(\id-\Ko\big)^{-1}\psi(\xi)$ as follows,
\begin{flalign}\label{4.0.1}
\begin{split}
\F\Big([\tau,\infty)\Big)=\ex\left(\rule{0cm}{0.6444cm}\int_\tau^\infty\Bigg[\qcv_\sigma\Bigg(\udf\qcv''_\sigma-\qcv^3_\sigma+\frac{\udd}{\ud}\qcv_\sigma\bigg(\frac{\qcv'_\sigma}{\psi_\sigma}-\frac{\psi'_\sigma\qcv_\sigma}{\psi^2_\sigma}\bigg)\Bigg)-\udf\qcv'^2_\sigma\Bigg]d\sigma\right)\,,
\end{split}
\end{flalign}
where $f_\sigma\coloneqq f(\sigma)$, $f'_\sigma\coloneqq\frac{df}{d\sigma}(\sigma)$ and $f''_\sigma\coloneqq\frac{d^2f}{d\sigma^2}(\sigma)$ with $f$ representing either $\qcv$ or $\psi$. Moreover the auxiliary function $\qcv$ satisfies the following second order integro-differential equation,
\begin{flalign}\label{4.0.2}
\begin{split}
\qcv''_\tau=\frac{\gamma^2}{\ud^2}\big(v_0+\tau\big)\qcv_\tau+\frac{2\gamma}{\ud}\Bigg(\qcv^3_\tau-\frac{\gamma\udd}{\ud^2}\qcv_\tau\int_\tau^\infty\qcv^2_\sigma\,d\sigma\Bigg)-\frac{2\gamma^2\udd^2}{\ud^4}\Bigg(\frac{\qcv^3_\tau}{\psi^2_\tau}-\frac{\qcv^2_\tau}{\psi_\tau}\Bigg)+\frac{\gamma\udd}{\ud^2}\Bigg(\qcv'_\tau+2\frac{\qcv^2_\tau}{\psi^2_\tau}\psi'_\tau-4\frac{\qcv_\tau}{\psi_\tau}\qcv'_\tau\Bigg)\,.
\end{split}
\end{flalign}
Whenceforth, provided with $\psi$, $u$ and $\gamma$, the Fredholm determinant $\F\Big([\tau,\infty)\Big)$ is completely determined by the solution of\eref{4.0.2} whose behaviour as $\tau\rightarrow\infty$ coincides with $\psi$, more precisely $\qcv_\tau\xrightarrow{\tau\rightarrow\infty}\psi_\tau$.

Alternatively, the Fredholm determinant also admits the following explicit representation,
\begin{flalign}\label{4.0.3}
\begin{split}
\F\Big([\tau,\infty)\Big)=\ex\left(\rule{0cm}{0.65cm}-\frac{\gamma}{\ud}\int_\tau^\infty\big(\sigma-\tau\big)\left(\rule{0cm}{0.6444cm}\qcv^2_\sigma+\frac{\udd}{\gamma}\bigg(\frac{2\,\qcv_\sigma}{\psi_\sigma}-1\bigg)\Bigg[\qcv'^2_\sigma-\qcv_\sigma\Bigg(\qcv''_\sigma-\frac{\gamma}{\ud}\qcv^3_\sigma+\frac{\gamma\udd}{\ud^2}\qcv_\sigma\bigg(\frac{\qcv'_\sigma}{\psi_\sigma}-\frac{\psi'_\sigma\qcv_\sigma}{\psi^2_\sigma}\bigg)\Bigg)\Bigg]\right)d\sigma\right)\,.
\end{split}
\end{flalign}
\end{tm}

The proof has been greatly initiated in the previous sections, now we bring together the already established results and derive several new ones in order to yield Theorem \ref{t4.0.1}. Throughout \sref{s4}, we leave implicit $\I=[\tau,\infty)$.

\subsection{Integro-Differential Equation for the Auxiliary Wave Function}\label{s41}

We begin with the derivation of the integro-differential equation \eref{4.0.2}.

\begin{lem}\label{t4.1.1}
We can express $\au{0}{1}$, most notably in terms of $\au{0}{0}$ and $\au{1}{0}$, as
\begin{flalign}\label{4.1.1}
\begin{split}
2\au{0}{1}=\frac{\gamma}{\ud}\bigg(\au{0}{0}^2+\frac{\udd}{\ud}\au{1}{0}\bigg)-\qcv^2(\tau)\,.
\end{split}
\end{flalign}
\end{lem}

\begin{proof}
According to Definition \ref{t2.2.25}, we have
\begin{flalign}\label{4.1.2}
\begin{split}
\au{0}{1}&=\ipr{\psi}{\Pi_{\I}\Big([\rho,\Po]+\Po\rho\Big)\psi}=\ipr{\psi}{\Big(\Pi_{\I}[\rho,\Po]+[\Pi_{\I},\Po]\rho\Big)\psi}+\ipr{\psi}{\Po\Pi_{\I}\rho\psi} \\
&=\ipr{\psi}{\Big(\Pi_{\I}[\rho,\Po]+[\Pi_{\I},\Po]\rho\Big)\psi}-\ipr{\Po\psi}{\Pi_{\I}\rho\psi}=\ipr{\psi}{\Big(\Pi_{\I}[\rho,\Po]+[\Pi_{\I},\Po]\rho\Big)\psi}-\au{0}{1}\,,
\end{split}
\end{flalign}
where we used Properties \ref{t2.1.3} and \ref{t2.2.26}. Besides, notice that, for all $f\in\sob^1(\R)$ and $g\in\leb^2(\R)$, Properties \ref{t2.1.4} and \ref{t2.2.21} together entail
\begin{flalign}\label{4.1.3}
\begin{split}
\ipr{g}{[\Po,\Pi_{\I}]f}=-\ipr{g}{\dl{\tau}f}=g(\tau)f(\tau)\,.
\end{split}
\end{flalign}
Whilst Lemmas \ref{t2.2.32} for $n=0$ and \ref{t2.2.22} bring forth
\begin{flalign}\label{4.1.4}
\begin{split}
[\Po,\rho]=-\frac{\gamma}{\ud}\rho\Gamma\rho-\partial_\tau\rho=-\frac{\gamma}{\ud}\rho\Gamma\rho-\Ro\dl{\tau}\rho\,,
\end{split}
\end{flalign}
where Notation \ref{t2.2.29} was utilized. Thus we obtain
\begin{flalign}\label{4.1.5}
\begin{split}
2\au{0}{1}&=\frac{\gamma}{\ud}\ipr{\psi}{\Pi_{\I}\rho\Gamma\rho\psi}+\ipr{\psi}{\Pi_{\I}\Ro\dl{\tau}\rho\psi}+\ipr{\psi}{\dl{\tau}\rho\psi} \\
&=\frac{\gamma}{\ud}\ipr{\psi}{\Pi_{\I}\rho\bigg(\T_{\psi}\Pi_{\I}+\frac{\udd}{\ud}\Ko\bigg)\rho\psi}+\ipr{\psi}{\Pi_{\I}\Ro\dl{\tau}\rho\psi}+\ipr{\psi}{\Pi_{\I}\dl{\tau}\rho\psi} \\
&=\frac{\gamma}{\ud}\bigg(\ipr{\psi}{\Pi_{\I}\rho\psi}\ipr{\psi}{\Pi_{\I}\rho\psi}+\frac{\udd}{\ud}\ipr{\psi}{\Pi_{\I}\Ro\rho\psi}\bigg)+\ipr{\psi}{\Pi_{\I}\rho\dl{\tau}\rho\psi} \\
&=\frac{\gamma}{\ud}\bigg(\au{0}{0}^2+\frac{\udd}{\ud}\au{1}{0}\bigg)+\ipr{\rho\psi}{\Pi_{\I}\dl{\tau}\rho\psi}=\frac{\gamma}{\ud}\bigg(\au{0}{0}^2+\frac{\udd}{\ud}\au{1}{0}\bigg)-\big(\rho\psi\big)(\tau)\big(\rho\psi\big)(\tau)\,,
\end{split}
\end{flalign}
where\eref{2.2.32}, Property \ref{t2.2.9} and Lemma \ref{t2.2.12} were employed.
\end{proof}

\begin{lem}\label{t4.1.2}
We have $\partial_\tau\au{0}{0}=-\qcv^2(\tau)$.
\end{lem}

\begin{proof}
We directly compute
\begin{flalign}\label{4.1.6}
\begin{split}
\partial_\tau\ipr{\psi}{\Pi_{\I}\rho\psi}=\ipr{\psi}{\dl{\tau}\rho\psi}+\ipr{\psi}{\Pi_{\I}\Ro\dl{\tau}\rho\psi}=\ipr{\rho\psi}{\Pi_{\I}\dl{\tau}\rho\psi}\,,
\end{split}
\end{flalign}
so that\eref{4.1.3} leads to the result.
\end{proof}

Henceforth we are equipped to discuss $\frac{d^2\qcv}{d\tau^2}(\tau)\eqqcolon\qcv''(\tau)\coloneqq\af{0}{0}''(\tau)$.

\begin{lem}\label{t4.1.3}
Using Notation \ref{t3.1.16}, the following differential equation is obeyed by the AWF,
\begin{flalign}\label{4.1.7}
\begin{split}
\qcv''=\frac{\gamma^2}{\ud^2}\big(v_0+\tau\big)\qcv+\frac{2\gamma}{\ud}\bigg(\qcv^3-\frac{\gamma\udd}{\ud^2}\qcv\big(\au{0}{0}+\au{1}{0}\big)\bigg)-\frac{2\gamma^2\udd^2}{\ud^4}\bigg(\af{1}{0}+\af{2}{0}\bigg)-\frac{\gamma\udd}{\ud^2}\bigg(\qcv'+2\af{1}{0}'\bigg)\,.
\end{split}
\end{flalign}
\end{lem}

\begin{proof}
To begin with, we recall that Corollary \ref{t2.2.28} states
\begin{flalign}\label{4.1.8}
\begin{split}
\af{0}{\alpha}'=\af{0}{\alpha+1}-\frac{\gamma}{\ud}\bigg(\au{0}{\alpha}\qcv+\frac{\udd}{\ud}\af{1}{\alpha}\bigg)\,,
\end{split}
\end{flalign}
meanwhile Proposition \ref{t2.2.33} further entails
\begin{flalign}\label{4.1.9}
\begin{split}
\af{1}{0}'=\af{1}{1}-\frac{\gamma}{\ud}\bigg(\au{0}{0}\af{1}{0}+\big(\au{1}{0}+\au{0}{0}\big)\qcv+\frac{\udd}{\ud}\big(\af{1}{0}+2\af{2}{0}\big)\bigg)\,.
\end{split}
\end{flalign}
Now the idea is to use the closure relation given by Lemma \ref{t2.2.34} in order to obtain a differential equation for $\qcv''$. To be more precise, we equate\eref{2.2.53}, Lemma \ref{t2.2.34}, and\eref{4.1.8} for $\alpha=1$, and then we inject\eref{4.1.9} as well as\eref{4.1.8} for $\alpha=0$ in order to obtain a relation only in terms of the $\af{p}{0}$, for $p=0,1,2$, resulting in\eref{4.1.7}. Namely
\begin{flalign}\label{4.1.10}
\begin{split}
\af{0}{2}&=\frac{\gamma^2}{\dot{u}_0^2}\big(v_0+\tau\big)\qcv-\frac{\gamma\ddot{u}_0}{\dot{u}_0^2}\af{0}{1}+\frac{\gamma}{\dot{u}_0}\Big(\au{0}{0}\af{0}{1}-\au{0}{1}\qcv\Big)=\af{0}{1}'+\frac{\gamma}{\ud}\bigg(\au{0}{1}\qcv+\frac{\udd}{\ud}\af{1}{1}\bigg) \\
&=\frac{\gamma^2}{\dot{u}_0^2}\big(v_0+\tau\big)\qcv-\frac{\gamma\ddot{u}_0}{\dot{u}_0^2}\Bigg(\qcv'+\frac{\gamma}{\ud}\bigg(\au{0}{0}\qcv+\frac{\udd}{\ud}\af{1}{0}\bigg)\Bigg)+\frac{\gamma}{\dot{u}_0}\Bigg[\au{0}{0}\Bigg(\qcv'+\frac{\gamma}{\ud}\bigg(\au{0}{0}\qcv+\frac{\udd}{\ud}\af{1}{0}\bigg)\Bigg)-\au{0}{1}\qcv\Bigg] \\
&=\qcv''+\frac{\gamma}{\ud}\Bigg(\big(\partial_\tau\au{0}{0}\big)\qcv+\au{0}{0}\qcv'+\frac{\udd}{\ud}\af{1}{0}'\Bigg) \\
&\,\,\,\,\,\,\,\,\,\,\,\,\,\,\,\,\,\,\,\,\,\,\,\,\,\,\,\,\,\,\,\,\,\,\,\,\,\,\,\,\,\,\,\,\,\,\,\,\,\,\,\,\,\,\,\,\,\,\,\,\,\,\,\,+\frac{\gamma}{\ud}\Bigg[\au{0}{1}\qcv+\frac{\udd}{\ud}\Bigg(\af{1}{0}'+\frac{\gamma}{\ud}\bigg(\au{0}{0}\af{1}{0}+\big(\au{1}{0}+\au{0}{0}\big)\qcv+\frac{\udd}{\ud}\big(\af{1}{0}+2\af{2}{0}\big)\bigg)\Bigg)\Bigg]\,,
\end{split}
\end{flalign}
which is indeed leading to
\begin{flalign}\label{4.1.11}
\begin{split}
\qcv''=\frac{\gamma^2}{\dot{u}_0^2}\big(v_0+\tau+\au{0}{0}\big)\qcv+\frac{\gamma}{\ud}\Bigg(-2\au{0}{1}\qcv-\big(\partial_\tau\au{0}{0}\big)\qcv-&\frac{\gamma\udd}{\ud^2}\qcv\big(2\au{0}{0}+\au{1}{0}\big)\Bigg) \\
&-\frac{2\gamma^2\udd^2}{\ud^4}\bigg(\af{1}{0}+\af{2}{0}\bigg)-\frac{\gamma\udd}{\ud^2}\bigg(\qcv'+2\af{1}{0}'\bigg)\,.
\end{split}
\end{flalign}
Finally Lemmas \ref{t4.1.1} and \ref{t4.1.2} yield the result.
\end{proof}

Notice that we began with Proposition \ref{t2.2.33}, (Corollary \ref{t2.2.28} is essentially a special case of this proposition), which provides us with a differential equation relating the AWF. Then, in Lemma \ref{t4.1.3}, we used our first closure relation, Lemma \ref{t2.2.34}, in order to obtain a second order differential equation involving solely the $\af{n}{0}$, $n\in\Z_{\geq0}$, and not the $\af{n}{1}$ or $\af{n}{2}$ anymore. Next we finish closing this equation for $\qcv$ using our second closure relation Lemma \ref{t3.1.9}.

One may observe that this procedure is generic, for any $n\in\Z_{\geq0}$ and $a\in\{0,1,2\}$, one may convert $\af{n}{a}$ into $\af{p}{0}$, with $p=n,n+1,n+2$, using Propositions \ref{t2.2.33} and \ref{t2.2.36} (of which Lemma \ref{t2.2.34} is a special case). And one may in turn express the $\af{p}{0}$ in terms of $\af{0}{0}$ and $\psi$ only through Lemma \ref{t3.1.9}.

Actually, to completely close\eref{4.1.7}, it will remain to express $\au{0}{0}+\au{1}{0}$ in terms of $\qcv$, and this will be what amounts to the appearance of the integral in\eref{4.0.2}.

\begin{lem}\label{t4.1.4}
The second order integro-differential equation\eref{4.0.2} is indeed satisfied by $\qcv\coloneqq\af{0}{0}$.
\end{lem}

\begin{proof}
Firstly we note that Property \ref{t2.2.9} and Lemma \ref{t2.2.12} bring forth
\begin{flalign}\label{4.1.12}
\begin{split}
\au{0}{0}+\au{1}{0}=\ipr{\psi}{\Pi_{\I}\big(\id+\Ro\big)\rho\psi}=\ipr{\rho\psi}{\Pi_{\I}\rho\psi}=\int_\R\qcv(\sigma)\,\ind{\I}(\sigma)\,\qcv(\sigma)\,d\sigma=\int_\tau^\infty\qcv^2(\sigma)\,d\sigma\,.
\end{split}
\end{flalign}
Whilst Lemma \ref{t3.1.9}, for $p=0$, allows us to evaluate
\begin{flalign}\label{4.1.13}
\begin{split}
\af{1}{0}+\af{2}{0}=\gao\qcv+\ga{2}\qcv=\frac{\qcv^3}{\psi^2}-\frac{\qcv^2}{\psi}\,,
\end{split}
\end{flalign}
and
\begin{flalign}\label{4.1.14}
\begin{split}
\af{1}{0}'=2\frac{\qcv}{\psi}\qcv'-\frac{\qcv^2}{\psi^2}\psi'-\qcv'\,.
\end{split}
\end{flalign}
Whence, injecting\eref{4.1.12} in the second term of\eref{4.1.7},\eref{4.1.13} in the the third, and\eref{4.1.14} in the fourth one yields
\begin{flalign}\label{4.1.15}
\begin{split}
\qcv''=\frac{\gamma^2}{\ud^2}\big(v_0+\tau\big)\qcv+\frac{2\gamma}{\ud}\bigg(\qcv^3-\frac{\gamma\udd}{\ud^2}\qcv\int_\tau^\infty\qcv^2(\sigma)\,d\sigma\bigg)-\frac{2\gamma^2\udd^2}{\ud^4}\bigg(\frac{\qcv^3}{\psi^2}-\frac{\qcv^2}{\psi}\bigg)+\frac{\gamma\udd}{\ud^2}\bigg(\qcv'+2\frac{\qcv^2}{\psi^2}\psi'-4\frac{\qcv}{\psi}\qcv'\bigg)\,,
\end{split}
\end{flalign}
so that one recognizes\eref{4.0.2}.
\end{proof}

\subsection{Resulting Fredholm Determinant}\label{s42}

Now we turn to the establishment of\eref{4.0.1} and\eref{4.0.3}, and justify the behaviour of $\qcv$ as $\tau\rightarrow\infty$ stated in Theorem \ref{t4.0.1}. We achieve this with five Lemmas, the two first ones are dedicated to the evaluation of $\Ha_1(\tau)$ and $\frac{d\Ha_1}{d\tau}(\tau)$ in terms of $\qcv$. The third Lemma relates $\F\big([\tau,\infty)\big)$ to $\Ha_1(\tau)$ which yields\eref{4.0.1}, whilst the fourth one establishes\eref{4.0.3} through the evaluation of $\F\big([\tau,\infty)\big)$ in terms of $\frac{d\Ha_1}{d\tau}(\tau)$. Finally the last Lemma comments the asymptotic behaviour of $\qcv$.

\begin{lem}\label{t4.2.1}
The first order Hamiltonian can be expressed solely in terms of $\qcv$, $\psi$, $u$ and $\gamma$ as follows,
\begin{flalign}\label{4.2.1}
\begin{split}
\Ha_\tau\coloneqq\Ha_1(\tau)=\qcv'^2-\qcv\Bigg(\qcv''-\frac{\gamma}{\ud}\qcv^3+\frac{\gamma\udd}{\ud^2}\qcv\bigg(\frac{\qcv'}{\psi}-\frac{\psi'\qcv}{\psi^2}\bigg)\Bigg)\,.
\end{split}
\end{flalign}
\end{lem}

\begin{proof}
We follow the general procedure described in the second paragraph below the proof of Lemma \ref{t4.1.3}, namely, using Proposition \ref{t2.2.33} as well as Lemma \ref{t3.1.9}, we are led to
\begin{flalign}\label{4.2.2}
\begin{split}
\af{0}{1}=\qcv'+\frac{\gamma}{\ud}\au{0}{0}\qcv+\frac{\gamma\udd}{\ud^2}\af{1}{0}=\qcv'+\frac{\gamma}{\ud}\au{0}{0}\qcv+\frac{\gamma\udd}{\ud^2}\gao\qcv\,,
\end{split}
\end{flalign}
which is in turn entailing
\begin{flalign}\label{4.2.3}
\begin{split}
\af{0}{1}'=\qcv''+\frac{\gamma}{\ud}\Big(\au{0}{0}\qcv'-\qcv^3\Big)+\frac{\gamma\udd}{\ud^2}\bigg(2\frac{\qcv}{\psi}\qcv'-\frac{\qcv^2}{\psi^2}\psi'-\qcv'\bigg)\,,
\end{split}
\end{flalign}
where we utilized Lemma \ref{t4.1.2}. Hence, injecting these relations in Lemma \ref{t3.1.4}, $\Ha_\tau=\qcv'\af{0}{1}-\af{0}{1}'\qcv$, leads to\eref{4.2.1}.
\end{proof}

\begin{lem}\label{t4.2.2}
The first order Hamiltonian also satisfies
\begin{flalign}\label{4.2.4}
\begin{split}
\frac{d\Ha_\tau}{d\tau}=-\frac{\gamma^2}{\ud^2}\left(\rule{0cm}{0.6444cm}\qcv^2+\frac{\udd}{\gamma}\bigg(\frac{2\qcv}{\psi}-1\bigg)\Bigg[\qcv'^2-\qcv\Bigg(\qcv''-\frac{\gamma}{\ud}\qcv^3+\frac{\gamma\udd}{\ud^2}\qcv\bigg(\frac{\qcv'}{\psi}-\frac{\psi'\qcv}{\psi^2}\bigg)\Bigg)\Bigg]\right)\,.
\end{split}
\end{flalign}
\end{lem}

\begin{proof}
Firstly we observe that, $\forall f\in\sob^1(\R)$,
\begin{flalign}\label{4.2.5}
\begin{split}
\frac{d}{d\tau}\Ro f(\tau)&=\big(\partial_\tau\Ro\big)f(\tau)+\Po\Ro f(\tau)=\Ro\dl{\tau}\rho f(\tau)+[\Po,\Ro]f(\tau)+\Ro\Po f(\tau) \\
&=\Ro\dl{\tau}\rho f(\tau)-\frac{\gamma}{\ud}\Bigg(\rho\psi(\tau)\ipr{\psi}{\Pi_{\I}\rho f}+\frac{\udd}{\ud}\Ro\rho f(\tau)\Bigg)-\Ro\dl{\tau}\rho f(\tau)+\Ro\Po f(\tau) \\
&=-\frac{\gamma}{\ud}\Bigg(\rho\psi(\tau)\ipr{\rho\psi}{\Pi_{\I}f}+\frac{\udd}{\ud}\big(\Ro+\Ro^2\big)f(\tau)\Bigg)+\Ro\Po f(\tau) \\
&=\int_\R\Bigg[\Ri_1(\tau,\xi)\,\partial_\xi f(\xi)-\frac{\gamma}{\ud}\Bigg(\qcv(\tau)\,\qcv(\xi)\,\ind{\I}(\xi)\,f(\xi)+\frac{\udd}{\ud}\Big(\Ri_1(\tau,\xi)+\Ri_2(\tau,\xi)\Big)f(\xi)\Bigg)\Bigg]\,d\xi \\
&=\int_\R\Bigg[-\partial_\xi\Ri_1(\tau,\xi)-\frac{\gamma}{\ud}\Bigg(\qcv(\tau)\,\qcv(\xi)\,\ind{\I}(\xi)+\frac{\udd}{\ud}\Big(\Ri_1(\tau,\xi)+\Ri_2(\tau,\xi)\Big)\Bigg)\Bigg]\,f(\xi)\,d\xi\,.
\end{split}
\end{flalign}
where we employed Lemma \ref{t2.2.22}, Lemma \ref{t2.2.11} together with Lemma \ref{t2.2.24} as well as Property \ref{t2.2.9} and Lemma \ref{t2.2.12}. Thus, this provides us with
\begin{flalign}\label{4.2.6}
\begin{split}
\frac{d}{d\tau}\Ri_1(\tau,\xi)=-\frac{\gamma}{\ud}\Bigg(\qcv(\tau)\,\qcv(\xi)\,\ind{\I}(\xi)+\frac{\udd}{\ud}\Big(\Ri_1(\tau,\xi)+\Ri_2(\tau,\xi)\Big)\Bigg)-\partial_\xi\Ri_1(\tau,\xi)\,,
\end{split}
\end{flalign}
In turn leading to
\begin{flalign}\label{4.2.7}
\begin{split}
\frac{d}{d\tau}\Ri_1(\tau,\tau)=-\frac{\gamma}{\ud}\Bigg(\qcv(\tau)\,\qcv(\tau)\,\ind{\I}(\tau)+\frac{\udd}{\ud}\Big(\Ri_1(\tau,\tau)+\Ri_2(\tau,\tau)\Big)\Bigg)\,.
\end{split}
\end{flalign}
Thereby, according to Definition \ref{t3.1.1}, the first order Hamiltonian is related to the second order one by
\begin{flalign}\label{4.2.8}
\begin{split}
\frac{d\Ha_\tau}{d\tau}=-\frac{\gamma^2}{\ud^2}\Bigg(\qcv^2(\tau)+\frac{\udd}{\gamma}\Big(\Ha_\tau+\Ha_2(\tau)\Big)\Bigg)\,.
\end{split}
\end{flalign}
Then Lemma \ref{t3.1.15} enables us to close this relation for $\Ha_\tau$, namely
\begin{flalign}\label{4.2.9}
\begin{split}
\frac{d\Ha_\tau}{d\tau}=-\frac{\gamma^2}{\ud^2}\Bigg[\qcv^2(\tau)+\frac{\udd}{\gamma}\Bigg(\Ha_\tau+2\gao\Ha_\tau\Bigg)\Bigg]=-\frac{\gamma^2}{\ud^2}\Bigg[\qcv^2(\tau)+\frac{\udd}{\gamma}\bigg(\frac{2\qcv}{\psi}-1\bigg)\Ha_\tau\Bigg]\,,
\end{split}
\end{flalign}
so that writing $\Ha_\tau$ explicitly through Lemma \ref{t4.2.1} yields the result.
\end{proof}

\begin{lem}\label{t4.2.3}
The Fredholm determinant $\F\big([\tau,\infty)\big)$ is given by\eref{4.0.1}.
\end{lem}

\begin{proof}
We start from Proposition \ref{t3.1.2} which brings forth
\begin{flalign}\label{4.2.10}
\begin{split}
\F\Big([\tau,\infty)\Big)=\ex\left(h(t)-\udf\int_\tau^t\Ha_\sigma \,d\sigma\right)\,,
\end{split}
\end{flalign}
for some parameter $t$ and function $h(t)$. Noticing that
\begin{flalign}\label{4.2.11}
\begin{split}
\lim_{\tau\rightarrow\infty}\F\Big([\tau,\infty)\Big)=1\,,
\end{split}
\end{flalign}
we choose $t\rightarrow\infty$ and obtain
\begin{flalign}\label{4.2.12}
\begin{split}
\F\Big([\tau,\infty)\Big)=\ex\left(-\udf\int_\tau^\infty\Ha_\sigma\,d\sigma\right)\,.
\end{split}
\end{flalign}
As a short comment,\eref{4.2.11} is ensured by $\I\xrightarrow{\tau\rightarrow\infty}\emptyset$, which gives $\id-\Ko\xrightarrow{\tau\rightarrow\infty}\id$. Our final result is then directly brought by Lemma \ref{t4.2.1}.
\end{proof}

\begin{lem}\label{t4.2.4}
The Fredholm determinant $\F\big([\tau,\infty)\big)$ is expressed as\eref{4.0.3}.
\end{lem}

\begin{proof}
Integrating by parts, we may observe, under Condition \ref{t1.1.2},
\begin{flalign}\label{4.2.13}
\begin{split}
\int_\tau^\infty\big(\sigma-\tau\big)\,\frac{d\Ha_\sigma}{d\sigma}\,d\sigma=-\int_\tau^\infty\Ha_\sigma\,d\sigma\,.
\end{split}
\end{flalign}
Thus, combining this with\eref{4.2.12} entails
\begin{flalign}\label{4.2.14}
\begin{split}
\F\Big([\tau,\infty)\Big)=\ex\left(\udf\int_\tau^\infty\big(\sigma-\tau\big)\,\frac{d\Ha_\sigma}{d\sigma}\,d\sigma\right)\,.
\end{split}
\end{flalign}
Henceforth\eref{4.0.3} amounts to Lemma \ref{t4.2.2}.
\end{proof}

\begin{lem}\label{t4.2.5}
The zeroth order AWF admit the following asymptotics, $\forall a\in\{0,1,2\}\,$,
\begin{flalign}\label{4.2.15}
\begin{split}
\lim_{\tau\rightarrow\infty}\af{0}{a}(\xi)=\psi^{(a)}(\xi)\,.
\end{split}
\end{flalign}
\end{lem}

\begin{proof}
Notice that, since $\rho=\sum_{r=0}^\infty\big(\Ko_0\Pi_{\I}\big)^r$, it follows that $\rho\xrightarrow{\tau\rightarrow\infty}\id$, yielding $\af{0}{a}(\xi)\coloneqq\rho\psi^{(a)}(\xi)\xrightarrow{\tau\rightarrow\infty}\psi^{(a)}(\xi)$, which stands $\forall a\in\{0,1,2\}$.
\end{proof}

Lemmas \ref{t4.1.4}, \ref{t4.2.3}, \ref{t4.2.4} and \ref{t4.2.5} all together conclude the proof of Theorem \ref{t4.0.1}, hereby finalizing the central part of the present paper.

%\newpage

\appendix

\section{On the Characterization of Condition \ref{t1.1.4}}\label{sA}

In this appendix we try to characterize Condition \ref{t1.1.4} with three examples of generic cases in which it is systematically satisfied. More precisely, we give three conditions, each of which is sufficient to ensure Condition \ref{t1.1.4}.

\begin{cond}\label{ta1}
When $v(x;\xi)=v(x)$ does not depend on $\xi$, defining $\omega\coloneqq v(x)+\xi$, we require that there exists a solution to
\begin{flalign}\label{a1}
\begin{split}
\Big(\big(\partial_xv(x)\big)^2\partial_\omega^2+\big(\partial_x^2v(x)\big)\partial_\omega-\omega\Big)\phi(\omega)=0\,.
\end{split}
\end{flalign}
\end{cond}

\begin{cond}\label{ta2}
We require that $\exists f,\phi$ any two functions such that :
\begin{enumerate}
\item $f(x)\neq0\,$, $\,\forall x\in\R\,$;
\item $\exists\phi^{-1}\,$ any function satisfying $\,\phi\circ\phi^{-1}:x\mapsto x\,$;
\item $\varphi_\xi(x)=f(x)\phi(\xi)\,$, $\,\forall x,\xi\in\R\,$;
\item The mapping $\,\varrho_\xi:f(x)\mapsto f(x)\phi(\xi)\,$ defines a flow.
\end{enumerate}
\end{cond}

\begin{cond}\label{ta3}
We demand that $\exists f,\phi$ any two functions such that :
\begin{enumerate}
\item $\exists\phi^{-1}\,$ any function satisfying $\,\phi\circ\phi^{-1}:x\mapsto x\,$;
\item $\varphi_\xi(x)=f(x)+\phi(\xi)\,$, $\,\forall x,\xi\in\R\,$;
\item The function $\phi$ is linear.
\end{enumerate}
\end{cond}

\begin{lem}\label{ta4}
If Condition \ref{ta1} is satisfied, then Condition \ref{t1.1.4} is satisfied too with $u=v$ and $\gamma=1$.
\end{lem}

\begin{proof}
If there exists a solution to\eref{a1}, then it satisfies
\begin{flalign}\label{a2}
\begin{split}
\partial_x^2\phi\big(v(x)+\xi\big)=\partial_x\big(\partial_xv(x)\big)\partial_\omega\phi(\omega)=\big(\partial_xv(x)\big)^2\partial_\omega^2\phi(\omega)+\big(\partial_x^2v(x)\big)\partial_\omega\phi(\omega)=\omega\phi(\omega)=\Big(v(x)+\xi\Big)\phi\big(v(x)+\xi\big)\,,
\end{split}
\end{flalign}
and hence $\varphi_\xi(x)=\phi\big(v(x)+\xi\big)$ indeed solves\eref{1.1.1} too.
\end{proof}

\begin{lem}\label{ta5}
If Condition \ref{ta2} is satisfied, then Condition \ref{t1.1.4} is satisfied too with $u=\phi^{-1}\circ f$ and $\gamma=1$.
\end{lem}

\begin{proof}
Since $\varrho_\xi$ defines a flow, it follows that $\varrho_{\xi+\zeta}=\varrho_\xi\circ\varrho_\zeta$ and, since $f(x)\neq0\,$, $\,\forall x\in\R\,$, this entails $\phi(\xi+\zeta)=\phi(\xi)\phi(\zeta)$. Thus
\begin{flalign}\label{a3}
\begin{split}
\phi\big(\phi^{-1}\circ f(x)+\xi\big)=f(x)\phi(\xi)=\varphi_\xi(x)\,,
\end{split}
\end{flalign}
so that Condition \ref{t1.1.4} is indeed ensured with $u=\phi^{-1}\circ f$ and $\gamma=1$.
\end{proof}

\begin{lem}\label{ta6}
If Condition \ref{ta3} is satisfied, then Condition \ref{t1.1.4} is satisfied too with $u=\phi^{-1}\circ f$ and $\gamma=1$.
\end{lem}

\begin{proof}
Indeed, under Condition \ref{ta3}, we have
\begin{flalign}\label{a4}
\begin{split}
\phi\big(\phi^{-1}\circ f(x)+\xi\big)=f(x)+\phi(\xi)=\varphi_\xi(x)\,,
\end{split}
\end{flalign}
which ensures Condition \ref{t1.1.4} for $u=\phi^{-1}\circ f$ and $\gamma=1$.
\end{proof}

\section{An Application to the KPZ Equation}\label{sB}

We now briefly lay out a precise topic for which our endeavors may directly be applied. The following discussion revolves around the key observation that, substituting\eref{1.1.8} for $\varphi_\xi(x)$ in the kernel\eref{1.1.3} yields
\begin{flalign}\label{b1}
\begin{split}
\Ki(\xi,\zeta)=\int_{\R_+}\phi\big(u(x)+\gamma\xi\big)\phi\big(u(x)+\gamma\zeta\big)\,dx=\lim_{a\rightarrow\infty}\int_{u(0)}^{u(a)}\frac{\phi\big(u+\gamma\xi\big)\phi\big(u+\gamma\zeta\big)}{\dot{u}}\,du\,,
\end{split}
\end{flalign}
with $\dot{u}\coloneqq\frac{du}{dx}$, this is slightly formal because we do not relabel the integration variable, but this will be made more precise below. One may recognize the generic form of the so-called finite temperature kernel in the context of the KPZ equation. Hereby motivating us to formulate the following result.

\begin{prop}\label{tb1}
Let $\phi$ and $u$ be any twice differentiable functions, i.e. $u,\phi\in\operatorname{C}^2(\R)$, and $\gamma\in\R\char`\\\{0\}$. Upon probing the regularity Conditions \ref{t1.1.2}, \ref{t1.1.3}, \ref{t1.1.5}, \ref{t1.1.6} and \ref{t3.1.6} for $\varphi_\xi(x)\coloneqq\phi\big(u(x)+\gamma\xi\big)$ satisfying\eref{1.1.1}, Theorem \ref{t4.0.1} may be applied for any Fredholm determinant arising from a kernel of the following generic form,
\begin{flalign}\label{b2}
\begin{split}
\Ki(\xi,\zeta)=\lim_{a\rightarrow\infty}\int_{u(0)}^{u(a)}\frac{\phi\big(\lambda+\gamma\xi\big)\phi\big(\lambda+\gamma\zeta\big)}{f(\lambda)}\,d\lambda\,,
\end{split}
\end{flalign}
where $f(u)\coloneqq\frac{du}{dx}(x)$ depends solely on $u$ and parameters not related to the integration.
\end{prop}

\begin{proof}
Formulating\eref{b1} more carefully, we are led to
\begin{flalign}\label{b3}
\begin{split}
\Ki(\xi,\zeta)=\lim_{a\rightarrow\infty}\int_{u(0)}^{u(a)}\frac{\phi\big(\lambda+\gamma\xi\big)\phi\big(\lambda+\gamma\zeta\big)}{\dot{u}|_{u=\lambda}}\,d\lambda\,,
\end{split}
\end{flalign}
where we need to express $\dot{u}$ only in terms of $u$, which we then set to $\lambda$, in order to perform the integration. And, if $u$ is invertible, then this procedure is direct as we may write $x$ in terms of $u$. Henceforth this coincides with \ref{b2}, and Condition \ref{t1.1.4} is ensured by construction so that we are able to apply Theorem \ref{t4.0.1}.
\end{proof}

\paragraph{Specific Example.} We turn to a specific instance for which our results may be applied through Proposition \ref{tb1}.

Fredholm determinants are known to arise in the context of the KPZ equation \cite{28,29,Calabrese:2010EPL,6}, here we focus on the solution to this equation for the height distribution at a single point, which is expressed in terms of the Fredholm determinant whose related kernel is
\begin{flalign}\label{b4}
\begin{split}
\Ki(\xi,\zeta)=\int_{-\infty}^{+\infty}\frac{\Ai\big(\lambda+\xi\big)\Ai\big(\lambda+\zeta\big)}{c_1\ex\big(-c_2\lambda\big)+1}\,d\lambda\,,
\end{split}
\end{flalign}
where $c_1,c_2\in\R_+$ are two positive constants. This expression with the constants $c_1$ and $c_2$ written explicitly as well as a discussion with an emphasis on the relevance of the corresponding Fredholm determinant may be found, for instance, in \cite[eq.~(5.3)]{32}.

\begin{lem}\label{tb2}
The Fredholm determinant related to the kernel\eref{b4} belongs to the range of applications of Proposition \ref{tb1}, with $\phi=\Ai$ the Airy function, $\,\gamma=1\,$, as well as $u$ solving
\begin{flalign}\label{b5}
\begin{split}
\frac{du}{dx}(x)=-c_1\ex\Big(-c_2u(x)\Big)-1\,,\,\,\,\,\,\,\,\,\,\,\,\,\lim_{x\rightarrow0}u(x)=+\infty\,,\,\,\,\,\,\,\,\,\,\,\lim_{x\rightarrow+\infty}u(x)=-\infty\,,
\end{split}
\end{flalign}
and $v$ given by
\begin{flalign}\label{b6}
\begin{split}
v(x;\xi)=\frac{\partial_x^2\Ai\big(u(x)+\xi\big)}{\Ai\big(u(x)+\xi\big)}-\xi\,.
\end{split}
\end{flalign}
\end{lem}

\begin{proof}
The first step is to identify\eref{b4} with\eref{b2}, and to do so we rewrite\eref{b2} as follows,
\begin{flalign}\label{b7}
\begin{split}
\Ki(\xi,\zeta)=\lim_{a\rightarrow+\infty}\int_{u(a)}^{u(0)}\frac{\phi\big(\lambda+\gamma\xi\big)\phi\big(\lambda+\gamma\zeta\big)}{-f(\lambda)}\,d\lambda\,,
\end{split}
\end{flalign}
so that setting $\phi=\Ai$, $\gamma=1$ and $f(\lambda)=-c_1\ex\big(-c_2\lambda\big)-1$, we indeed retrieve the integrand of\eref{b4}. Then, to fully recover\eref{b4}, we identify the bounds of the integrals, which lead to the second and third relations of\eref{b5}. Provided with this identification, $\frac{du}{dx}=f(u)$ yields the first equation of\eref{b5}, and one notices that $\frac{du}{dx}(x)<0\,,\,\,\forall x\,$, is compatible with the boundary conditions. \\
Finally, the last subtle step is to find a function $v$ satisfying both\eref{1.1.1} and Condition \ref{t1.1.6}. Since $\varphi_\xi(x)=\Ai\big(u(x)+\xi\big)$ with $u$ satisfying\eref{b5},\eref{1.1.1} amounts to\eref{b6}, thus it remains to justify that such a $v$ does obey Condition \ref{t1.1.6}. And in order to achieve this, we shall be guided by the following observation, with $\dot{u}\coloneqq\frac{du}{dx}$ and $\omega\coloneqq u(x)+\xi$,
\begin{flalign}\label{b8}
\begin{split}
\partial_x^2\Ai\big(u(x)+\xi\big)=\dot{u}^2(x)\partial^2_\omega\Ai(\omega)+\frac{d^2u}{dx^2}(x)\partial_\omega\Ai(\omega)=\dot{u}^2(x)\Big(u(x)+\xi\Big)\Ai\big(u(x)+\xi\big)+\frac{d^2u}{dx^2}(x)\partial_\omega\Ai(\omega)\,,
\end{split}
\end{flalign}
where we used the defining relation, $\partial^2_\omega\Ai(\omega)=\omega\Ai(\omega)$, of the Airy function. Additionally, we may observe that, according to the first relation of\eref{b5},
\begin{flalign}\label{b9}
\begin{split}
\dot{u}^2(x)=c_1^2\ex\Big(-2c_2u(x)\Big)+2c_1\ex\Big(-c_2u(x)\Big)+1\,,
\end{split}
\end{flalign}
whilst, still using the same expression but twice,
\begin{flalign}\label{b10}
\begin{split}
\frac{d^2u}{dx^2}(x)=c_1c_2\dot{u}(x)\ex\Big(-c_2u(x)\Big)=-c_1c_2\Bigg(c_1\ex\Big(-2c_2u(x)\Big)+\ex\Big(-c_2u(x)\Big)\Bigg)\,.
\end{split}
\end{flalign}
Whenceforth we may employ $u(x)\xrightarrow{x\rightarrow0}+\infty$ together with $c_2>0$ in order to compute
\begin{flalign}\label{b11}
\begin{split}
&\lim_{x\rightarrow0}\,\frac{d^2u}{dx^2}(x)=0\,, \\
&\lim_{x\rightarrow0}\dot{u}^2(x)\Big(u(x)+\xi\Big)=\lim_{x\rightarrow0}u(x)+\xi\,.
\end{split}
\end{flalign}
Implying the following asymptotic behaviour for\eref{b8} as $x\rightarrow0$,
\begin{flalign}\label{b12}
\begin{split}
\lim_{x\rightarrow0}\partial_x^2\Ai\big(u(x)+\xi\big)=\lim_{x\rightarrow0}\Big(u(x)+\xi\Big)\Ai\big(u(x)+\xi\big)\,,
\end{split}
\end{flalign}
ultimately, injecting this in\eref{b6} does ensure $\partial_\xi v(x;\xi)\xrightarrow{x\rightarrow0}0$.
\end{proof}

\section{On the Integro-Differential Equation for any AWFs}\label{sC}

We derive a second order integro-differential equation for the $\af{k}{0}\,,\,\,k\in\{0,...,n\}\subset\Nz\,,$ with $n\in\N$. This is analogous to\eref{4.0.2} but for higher order AWFs, and we proceed with a method somewhat similar to \sref{s41}. Throughout appendix \sref{sC}, we leave implicit $\I=[\tau,\infty)$.

\begin{lem}\label{tc1}
Employing Notations \ref{t3.1.14} and \ref{t3.1.16}, we have $\partial_\tau\au{n}{0}=-(n+1)\fa{n}\qcv^2\,$ for any $n\in\Nz$.
\end{lem}

\begin{proof}
Lemma \ref{t2.2.30} together with\eref{4.1.3} are yielding, $\forall n\in\Nz\,$,
\begin{flalign}\label{c1}
\begin{split}
\partial_\tau\au{n}{0}&=\partial_\tau\ipr{\psi}{\Pi_{\I}\Ro^n\rho\psi}=\ipr{\psi}{\dl{\tau}\Ro^n\rho\psi}+\ipr{\psi}{\Pi_{\I}\Ro\dl{\tau}\Ro^n\rho\psi}+\sum_{k=0}^{n-1}\ipr{\psi}{\Pi_{\I}\big(\Ro^{n-k+1}+\Ro^{n-k}\big)\dl{\tau}\Ro^k\rho\psi} \\
&=\ipr{\rho\psi}{\Pi_{\I}\dl{\tau}\Ro^n\rho\psi}+\sum_{k=0}^{n-1}\ipr{\big(\Ro^{n-k+1}+\Ro^{n-k}\big)\psi}{\Pi_{\I}\dl{\tau}\Ro^k\rho\psi}=-\qcv\af{n}{0}-\sum_{k=0}^{n-1}\Big(\big(\Ro^{n-k+1}+\Ro^{n-k}\big)\psi\Big)\af{k}{0}\,.
\end{split}
\end{flalign}
Meanwhile one may notice that, $\forall n\in\N\,$, $\,\,\Ro^n\psi=\Ro^{n-1}\rho\psi-\Ro^{n-1}\psi=\af{n-1}{0}-\Ro^{n-1}\psi\,$, so that one may obtain by iteration, $\forall n\in\Nz$,
\begin{flalign}\label{c2}
\begin{split}
\Ro^n\psi=-\sum_{k=1}^n(-1)^k\af{n-k}{0}+(-1)^n\psi\,.
\end{split}
\end{flalign}
Whenceforth we are able to evaluate the sum
\begin{flalign}\label{c3}
\begin{split}
\sum_{k=0}^{n-1}\big(\Ro^{n-k+1}+\Ro^{n-k}\big)\psi&=\sum_{k=0}^{n-1}\Bigg((-1)^{n-k+1}\psi-\sum_{l=1}^{n-k+1}(-1)^l\af{n-k-l+1}{0}+(-1)^{n-k}\psi-\sum_{l=1}^{n-k}(-1)^l\af{n-k-l}{0}\Bigg) \\
&=\sum_{k=0}^{n-1}\Bigg(-\sum_{l=0}^{n-k}(-1)^{l+1}\af{n-k-l}{0}-\sum_{l=1}^{n-k}(-1)^l\af{n-k-l}{0}\Bigg)=\sum_{k=0}^{n-1}\af{n-k}{0}\,.
\end{split}
\end{flalign}
And, when injected in\eref{c1}, this is leading to
\begin{flalign}\label{c4}
\begin{split}
\partial_\tau\au{n}{0}=-\qcv\af{n}{0}-\sum_{k=0}^{n-1}\af{n-k}{0}\af{k}{0}=-\sum_{k=0}^n\af{n-k}{0}\af{k}{0}\,.
\end{split}
\end{flalign}
Ultimately, Lemma \ref{t3.1.9} together with the straightforward property $\eta_n\fa{k}=\eta_{n+k}\,$, $\,\,\forall (n,k)\in\Nz\times\Nz\,$, are leading to
\begin{flalign}\label{c5}
\begin{split}
\sum_{k=0}^n\af{n-k}{0}\af{k}{0}=\qcv^2\sum_{k=0}^n\fa{n-k+k}=(n+1)\fa{n}\qcv^2\,,
\end{split}
\end{flalign}
hereby completing the proof.
\end{proof}

\begin{lem}\label{tc2}
We also have, $\forall n\in\N\,$,
\begin{flalign}\label{c6}
\begin{split}
\au{n}{\alpha}+\au{n-1}{\alpha}=\int_{\I}\af{n-1}{\alpha}(\xi)\qcv(\xi)\,d\xi\,.
\end{split}
\end{flalign}
\end{lem}

\begin{proof}
The following computation, $\forall n\in\N\,$,
\begin{flalign}\label{c7}
\begin{split}
\au{n}{\alpha}+\au{n-1}{\alpha}=\ipr{\psi}{\Pi_{\I}\Big(\Ro^n+\Ro^{n-1}\Big)\rho\Po^\alpha\psi}=\ipr{\psi}{\Pi_{\I}\Big(\big(\rho-\id\big)\Ro^{n-1}+\Ro^{n-1}\Big)\rho\Po^\alpha\psi}=\ipr{\rho\psi}{\Pi_{\I}\Ro^{n-1}\rho\Po^\alpha\psi}\,,
\end{split}
\end{flalign}
ensures the result.
\end{proof}

\begin{lem}\label{tc3}
The AWFs satisfy, $\forall n\in\N\,$,
\begin{flalign}\label{c8}
\begin{split}
\chi_{n,0}''=\frac{\gamma^2}{\ud^2}\big(v_0+&\tau\big)\chi_{n,0}-\frac{\gamma\udd}{\ud^2}\Big((2n+1)\chi_{n,0}'+2(n+1)\chi_{n+1,0}'\Big) \\
&-\frac{\gamma^2\udd^2}{\ud^4}\Big((n^2+n)\chi_{n,0}+2(n+1)^2\chi_{n+1,0}+(n+1)(n+2)\chi_{n+2,0}\Big) \\
&-\frac{\gamma^2\udd}{\ud^3}(n+1)\mu_{0,0}\chi_{n+1,0}+\sum_{k=0}^n\bigg[\frac{-\gamma}{\ud}\Big((\mu_{n-k,0}'+\mu_{n-k-1,0}')+2(\mu_{n-k,1}+\mu_{n-k-1,1})\Big)\chi_{k,0} \\
&+\frac{\gamma^2}{\ud^2}\sum_{l=0}^k(\mu_{n-k,0}+\mu_{n-k-1,0})(\mu_{k-l,0}+\mu_{k-l-1,0})\chi_{l,0} \\
&-\frac{\gamma^2\udd}{\ud^3}\Big((\mu_{n-k,0}+\mu_{n-k-1,0})\big((n-k+1)\chi_{k,0}-(k+1)\chi_{k+1,0}\big)+(n+1)(\mu_{n-k+1,0}+\mu_{n-k,0})\chi_{k,0}\Big)\bigg]\,,
\end{split}
\end{flalign}
where $\au{k}{0}'\coloneqq\partial_\tau\au{k}{0}\,$, $\,\,\forall k\in\Nz\,$.
\end{lem}

\begin{proof}
For $\I=[\tau,\infty)$, Proposition \ref{t2.2.33} brings forth, $\forall n\in\Nz\,$,
\begin{flalign}\label{c9}
\begin{split}
\af{n}{\alpha}'=\af{n}{\alpha+1}-\frac{\gamma}{\dot{u}_0}\Bigg(\au{0}{\alpha}\af{n}{0}+\sum_{k=0}^{n-1}\big(\au{n-k}{\alpha}+\au{n-k-1}{\alpha}\big)\af{k}{0}\Bigg)-\frac{\gamma\ddot{u}_0}{\dot{u}_0^2}\bigg(n\af{n}{\alpha}+(n+1)\af{n+1}{\alpha}\bigg)\,.
\end{split}
\end{flalign}
Then we first set $\alpha=1$ in order to arrive at
\begin{flalign}\label{c10}
\begin{split}
\af{n}{2}=\af{n}{1}'+\frac{\gamma}{\dot{u}_0}\Bigg(\au{0}{1}\af{n}{0}+\sum_{k=0}^{n-1}\big(\au{n-k}{1}+\au{n-k-1}{1}\big)\af{k}{0}\Bigg)+\frac{\gamma\ddot{u}_0}{\dot{u}_0^2}\bigg(n\af{n}{1}+(n+1)\af{n+1}{1}\bigg)\,.
\end{split}
\end{flalign}
At last, equating this with Proposition \ref{t2.2.36} and noting that all of the $\af{k}{1}\,$, $\,\forall k\in\Nz\,$, in the resulting expression can be substituted with\eref{c9} for $\alpha=0$, we are led to\eref{c8}.
\end{proof}

\begin{prop}\label{tc4}
For any $n\in\N$, one may derive the following second-order integro-differential equation for the $\af{k}{0}\,,\,\,k\in\{0,...,n\}\subset\Nz\,$,
\begin{flalign}\label{c11}
\begin{split}
\af{n}{0}''=&\frac{\gamma}{\ud^2}\big(v_0+\tau\big)\af{n}{0}-\frac{\gamma\udd}{\ud^2}\Bigg((2n+1)\af{n}{0}'+2(n+1)\gao\af{n}{0}'+2(n+1)\bigg(\frac{\qcv'}{\psi}-\frac{\psi'\qcv}{\psi^2}\bigg)\af{n}{0}\Bigg) \\
&-\frac{\gamma^2\udd^2}{\ud^4}\Bigg((n^2+n)\af{n}{0}+2(n+1)^2\gao\af{n}{0}+(n+1)(n+2)\ga{2}\af{n}{0}\Bigg) \\
&-\frac{\gamma^2\udd}{\ud^3}(n+1)\gao\af{n}{0}\int_\tau^\infty\psi(\xi)\qcv(\xi)\,d\xi+\sum_{k=0}^n\Bigg[\frac{\gamma}{\ud}\Bigg(\bigg((n+1-k)\gao+(n-k)\bigg)\ga{n-1}\qcv^3 \\
&-2\qcv\int_\tau^\infty\af{n-k-1}{0}(\xi)\bigg(\qcv'(\xi)+\frac{\gamma}{\ud}\qcv(\xi)\int_\tau^\infty\psi(\zeta)\qcv(\zeta)\,d\zeta+\frac{\gamma\udd}{\ud^2}\af{1}{0}(\xi)\bigg)\,d\xi\Bigg) \\
&+\frac{\gamma^2}{\ud^2}\sum_{l=0}^k\ga{l}\qcv\bigg(\int_\tau^\infty\af{n-k-1}{0}(\xi)\qcv(\xi)\,d\xi\bigg)\bigg(\int_\tau^\infty\af{k-l-1}{0}(\xi)\qcv(\xi)\,d\xi\bigg) \\
&-\frac{\gamma^2\udd}{\ud^3}\Bigg(\bigg((n+1-k)\ga{k}\qcv-(k+1)\ga{k+1}\qcv\bigg)\int_\tau^\infty\af{n-k-1}{0}(\xi)\qcv(\xi)\,d\xi \\
&+(n+1)\ga{k}\qcv\int_\tau^\infty\af{n-k}{0}(\xi)\qcv(\xi)\,d\xi\Bigg)\Bigg]\,,
\end{split}
\end{flalign}
where we kept on employing Notation \ref{t3.1.16}.
\end{prop}

\begin{proof}
This is obtained by injecting Lemmas \ref{t3.1.9}, \ref{tc1} as well as \ref{tc2} into \ref{tc3}, explicitly writing $\fa{n}\coloneqq\ga{n}\,,\,\,\forall n\in\Nz\,$, instead of using Notation \ref{t3.1.14}.
\end{proof}

%\newpage

\bibliographystyle{ytamsalpha}
\bibliography{DraftRefs}

\providecommand{\bysame}{\leavevmode\hbox to3em{\hrulefill}\thinspace}
\providecommand{\MR}{\relax\ifhmode\unskip\space\fi MR }
\providecommand{\MRhref}[2]{%
  \href{http://www.ams.org/mathscinet-getitem?mr=#1}{#2}
}
\providecommand{\href}[2]{#2}
\providecommand{\doihref}[2]{\href{#1}{#2}}
\providecommand{\arxivfont}{\tt}
\begin{thebibliography}{DFGZJ93}

\bibitem[AA12]{18}
G.~Akemann and M.~R. Atkin, \emph{{Higher Order Analogues of Tracy-Widom Distributions via the Lax Method}}, \doihref{http://dx.doi.org/10.1088/1751-8113/46/1/015202}{J. Phys. A \textbf{46} (2013) 015202}, \href{http://arxiv.org/abs/1208.3645}{{\arxivfont arXiv:1208.3645 [math-ph]}}.

\bibitem[ABDF15]{12}
G.~Akemann, J.~Baik, and P.~Di~Francesco (eds.), \doihref{http://dx.doi.org/10.1093/oxfordhb/9780198744191.001.0001}{\emph{{The Oxford Handbook of Random Matrix Theory}}}, Oxford Handbooks in Mathematics, Oxford Univ. Press, 2015.

\bibitem[ACQ10]{29}
G.~Amir, I.~Corwin, and J.~Quastel, \emph{Probability distribution of the free energy of the continuum directed random polymer in 1 + 1 dimensions}, \href{http://dx.doi.org/10.1002/cpa.20347}{Commun. Pure Appl. Math. \textbf{64} (2010) 466–537}, \href{http://arxiv.org/abs/1003.0443}{{\arxivfont arXiv:1003.0443 [math.PR]}}.

\bibitem[AGZ11]{14}
G.~W. Anderson, A.~Guionnet, and O.~Zeitouni, \doihref{http://dx.doi.org/10.1017/CBO9780511801334}{\emph{{An Introduction to Random Matrices}}}, {Cambridge Studies in Advanced Mathematics}, vol. 118, Cambridge University Press, 2011.

\bibitem[BCT21]{36}
T.~Bothner, M.~Cafasso, and S.~Tarricone, \emph{{Momenta spacing distributions in anharmonic oscillators and the higher order finite temperature Airy kernel}}, \href{http://dx.doi.org/10.1214/21-AIHP1211}{Ann. Inst. H. Poincaré Probab. Statist. \textbf{58} (2022) 1505--1546}, \href{http://arxiv.org/abs/2101.03557}{{\arxivfont arXiv:2101.03557 [math-ph]}}.

\bibitem[BD01]{3}
A.~Borodin and P.~Deift, \emph{{Fredholm determinants, Jimbo-Miwa-Ueno $\tau$-functions, and representation theory}}, \doihref{http://dx.doi.org/10.1002/cpa.10042}{Comm. Pure Appl. Math. \textbf{55} (2002) 1160--1230}, \href{http://arxiv.org/abs/math-ph/0111007}{{\arxivfont arXiv:math-ph/0111007}}.

\bibitem[BDJ98]{30}
J.~Baik, P.~Deift, and K.~Johansson, \emph{On the distribution of the length of the longest increasing subsequence of random permutations}, \doihref{http://dx.doi.org/10.1090/s0894-0347-99-00307-0}{J. Amer. Math. Soc. \textbf{12} (1999) 1119–1178}, \href{http://arxiv.org/abs/math/9810105}{{\arxivfont arXiv:math/9810105 [math.CO]}}.

\bibitem[BDS16]{7}
J.~Baik, P.~Deift, and T.~Suidan, \emph{{Combinatorics and Random Matrix Theory}}, vol. 172, American Mathematical Soc., 2016. \url{https://bookstore.ams.org/gsm-172}.

\bibitem[BH16]{21}
E.~Br\'ezin and S.~Hikami, \doihref{http://dx.doi.org/10.1007/978-981-10-3316-2}{\emph{{Random Matrix Theory with an External Source}}}, SpringerBriefs in Mathematical Physics, vol.~19, Springer, 2016.

\bibitem[Bot22]{35}
T.~Bothner, \emph{{A Riemann-Hilbert approach to Fredholm determinants of Hankel composition operators: scalar-valued kernels}}, \href{https://arxiv.org/abs/2205.15007}{J. Funct. Anal. \textbf{285} (2023) 1--109}, \href{http://arxiv.org/abs/2205.15007}{{\arxivfont arXiv:2205.15007 [math-ph]}}.

\bibitem[CCG19]{20}
M.~Cafasso, T.~Claeys, and M.~Girotti, \emph{{Fredholm Determinant Solutions of the Painlevé II Hierarchy and Gap Probabilities of Determinantal Point Processes}}, \href{http://dx.doi.org/10.1093/imrn/rnz168}{Int. Math. Res. Not. \textbf{2021} (2019) 2437–2478}, \href{http://arxiv.org/abs/1902.05595}{{\arxivfont arXiv:1902.05595 [math-ph]}}.

\bibitem[CCR20]{37}
M.~{Cafasso}, T.~{Claeys}, and G.~{Ruzza}, \emph{{Airy Kernel Determinant Solutions to the KdV Equation and Integro-Differential Painlev{\'e} Equations}}, \doihref{http://dx.doi.org/10.1007/s00220-021-04108-9}{Commun. Math. Phys. \textbf{386} (2021) 1107--1153}, \href{http://arxiv.org/abs/2010.07723}{{\arxivfont arXiv:2010.07723 [math-ph]}}.

\bibitem[CEMR23]{25}
S.~Collier, L.~Eberhardt, B.~Muehlmann, and V.~A. Rodriguez, \emph{{The Virasoro minimal string}}, \doihref{http://dx.doi.org/10.21468/SciPostPhys.16.2.057}{SciPost Phys. \textbf{16} (2024) 057}, \href{http://arxiv.org/abs/2309.10846}{{\arxivfont arXiv:2309.10846 [hep-th]}}.

\bibitem[CET95]{33}
Y.~Chen, K.~J. Eriksen, and C.~A. Tracy, \emph{{Largest Eigenvalue distribution in the double scaling limit of matrix models: A Coulomb fluid approach}}, \doihref{http://dx.doi.org/10.1088/0305-4470/28/7/001}{J. Phys. A \textbf{28} (1995) L207--L212}, \href{http://arxiv.org/abs/hep-th/9502123}{{\arxivfont arXiv:hep-th/9502123}}.

\bibitem[CG21]{39}
T.~Claeys and G.~Glesner, \emph{Determinantal point processes conditioned on randomly incomplete configurations}, \href{http://dx.doi.org/10.1214/22-AIHP1311}{Ann. Inst. H. Poincaré Probab. Statist. \textbf{59} (2023) 2189--2219}, \href{http://arxiv.org/abs/2112.10642}{{\arxivfont arXiv:2112.10642 [math.PR]}}.

\bibitem[CKI09]{17}
T.~Claeys, I.~Krasovsky, and A.~Its, \emph{{Higher‐order analogues of the Tracy‐Widom distribution and the Painlevé II hierarchy}}, \href{http://dx.doi.org/10.1002/cpa.20284}{Commun. Pure Appl. Math. \textbf{63} (2009) 362–412}, \href{http://arxiv.org/abs/0901.2473}{{\arxivfont arXiv:0901.2473 [math-ph]}}.

\bibitem[CLDR10]{Calabrese:2010EPL}
P.~Calabrese, P.~Le~Doussal, and A.~Rosso, \emph{Free-energy distribution of the directed polymer at high temperature}, \doihref{http://dx.doi.org/10.1209/0295-5075/90/20002}{EPL \textbf{90} (2010) 20002}, \href{http://arxiv.org/abs/1002.4560}{{\arxivfont arXiv:1002.4560 [cond-mat.dis-nn]}}.

\bibitem[DFGZJ93]{15}
P.~Di~Francesco, P.~H. Ginsparg, and J.~Zinn-Justin, \emph{{2-D Gravity and random matrices}}, \doihref{http://dx.doi.org/10.1016/0370-1573(94)00084-G}{Phys. Rept. \textbf{254} (1995) 1--133}, \href{http://arxiv.org/abs/hep-th/9306153}{{\arxivfont arXiv:hep-th/9306153}}.

\bibitem[Dot10]{6}
V.~Dotsenko, \emph{{Bethe ansatz derivation of the Tracy-Widom distribution for one-dimensional directed polymers}}, \href{http://dx.doi.org/10.1209/0295-5075/90/20003}{EPL \textbf{90} (2010) 20003}, \href{http://arxiv.org/abs/1003.4899}{{\arxivfont arXiv:1003.4899 [cond-mat.dis-nn]}}.

\bibitem[EKR15]{24}
B.~Eynard, T.~Kimura, and S.~Ribault, \emph{{Random matrices}}, 10 2015. \href{http://arxiv.org/abs/1510.04430}{{\arxivfont arXiv:1510.04430 [math-ph]}}.

\bibitem[For93]{16}
P.~J. Forrester, \emph{{The spectrum edge of random matrix ensembles}}, \doihref{http://dx.doi.org/10.1016/0550-3213(93)90126-A}{Nucl. Phys. B \textbf{402} (1993) 709--728}.

\bibitem[For10]{13}
\bysame, \emph{Log-gases and random matrices}, Princeton University Press, 2010. \url{https://press.princeton.edu/books/hardcover/9780691128290/log-gases-and-random-matrices-lms-34}.

\bibitem[GS22]{40}
P.~{Ghosal} and G.~L.~F. {Silva}, \emph{{Universality for Multiplicative Statistics of Hermitian Random Matrices and the Integro-Differential Painlev{\'e} II Equation}}, \doihref{http://dx.doi.org/10.1007/s00220-022-04518-3}{Commun. Math. Phys. \textbf{397} (2023) 1237--1307}, \href{http://arxiv.org/abs/2201.12941}{{\arxivfont arXiv:2201.12941 [math-ph]}}.

\bibitem[HM80]{31}
S.~P. {Hastings} and J.~B. {McLeod}, \emph{{A boundary value problem associated with the second painlev{\'e} transcendent and the Korteweg-de Vries equation}}, \doihref{http://dx.doi.org/10.1007/BF00283254}{Arch. Rational Mech. Anal. \textbf{73} (1980) 31--51}.

\bibitem[IIKS90]{34}
A.~R. {Its}, A.~G. {Izergin}, V.~E. {Korepin}, and N.~A. {Slavnov}, \emph{{Differential Equations for Quantum Correlation Functions}}, \doihref{http://dx.doi.org/10.1142/S0217979290000504}{Int. J. Mod. Phys. B \textbf{4} (1990) 1003--1037}.

\bibitem[JMU81]{1}
M.~Jimbo, T.~Miwa, and K.~Ueno, \emph{{Monodromy preserving deformation of linear ordinary differential equations with rational coefficients}: {I. General theory and \ensuremath{\tau}-function}}, \doihref{http://dx.doi.org/10.1016/0167-2789(81)90013-0}{Physica D \textbf{2} (1981) 306--352}.

\bibitem[Joh19]{26}
C.~V. Johnson, \emph{{Nonperturbative Jackiw-Teitelboim gravity}}, \doihref{http://dx.doi.org/10.1103/PhysRevD.101.106023}{Phys. Rev. D \textbf{101} (2020) 106023}, \href{http://arxiv.org/abs/1912.03637}{{\arxivfont arXiv:1912.03637 [hep-th]}}.

\bibitem[Joh24]{27}
\bysame, \emph{{On the Random Matrix Model of the Virasoro Minimal String}}, 1 2024. \href{http://arxiv.org/abs/2401.06220}{{\arxivfont arXiv:2401.06220 [hep-th]}}.

\bibitem[Kra20]{10}
A.~Krajenbrink, \emph{{From Painlev\'e to Zakharov\textendash{}Shabat and beyond: Fredholm determinants and integro-differential hierarchies}}, \doihref{http://dx.doi.org/10.1088/1751-8121/abd078}{J. Phys. A \textbf{54} (2021) 035001}, \href{http://arxiv.org/abs/2008.01509}{{\arxivfont arXiv:2008.01509 [math-ph]}}.

\bibitem[KZ22]{2}
T.~Kimura and A.~Zahabi, \emph{{Universal cusp scaling in random partitions}}, \doihref{http://dx.doi.org/10.1007/s11005-024-01771-6}{Lett. Math. Phys. \textbf{114} (2024) 27}, \href{http://arxiv.org/abs/2208.07288}{{\arxivfont arXiv:2208.07288 [math-ph]}}.

\bibitem[LDMS18]{19}
P.~Le~Doussal, S.~N. Majumdar, and G.~Schehr, \emph{Multicritical edge statistics for the momenta of fermions in nonharmonic traps}, \href{http://dx.doi.org/10.1103/PhysRevLett.121.030603}{Phys. Rev. Lett. \textbf{121} (2018) }, \href{http://arxiv.org/abs/1802.06436}{{\arxivfont arXiv:1802.06436 [cond-mat.stat-mech]}}.

\bibitem[Meh04]{11}
M.~L. Mehta, \emph{{Random Matrices}}, Pure and Applied Mathematics, vol. 142, Elsevier, 2004.

\bibitem[MT22]{22}
T.~G. Mertens and G.~J. Turiaci, \emph{{Solvable models of quantum black holes: a review on Jackiw\textendash{}Teitelboim gravity}}, \doihref{http://dx.doi.org/10.1007/s41114-023-00046-1}{Living Rev. Rel. \textbf{26} (2023) 4}, \href{http://arxiv.org/abs/2210.10846}{{\arxivfont arXiv:2210.10846 [hep-th]}}.

\bibitem[{Ruz}24]{38}
G.~{Ruzza}, \emph{{Bessel kernel determinants and integrable equations}}, \doihref{http://dx.doi.org/10.48550/arXiv.2401.11213}{arXiv e-prints (2024) }, \href{http://arxiv.org/abs/2401.11213}{{\arxivfont arXiv:2401.11213 [math-ph]}}.

\bibitem[{Sas}16]{32}
T.~{Sasamoto}, \emph{{The 1D Kardar-Parisi-Zhang equation: Height distribution and universality}}, \doihref{http://dx.doi.org/10.1093/ptep/ptw002}{Prog. Theor. Exp. Phys. \textbf{2016} (2016) 022A01}.

\bibitem[Sos00]{9}
A.~Soshnikov, \emph{Determinantal random point fields}, \href{http://dx.doi.org/10.1070/RM2000v055n05ABEH000321}{Russ. Math. Surv. \textbf{55} (2000) 923–975}, \href{http://arxiv.org/abs/math/0002099}{{\arxivfont arXiv:math/0002099 [math.PR]}}.

\bibitem[SS10]{28}
T.~Sasamoto and H.~Spohn, \emph{{One-Dimensional Kardar-Parisi-Zhang Equation: An Exact Solution and its Universality}}, \href{http://dx.doi.org/10.1103/PhysRevLett.104.230602}{Phys. Rev. Lett. \textbf{104} (2010) }, \href{http://arxiv.org/abs/1002.1883}{{\arxivfont arXiv:1002.1883 [cond-mat.stat-mech]}}.

\bibitem[SSS19]{23}
P.~Saad, S.~H. Shenker, and D.~Stanford, \emph{{JT gravity as a matrix integral}}, 3 2019. \href{http://arxiv.org/abs/1903.11115}{{\arxivfont arXiv:1903.11115 [hep-th]}}.

\bibitem[TW92]{4}
C.~A. Tracy and H.~Widom, \emph{{Level spacing distributions and the Airy kernel}}, \doihref{http://dx.doi.org/10.1007/BF02100489}{Commun. Math. Phys. \textbf{159} (1994) 151--174}, \href{http://arxiv.org/abs/hep-th/9211141}{{\arxivfont arXiv:hep-th/9211141}}.

\bibitem[TW95]{5}
\bysame, \emph{{On orthogonal and symplectic matrix ensembles}}, \doihref{http://dx.doi.org/10.1007/BF02099545}{Commun. Math. Phys. \textbf{177} (1996) 727--754}, \href{http://arxiv.org/abs/solv-int/9509007}{{\arxivfont arXiv:solv-int/9509007}}.

\end{thebibliography}

\end{document}